\documentclass[11pt]{article}

\usepackage{colortbl}
\usepackage{booktabs}\usepackage[table]{xcolor}
\usepackage{mathtools} \usepackage{dsfont}
\usepackage{algpseudocode}
\usepackage{mathrsfs,verbatim,rotate}
\usepackage{graphicx}
\usepackage{amssymb,amsfonts, amsmath, amsthm}
\usepackage{graphics}
\usepackage{color}
\usepackage{float}
\usepackage{lscape}
\usepackage{setspace}
\usepackage{subfigure}
\usepackage{natbib}
\usepackage{multirow}
\usepackage{enumerate}
\usepackage{array}

\definecolor{darkgreen}{rgb}{0.0, 0.5, 0.0}
\usepackage[hidelinks, colorlinks = true, linkcolor = darkgreen, urlcolor = blue, citecolor = blue]{hyperref}
\usepackage{prodint}
\usepackage{thmtools}
\usepackage{changepage}
\usepackage[shortlabels]{enumitem}

\usepackage[mathscr]{euscript}
\usepackage[bottom,flushmargin,hang,multiple]{footmisc}
\usepackage[letterpaper, margin=1in]{geometry}

\definecolor{lightgray}{gray}{0.9}

\declaretheorem[name=Theorem]{thm}

\declaretheorem[name=Lemma]{lemma}

\newcommand{\E}{\mathbb{E}}

\newcommand{\miid}{\,|\,}

\makeatletter
\def\namedlabel#1#2{\begingroup
	#2%
	\def\@currentlabel{#2}%
	\phantomsection\label{#1}\endgroup
}
\makeatother

\newcolumntype{x}[1]{%
	>{\centering\hspace{0pt}}p{#1}}%

\AfterEndEnvironment{thm}{\noindent\ignorespaces}
\usepackage{mathbbol}
\usepackage{graphics,amsmath,pstricks,enumitem,float}
\usepackage{amssymb}
\usepackage{amsbsy,amsmath,amsthm,amsfonts, amssymb}
\usepackage{graphicx, rotate, array}  
\usepackage{geometry,natbib,setspace,multirow}
\usepackage{algorithm}
\usepackage{mathalpha}

\allowdisplaybreaks

\newcounter{algsubstate}

\title{Debiased machine learning for counterfactual survival functionals\\ based on left-truncated right-censored data}
\author{Eric R. Morenz{\small $^{1*}$}, Charles J. Wolock{\small $^{2*}$}, Marco Carone{\small $^{3}$}\\[1em]
	\small $1$: Canada's Drug Agency\\
	\small $2$: Department of Biostatistics, Epidemiology and Informatics, University of Pennsylvania\\
	\small $3$: Department of Biostatistics, University of Washington\\[0.5em]
	\small $*$: Contributed equally to this work.
}
\date{}
\begin{document}
	
	\maketitle
	
	\vspace{1in}
	\begin{abstract}
		Learning causal effects of a binary exposure on time-to-event endpoints can be challenging because survival times may be partially observed due to censoring and systematically biased due to truncation. In this work, we present debiased machine learning-based nonparametric estimators of the joint distribution of a counterfactual survival time and baseline covariates for use when the observed data are subject to covariate-dependent left truncation and right censoring and when baseline covariates suffice to deconfound the relationship between exposure and survival time. Our inferential procedures explicitly allow the integration of flexible machine learning tools for nuisance estimation, and enjoy certain robustness properties. The approach we propose can be directly used to make pointwise or uniform inference on smooth summaries of the joint counterfactual survival time and covariate distribution, and can be valuable even in the absence of interventions, when summaries of a marginal survival distribution are of interest. We showcase how our procedures can be used to learn a variety of inferential targets and illustrate their performance in simulation studies.
	\end{abstract}
	
	\newpage
	\doublespacing
	
	\section{Introduction}
	\label{s:intro}
	In biomedical studies, the outcome of interest is often the time elapsed between an initiating event  and a terminating event. For example, investigators may wish to study the time from some exposure or treatment (e.g., administration of vaccine) until a particular clinical event (e.g., onset of symptomatic disease). In particular, they may be interested in determining the effect of a treatment on the event time. Even in the context of a randomized trial, in which the design ensures that the relationship between treatment and event time is unconfounded, the analysis of time-to-event data remains challenging because event times are typically only partially observed in some study participants. Indeed, some participants may exit the study during the course of follow-up, or may not yet have experienced the event of interest by the end of the study, in which case their event times are right-censored. Right censoring complicates the identification of the time-to-event distribution --- notably, required assumptions about the censoring mechanism may fail to hold even in randomized trials --- and ensuing procedures for assumption-lean statistical inference are also much more involved. The problems that arise due to incomplete observation of terminating events are compounded when the study does not include randomization, in which case appropriate deconfounding, whenever possible, must also be incorporated into statistical procedures.
	
	In many observational studies, in addition to right censoring, the available data are subject to left truncation, wherein only participants for whom the event time is larger than a corresponding truncation time can be recruited into the study. This may occur, for example, due to delayed entry into a prospective study or to the use of a cross-sectional sampling scheme. Unlike censoring, which results in partially observed data but has no bearing on who may be sampled, truncation implies a restriction on the sampling mechanism, and usually renders the sampling population biased relative to the target population. Indeed, truncation induces systematic selection bias into the study design, with an over-representation of participants with a longer event time. Failure to account for left truncation can result in severely biased inferences and misleading scientific conclusions --- \cite{wolfson2001nejm} provides a compelling example of such bias in the medical literature.
	
	While the field of survival analysis is mature, with many decades of rigorous methodological developments pertaining to the analysis of left-truncated right-censored data, most existing works have relied heavily either on semiparametric and parametric modeling assumptions, or on strong uninformativeness assumptions about the censoring and truncation mechanisms. Furthermore, while there has been a growing literature at the intersection of survival analysis and causal inference, the focus has been almost exclusively on data subject to right censoring without truncation --- see, e.g., \cite{westling2021inference} for a sampling of such existing methods. In this work, we contribute to addressing this gap by developing novel nonparametric statistical methods for estimating causal effect summaries with left-truncated right-censored data.
	
	In the developments below, we propose debiased machine learning techniques for nonparametric inference on smooth summaries of a counterfactual time-to-event distribution using left-truncated right-censored data. The class of summaries we consider is broad and includes, in particular, commonly reported estimands, such as survival probabilities, restricted means, and quantiles, as well as more complex functionals. Notably, the methods we develop allow informative censoring and truncation insofar as can be explained by recorded covariates --- in other words, the censoring and truncation mechanisms may be covariate-dependent. They also allow the use of flexible learning algorithms for estimating involved nuisance functions without compromising the calibration of resulting statistical inferences. This is desirable since the use of such algorithms can mitigate the risk of systematic bias possibly resulting from inconsistent estimation of such nuisance functions.
	
	We note that \cite{wang2022doubly} recently made important advances in the development of debiased machine learning methods for use with left-truncated data. However, their work focuses on inference for a marginal (rather than counterfactual) survival function, and their procedures neither facilitate flexible estimation of the censoring mechanism nor generally achieve the efficiency bound in the presence of right censoring. As such, our work extends theirs by both including consideration of summaries of a counterfactual time-to-event distribution and restoring efficiency even in the presence of both left truncation and right censoring. We also note that our work can be seen as a natural generalization of the recent work of \cite{westling2021inference}, which develops flexible techniques for nonparametric efficient inference on a counterfactual survival function using right-censored data without truncation. While traditional risk set-based methods for the analysis of right-censored data can often be effortlessly extended to the analysis of left-truncated right-censored data, this is not necessarily the case for other methods, including those based on influence functions, as in this work.
	
	This article is organized as follows. In Section~\ref{s:id}, we define the survival integral, our estimand of interest, and discuss its identification in contexts in which the time-to-event random variable is observed subject to possibly covariate-dependent left truncation and right censoring. In Section~\ref{s:param}, we derive a linearization of the survival integral parameter viewed as a functional of the observed data distribution. In Section~\ref{s:est}, we use this linearization to construct two distinct cross-fitted inferential procedures that explicitly allow the incorporation of machine learning methods. In Section~\ref{s:inf}, we establish certain large-sample properties of the proposed procedures, including both pointwise and uniform distributional results, and extend these results to a larger class of smooth survival functionals.  In Section~\ref{s:ex}, we discuss several analytic examples of our general results, whereas in Section~\ref{s:sim}, we present results from numerical experiments to illustrate the operating characteristics of our procedures. We conclude with final remarks in Section~\ref{s:concl}. All technical proofs are provided in Part \hyperref[parta]{A} of the Appendix.
	
	\section{Statistical setup and identification}
	\label{s:id}
	\subsection{Notation and examples}
	
	The ideal data unit is $X:=(T,C,W,A,Z)\sim P_{X,0}$, where $Z\in\mathscr{Z}\subseteq\mathbb{R}^p$ denotes a vector of baseline covariates, $A\in\{0,1\}$ is a binary exposure level indicator, $W\in[0,\infty)$ and $C\in(0,\infty)$ are the truncation and censoring times, respectively, and the event time (or survival time) is $T\in (0,\infty)$. Here, $P_{X,0}$ denotes the true (unknown) distribution of $X$ in the \emph{target} population. For a fixed $a_0\in\{0,1\}$ and known kernel function $\varphi:\mathbb{R}\times\mathbb{R}^p\rightarrow\mathbb{R}$, we begin by studying inference on the survival integral \begin{align}
		\upsilon_0:=\iint \varphi(t,z)\,F_{X,0}(dt\,|\,a_0,z)\,H_{X,0}(dz)\ ,\label{eq:target}
	\end{align} where we define $F_{X,0}(t\,|\,a,z):=P_{X,0}\left(T\leq t\,|\, A=a,Z=z\right)$ and $H_{X,0}(z):=P_{X,0}\left(Z\leq z\right)$ pointwise. We note here that the ideal data unit may have been taken to simply be $(T,A,Z)$ since the estimand of interest depends only on the conditional distribution of $T$ given $(A,Z)$ and on the marginal distribution of $Z$, and if neither truncation nor censoring act on the data unit, the value of $(C,W)$ is irrelevant. Nevertheless, for notational convenience in developments below, we define $X$ to also include $(C,W)$.
	
	Survival integrals encompass several estimands of interest at the intersection of survival analysis and causal inference. Under typical causal conditions, including that, within each stratum of $Z$, the counterfactual event time $T(a_0)$ corresponding to the intervention that sets $A=a_0$ is independent of $A$, and that $A=a_0$ occurs with positive probability, $\upsilon_0$ identifies the counterfactual mean value $\mathbb{E}_0\left[\varphi(T(a_0),Z)\right]$ computed under the joint distribution $\mathbb{P}_0$ of $(T(a_0),Z)$. Various choices of $\varphi$ yield different causal estimands of practical interest. As a special case, by considering the exposure $A$ to be degenerate at $a_0$, the survival integral trivially corresponds instead to moments of the joint distribution of $(T,Z)$, estimands that arise in traditional survival analyses. Specific examples of estimands that motivate our work and are later discussed include:
	\begin{enumerate}[(1)]
		\item the marginal survival probability $P_{X,0}\left(T>t\right)$;\label{ch1ex1}\vspace{-.075in}
		\item the Brier score $E_{X,0}\left\{I(T\geq \tau)-b(Z)\right\}^2$ of a given function $b:\mathscr{Z}\rightarrow [0,1]$ for predicting survival at time $\tau$ \citep{brier1950verification, gerds2006consistent};\label{ch1ex2}\vspace{-.075in}
		\item the counterfactual survival probability $\mathbb{P}_0\left\{T(a_0)>t\right\}$ \citep{westling2021inference}.\label{ch1ex3}
	\end{enumerate}
	Later, we build upon our results on survival integrals to develop inferential methods for nonlinear survival functionals. This extension allows us to tackle many more estimands of interest. Examples of such estimands that we study in greater detail include:
	\begin{enumerate}[(1)]\setcounter{enumi}{3}
		\item the median counterfactual event time \citep{diaz2017efficient};\label{ch1ex4}\vspace{-.075in}
		\item a model-agnostic measure of dependence of $T(a_0)$ on $a_0$ \citep{vansteelandt2022assumption}.\label{ch1ex5}
	\end{enumerate}
	Example 1 is the primary target of inference in classical survival analysis, although here we wish to allow possibly covariate-dependent censoring and truncation. Example 2 arises in the evaluation of prediction models in survival analysis, and emphasizes the value of allowing the kernel value $\varphi(t,z)$ to depend on both $t$ and $z$. Examples 3 and  4 are commonly reported summaries of the counterfactual survival distribution. Example 5 is a novel parsimonious measure of the causal effect of $A$ on $T$ inspired by recent work on assumption-lean Cox regression \citep{vansteelandt2022assumption}.
	
	We refer to $X$ as an ideal data unit because statistical inference would be straightforward if it were directly observed. However, in practice, this is rarely the case. In many prospective cohort studies, the sampling distribution of $X$ is a systematically biased version of its target distribution due to left truncation resulting, for example, from delayed entry or cross-sectional sampling. Additionally, only a coarsened version $O$ of $X$ is observed due to right censoring. Thus, in order to derive procedures for statistical inference, we must first represent $\upsilon_0$, which is explicitly defined as a summary of the \emph{ideal data} distribution $P_{X,0}$, as a summary of the \emph{observed data} distribution $P_0$, thereby establishing identification under suitable conditions. We refer to the distribution $P_0$ as defining the \emph{observable} population.

	\subsection{Identification}
	\label{ss:id}
	
	The observed data structure is $O:=(Y,\Delta, W,A,Z)\sim P_0$ with $Y:=\min(T,C)$ and $\Delta:=I(T\leq C)$, and results from left truncation and right censoring of the ideal data unit $X$. Specifically, only individuals with $Y\geq W$ can be sampled --- those who are neither censored nor experience the terminating event before possible recruitment into the study --- and the event time $T$ is subject to right censoring by $C$.  We assume that $C\geq W$ with $P_{X,0}-$probability one, since we are interested primarily in settings in which censoring is a study-induced nuisance and only operates on individuals who can possibly be recruited into the study. As such, the sampling conditions $Y\geq W$ and $T\geq W$ are equivalent. The observed data distribution $P_0$ is obtained from the target population distribution $P_{X,0}$ through the relationship 
	\begin{align*}
		&P_0\left(dy,d\delta,dw,da,dz\right)\\
		&= \frac{I(y\geq w)}{P_{X,0}\left(T\geq W\right)}\left\{\delta \int_{c\geq y}P_{X,0}\left(dy,dc,dw,da,dz\right)+(1-\delta)\int_{t\geq y}P_{X,0}\left(dt,dy,dw,da,dz\right)\right\}\nu(d\delta)
	\end{align*} for $\delta\in\{0,1\}$, where $\nu$ is the counting measure on $\{0,1\}$. As indicated above, even in the absence of censoring, the sampling distribution $P_0$ of the observed data unit does not coincide with the target distribution $P_{X,0}$ because individuals with $T<W$ are systematically excluded. Individuals with larger values of $T$ are therefore over-represented in the observable population relative to the target population. Throughout this article, the observed data consist of $n$ independent draws $O_1,O_2,\ldots,O_n$ from $P_0$.
	
	We now consider the problem of recovering $P_{X,0}$ from $P_0$ on relevant portions of its support, some of which may not be fully recoverable. For example, right censoring often precludes the identification of the right tail of the time-to-event distribution. Nevertheless, $\upsilon_0$ may still be identified. To formalize these issues, we first define for a generic random variable $B$ the lower and upper support bounds\begin{align*}
		\underline{\tau}_B(a,z)\ &:=\ \sup\{u:P_{X,0}\left(B\geq u\,|\,A=a,Z=z\right)=1\}\,;\\
		\overline{\tau}_B(a,z)\ &:=\ \sup\{u:P_{X,0}\left(B\geq u\,|\,A=a,Z=z\right)>0\}\,.
	\end{align*} We also denote by $\overline{\tau}_C(a,z):=\sup\{u:P_{X,0}\left(C\geq u\,|\,A=a,Z=z, T\geq W\right)>0\}$ the upper bound for the support of the censoring distribution, and by $\pi_{X,0}(z):=P_{X,0}(A=a_0\,|\,Z=z)$ the propensity score for each $z\in\mathscr{Z}$. We make the following support recovery conditions for identifiability:
	\begin{enumerate}[({A}1)]
		\item \label{ass:support} for $a\in\{0,1\}$ and $P_{X,0}-$almost every value $z\in\mathscr{Z}$, it holds that:\vspace{-.1in}
		\begin{enumerate}[(i)]\itemsep-0.05in
			\item $\underline{\tau}_W(a,z)\leq \underline{\tau}_T(a,z)\,;$
			\item $\overline{\tau}_W(a,z) + \alpha \leq \min\{\overline{\tau}_T(a,z),\overline{\tau}_C(a,z)\}$ for some $\alpha > 0$\,;
			\item  $t\mapsto \varphi(t,z)\text{ is constant for }t\geq \overline{\tau}_C(a,z)\,;\vspace{.0in}$
		\end{enumerate}
		\item for $P_{X,0}-$almost every value $z\in\mathscr{Z}$, it holds that $\pi_{X,0}(z)>0$.\label{ass:positivity}
	\end{enumerate}
	These conditions can be interpreted heuristically as follows. First, if $\underline{\tau}_W(a,z)> \underline{\tau}_T(a,z)$, then individuals with exposure $A=a$, covariate vector $Z=z$ and event time $T=t$ such that $\underline{\tau}_T(a,z)\leq t<\underline{\tau}_W(a,z)$ are systematically excluded in the observable population. As such, the target conditional time-to-event distribution function $F_{X,0}(t\miid a,z)$ cannot be identified for any $t>0$, and  neither can $\upsilon_0$ if the set of such values has positive $P_{X,0}-$probability. Second, if $\min\{\overline{\tau}_T(a,z),\overline{\tau}_C(a,z)\} < \overline{\tau}_W(a,z)$, then individuals with exposure level $A=a$, covariate vector $Z=z$ and truncation time $W=w$ such that $\min\{\overline{\tau}_T(a,z),\overline{\tau}_C(a,z)\}\leq w<\overline{\tau}_W(a,z)$ are systematically excluded in the observable population. As such, the right tail of the target conditional truncation distribution function cannot be identified. This is problematic because, as we will see below, identification of the marginal covariate distribution $H_{X,0}$ --- and thus of $\upsilon_0$ --- hinges on that of the conditional truncation distribution. Third, for any $z\in\mathscr{Z}$, $F_{X,0}(t\,|\,a,z)$ can only be identified up to $\overline{\tau}_C(a,z)$ since values of $T$ above $\overline{\tau}_C(a,z)$ can never be observed, and so, unless $t\mapsto\varphi(t,z)$ is constant for $t>\overline{\tau}_C(a,z)$ and $P_{X,0}-$almost every $z\in\mathscr{Z}$, $\upsilon_0$ also cannot typically be identified. These facts motivate the need for condition \ref{ass:support}. Finally, it is necessary that $\pi_{X,0}(z)>0$ in order to be able to learn $F_{X,0}(t\,|\,a,z)$ without relying on extrapolating assumptions, since otherwise no inference could ever be made from the subpopulation defined by $(A,Z)=(a,z)$; this motivates condition \ref{ass:positivity}.
	
	Beyond support recovery conditions, identification hinges fundamentally on the vector $Z$ of baseline covariates being sufficiently rich to account for any dependence between $T$ and $(C,W)$. Specifically, we introduce the following additional conditions on the censoring and truncation mechanisms:
	\begin{enumerate}[({B}1)]
		\item $T$ and $W$ are independent given $(A,Z)$ $P_{X,0}-$almost surely;\label{ass:truncation}\vspace{-.075in}
		\item $T$ and $C$ are independent given $(W,A,Z)$ and $T\geq W$ $P_{X,0}-$almost surely.\label{ass:censoring}
	\end{enumerate}
	We note that distributional constraints on $C$ are only imposed in the observable population, that is, the subpopulation of individuals for whom $T\geq W$. In fact, $C$ need not even be defined for individuals with $T<W$ since censoring only ever affects those with $T\geq W$. Under conditions \ref{ass:support}--\ref{ass:positivity} and \ref{ass:truncation}--\ref{ass:censoring}, $P_{X,0}$ may be expressed in terms of $P_0$. For $0<t<\min\{\overline{\tau}_T(z),\overline{\tau}_C(z)\}$, the target conditional distribution function $F_{X,0}(t\,|\,a_0,z)$ is identified via conditional product-integration~\citep{gill1990survey} by
	\begin{align*}
		\tilde F_{0}(t\,|\,a_0,z) := 1-\Prodi_{u\leq t}\left\{1-\frac{F_{1,0}(du\,|\, a_0,z)}{R_0(u\,|\,a_0,z)}\right\},
	\end{align*}where $F_{1,0}(u\,|\,a,z):=P_0\left(Y\leq y,\Delta=1\,|\,A=a,Z=z\right)$ is an observable conditional subdistribution function and $R_0(u\miid a,z) := P_0\left(Y\geq u\geq W\,|\,A=a,Z=z\right)$ is an observable conditional at-risk probability. As indicated above, the target conditional truncation distribution function, defined pointwise as $G_{X,0}(w\,|\,a,z):=P_{X,0}\left(W\leq w\,|\,A=a,Z=z\right)$, is needed to recover the target covariate distribution. For any $w\geq 0$, it can be expressed as $G_{X,0}(w\miid a,z) \propto_w \int_{u\leq w}S_{X,0}(u\miid a,z)^{-1}G_0(du\miid a,z)$, where we denote the observable conditional truncation distribution function by $G_0(w\miid a,z) := P_0\left(W\le w\miid A=a, Z=z\right)$, the survival function corresponding to $F_{X,0}$ by $S_{X,0}:=1-F_{X,0}$, and $\propto_w$ refers to proportionality in $w$ for fixed $a$ and $z$. The identification of $F_{X,0}$ over $[0,\min\{\overline{\tau}_T(a,z),\overline{\tau}_C(a,z)\})$ --- and thus, by condition \ref{ass:support}, over its subset $[0,\overline{\tau}_W(a,z))$ --- then implies that of $G_{X,0}$. The target exposure-covariate distribution function $J_{X,0}(a,z):=P_{X,0}\left(A\leq a,Z\leq z\right)$ can be expressed as a reweighted version of its observable counterpart; indeed, we have that \begin{align}
		J_{X,0}(da,dz)\ &\propto_{a,z}\ \frac{J_0(da,dz)}{\int S_{X,0}(u\miid a,z)\,G_{X,0}(du\miid a,z)}\notag\\
		&\propto_{a,z}\ J_0(da,dz)\int \frac{G_0(du\,|\,a,z)}{1-\tilde F_0(u\,|\,a,z)}\label{eq:truncation} 
	\end{align} with $J_0(a,z):=P_0\left(A\leq a,Z\leq z\right)$ denoting the observable exposure-covariate distribution function; here, $\propto_{a,z}$ refers to proportionality in $(a,z)$. In particular, \eqref{eq:truncation} implies an identification $\tilde{H}_0$ of the target covariate distribution $H_{X,0}$ using that $H_{X,0}(dz)= \int_a J_{X,0}(da,dz)$ for each $z$. These expressions suffice to identify $\upsilon_0$ as the summary $\iint \varphi(t,z)\,\tilde{F}_{0}(dt\,|\,a_0,z)\,\tilde{H}_{0}(dz)$ of the observed data distribution $P_0$. Rather than focusing on the ideal data distribution of the censoring random variable, which as discussed earlier need not even be defined, we note that the observable conditional censoring survival function $Q_0$, defined pointwise as $Q_0(c\,|\,w,a,z):=P_0\left(C>c\,|\,W=w,A=A,Z=z\right)$ can be identified using conditional product-integration as used in $\tilde{F}_0$ but without truncation and reversing the value of $\Delta$ to $1-\Delta$. Additional details on these identification results are provided in Part \hyperref[partb]{B} of the Appendix.

	\section{Study of the target parameter}
	\label{s:param}
	
	The above identification formulas motivate us to study the observed data parameter 
	\begin{equation}\label{eq:id}
		\Psi:P\mapsto\iint \varphi(t,z)\,\tilde{F}_P(dt\,|\,a_0,z)\,\tilde{H}_P(dz)\ .
	\end{equation} Here, $\tilde{F}_P$ is defined pointwise as \begin{equation*}
		\tilde{F}_P(t\,|\,a_0,z):=1-\Prodi_{u\leq t}\left\{1-\frac{F_{1,P}(du\,|\,a_0,z)}{R_P(u\,|\,a_0,z)}\right\},
	\end{equation*}where $F_{1,P}$ and $R_P$ are defined as $F_{1,P}(u\,|\,a,z):=P\left(Y\leq u,\Delta=1\,|\,A=a,Z=z\right)$ and  $R_P(u\,|\,a,z):=P\left(Y\geq u\geq W\,|\,A=a,Z=z\right)$ pointwise, respectively. Further, $\tilde{H}_P$ is defined pointwise as $\tilde{H}_P(dz):=\int_a \bar{\gamma}_P(a,z)\,J_P(da,dz)$ with $J_P(a,z):=P\left(A\leq a,Z\leq z\right)$, $\gamma_P(a,z):=\int \tilde{S}_P(w\,|\,a,z)^{-1}G_P(dw\,|\,a,z)$, $\gamma_P:=\iint \gamma_P(a,z)\,J_P(da,dz)$ and $\bar{\gamma}_P(a,z):=\gamma_P(a,z)/\gamma_P$, where $\tilde{S}(t\,|\,a,w):=1-\tilde{F}(t\,|\,a,w)$ is the survival function corresponding to $\tilde{F}_P$ and $G_P(w\,|\,a,z):=P\left(W\leq w\,|\,A=a,Z=z\right)$ is the observable conditional truncation distribution function. Under conditions \ref{ass:support}--\ref{ass:positivity} and \ref{ass:truncation}--\ref{ass:censoring}, the survival integral $\upsilon_0$ is identified by $\psi_0:=\Psi(P_0)$. Thus, in the remainder of this article, we focus on developing inferential methods for $\psi_0$ and related estimands.

	We wish to employ flexible learning strategies to avoid unnecessarily strong modeling assumptions on the data-generating mechanism $P_0$. As such, in order to carry out valid nonparametric efficient inference, we develop debiased machine learning methods for this problem. As a first step, we derive a linearization of the parameter mapping $\Psi$ around $P_0$ based on the nonparametric efficient influence function of $\Psi$ at $P=P_0$~\citep{pfanzagl1985contributions}. This linearization is critical for guiding the construction of our estimation procedure and elucidating the conditions under which this procedure has desirable statistical properties.

	Before tackling the problem in its generality, it is instructive to first examine the simpler setting in which the support $\mathscr{Z}$ of the covariate vector is finite. In such case, for any fixed $z_0\in\mathscr{Z}$, the inner integral $\int \varphi(t,z)\,F_0(dt\,|\,a_0,z_0)$  can be estimated nonparametrically under the conditions we have introduced so far using the stratum-specific Kaplan-Meier integral $\int \varphi(t,z)\,F_n(dt\,|\,a_0,z_0)$, where $F_n(t\,|\,a_0,z)$ denotes the Kaplan-Meier estimator of $F_0(t\,|\,a_0,z_0)$ computed using only data from stratum $(A,Z)=(a_0,z_0)$. Under certain regularity conditions, this stratum-specific Kaplan-Meier integral can be shown to be regular and asymptotically linear~\citep{reid1981censored, stute1994bias} with influence function given by \[(z,a,w,\delta,y)\mapsto \frac{I(a=a_0,z=z_0)}{P\left(A=a_0,Z=z_0\right)}\phi_{\text{KM},P}(L_{P,\varphi})(z,a,w,\delta,y)\ ,\]where for any $P$ we denote by $L_{P,\varphi}:(y,a,z)\mapsto\int_y^\infty \tilde S_{P}(u\miid a,z)\varphi(du,z)$ and we define pointwise, for any function $m:\mathscr{Y}\times\{0,1\}\times\mathscr{Z}\rightarrow\mathbb{R}$, \[\phi_{\text{KM},P}(m)(z,a,w,\delta,y):=-\frac{\delta m(y,a,z)}{R_P(y\miid a,z)} + \int I_{[w,y]}\,(u)\frac{m(u,a,z)}{R_P(u\miid a,z)}\,\tilde{\Lambda}_P(du\miid a,z)\ \, \]with $\tilde{\Lambda}_P(t\,|\,a,z):=\int_{u\leq t}\tilde{S}_P(u\,|\,a,z)^{-1}\tilde{F}_P(du\,|\,a,z)$ denoting the cumulative hazard function corresponding to $\tilde{F}_P$. In our constructions and theoretical results below, the function $\phi_{\text{KM},P}$ appears prominently, as we will now see.
	
	Our linearization results involve additional notation that we now introduce. We define the observable survival regression $\mu_{P,\varphi}(z) := \int \varphi(t,z)\,\tilde{F}_{P}(dt\miid a_0,z)$, the observable propensity score $\pi_P(z) := P\left(A=a_0\miid Z=z\right)$, and the partial truncation weight function \[\gamma_{P,\natural}(y,a,z) := \int I_{[y,\infty)}(u)\, \tilde S_P(u\miid a,z)^{-1}G_P(du\miid a,z)\ .\] As we establish in the following theorem, the nonparametric linearization of $\Psi$ hinges critically on the nonparametric efficient influence function of $\Psi$ at $P$, which can be written as $\phi_{P}:=\phi_{1,P}+\phi_{2,P}$ with
	\begin{align*}
		\phi_{1,P}&:(z,a,w,y,\delta)\mapsto\frac{I(a=a_0)}{\pi_P(z)}\bar{\gamma}_P(z)\phi_{\text{KM},P}(L_{P,\varphi})(z,a,w,y,\delta)\\
		\phi_{2,P}&:(z,a,w,y,\delta)\mapsto\frac{\mu_{P,\varphi}(z) - \Psi(P)}{\gamma_P}\left\{\frac{1}{\tilde S_{P}(w\miid a,z)} - \phi_{\text{KM},P}(\gamma_{P,\natural})(z, a,w,y,\delta)\right\}\,,
	\end{align*}
	and where $\bar\gamma_P(z) :=\sum_a \bar\gamma_P(a,z)\pi_P(a\miid z)$.
	The nonparametric linearization of $\Psi(P)$ around $P=P_0$ involves a second-order remainder term that can be written as $R(P,P_0):=R_1(P,P_0)+R_2(P,P_0)+R_3(P,P_0) + R_4(P,P_0) $ with
	
	\begin{align*}
		R_1(P,P_0) :=&\, \int L_{P,\varphi}(y,a,z)\left\{\frac{\pi_0(z)\bar\gamma_P(z)\nu_P(y,a,z)}{\pi_P(z)\bar\gamma_0(z)\nu_0(y,a,z)} - 1\right\}\left(\frac{\tilde S_0}{\tilde S_P}-1\right)(dy\miid a_0,z)\,\tilde H_0(dz)\,;\\
		R_2(P,P_0) :=&\, \int\xi_P(z)  \gamma_{P,\natural}(y,a,z)\left\{1- \frac{\nu_P(y,a,z)}{\nu_0(y,a,z)}\right\}\left(\frac{\tilde S_0}{\tilde S_P}-1\right)(dy\miid a,z)\,\tilde H_0(dz)\,;\\
		R_3(P,P_0) :=&\, \int\xi_P(z)  \left\{\frac{1}{\tilde S_0(w\miid a,z)} - \frac{1}{\tilde S_P(w\miid a,z)}\right\}\left(G_P-G_0\right)(dw\miid a,z)\,J_0(da,dz)\\
		&-\int \xi_P(z) \left\{\tilde S_P(y\miid a,z) - \tilde S_0(y\miid a,z)\right\}^2\frac{\tilde G_P(dw\miid a,z)}{\tilde S_0(w\miid a,z)\tilde S_P^2(w\miid a,z)}\,J_0(da,dz)\,;\\
		R_4(P,P_0) :=&\, \left(\frac{\gamma_0 - \gamma_P}{\gamma_0}\right)\left\{R_{4,1}(P,P_0) + R_{4,2}(P,P_0) + R_{4,3}(P,P_0)\right\}\,,
	\end{align*}
	where we define \begin{align*}
		R_{4,1}(P,P_0) :=&\, \iint \xi_P(z) \left\{ \gamma_P(a,z) - \gamma_0(a,z)\right\}  J_0(da,dz)\,;\\
		R_{4,2}(P,P_0) :=&\, \int\sum_{a\in \{0,1\}} \xi_P(z)  \gamma_P(a,z) \left\{\pi_P(a\miid z) - \pi_0(a\miid z)\right\}H_0(dz)\,;\\
		R_{4,3}(P,P_0) :=&\,\int\sum_{a\in \{0,1\}} \xi_P(z) \pi_P(a\miid z) \gamma_P(a,z) \left(H_P - H_0\right)(dz)\,;
	\end{align*} and also write $\xi_P(z) := \{\mu_{P,\varphi}(z) - \Psi(P)\}/\gamma_P$ and $\nu_P(y,a,z) := \tilde S_P(y\miid a,z)/R_P(y\miid a,z)$.
	
	\begin{thm} 
		\label{thm:if} 
		Suppose that conditions \ref{ass:support}--\ref{ass:positivity} and \ref{ass:truncation}--\ref{ass:censoring} hold. Then, the survival integral parameter $P\mapsto \Psi(P)$ is pathwise differentiable in a nonparametric model with efficient influence function $\phi_{P}$, and for each $P$, admits the linearization
		\[\Psi(P) -\Psi(P_0) =\int \phi_P(o)\,(P-P_0)(do)+R(P,P_0)\ .\]
	\end{thm}
	We note that the efficient influence function $\phi_P$ we have provided above agrees with existing results for special cases. In the absence of left truncation, it coincides with results provided in \cite{gerds2017kaplan} for a general survival integral --- this is seen by taking the truncation distribution to be degenerate at zero, in which case $\phi_{2,P}$ simplifies to $(z,a,w,y,\delta)\mapsto \mu_{P,\varphi}(z)-\Psi(P)$ --- and in \cite{westling2021inference} for a counterfactual survival probability, obtained by taking $\varphi(t,z) = I(t\le t_0)$ for a fixed value $t_0$. In the absence of right censoring and any treatment intervention, our result agrees with results presented in~\cite{wang2022doubly} for a marginal survival probability. 
	
	\section{Proposed estimation procedure}
	\label{s:est}
	
	Equation \eqref{eq:id} expresses $\psi_0$ in terms of components of the observed data distribution $P_0$, which are themselves functions of components of the ideal data distribution $P_{X,0}$. An estimator of $\psi_0$ can therefore be obtained by substituting an estimator of relevant components of $P_0$ into $\eqref{eq:id}$, or an estimator of relevant components of $P_{X,0}$, namely $F_{X,0}$ and $H_{X,0}$, into \eqref{eq:target}. Additional parametrizations of $P_0$ --- for example, using a combination of components of $P_0$ and of $P_{X,0}$ --- can also be considered, each leading to a strategy for estimating $\psi_0$ with relative advantages and disadvantages. Here, we consider a particular parametrization that we believe provides a balance between implementability and desirable statistical properties, as we elaborate below. This parametrization is characterized by  the following  combination of components of $P_0$ and $P_{X,0}$:
	
	\begin{enumerate}[1.]\itemsep-0.05in
		\item the observable conditional covariate distribution function $H_0$;
		\item the observable conditional exposure probability $\pi_0$;
		\item the observable conditional truncation distribution function $G_0$;
		\item the observable conditional censoring survival function $Q_0$;
		\item the target conditional time-to-event distribution function $F_{X,0}$ (equivalently, survival function $S_{X,0}$), considered  only on the observable region $[0,\min\{\overline{\tau}_T(a_0,z),\overline{\tau}_C(a_0,z)\})$ for each covariate value $z\in\mathcal{Z}$.
	\end{enumerate} For notational convenience, we denote the vector $(H_0,\pi_0,G_0,Q_0,F_{X,0})$ of nuisance functions as $\eta_0$. We note first that this is indeed a valid parametrization in the sense that two observed data distributions are the same if and only if they agree in this parametrization. We also note that components of this parametrization are variationally-independent in the sense that fixing the value of a subset of components of $\eta_0$ does not constrain the values that the remaining components of $\eta_0$ can take. Additional details on this parametrization are provided in Part \hyperref[partc]{C} of the Appendix. The survival integral value $\psi_0$ can be expressed in terms of $\eta_0$ as \begin{equation}\label{eq:estimationID}\frac{\iiint \varphi(t,z)\,F_{X,0}(dt\miid a_0,z)\int S_{X,0}(u\miid a,z)^{-1}G_{0}(du\miid a,z)\,J_0(da,dz)}{\iiint S_{X,0}(u\miid a,z)^{-1}G_{0}(du\miid a,z)\,J_0(da,dz)}\end{equation} with $J_0(da,dz)=[(1-a)\{1-\pi_0(z)\}+a\pi_0(z)]\,H_0(dz)\,\nu(da)$ itself a function of $\eta_0$. Similarly, the efficient influence function $\phi_0:=\phi_{P_0}$ of $P\mapsto \Psi(P)$ under sampling from $P_0$ can be expressed as a function of $\eta_0$ using the fact that we can write \[ R_{0}(u\miid a,z) :=  S_{X,0}(u\miid a,z)\int_0^u \frac{Q_0(u\miid w,a,z)}{S_{X,0}(w\miid a,z)}\,G_0(dw\miid a,z)\ .\] We write $\psi_0=\psi_{\eta_0}$ and $\phi_0=\phi_{\eta_0}$ to emphasize that $\psi_0$ and $\phi_0$ can be computed based on $\eta_0$. This parametrization is convenient for the purpose of estimation since, on one hand, $H_0$, $\pi_0$ and $G_0$ can be estimated using off-the-shelf regression algorithms based on the observed data, and on the other hand, $Q_0$ and $F_{X,0}$ can be estimated using regression methods for survival data subject to right censoring or right censoring and left truncation, respectively. Furthermore, it facilitates the construct of estimators of $\psi_0$ that enjoy certain robustness properties, as discussed in Section \ref{s:inf}. 
	
	We wish to incorporate flexible learning strategies in our estimation procedure to minimize the risk of systematic bias stemming from the use of misspecified parametric or semiparametric nuisance models. Once an estimator $\eta_n$ of $\eta_0$ is obtained, the naive plug-in estimator $\psi_{\eta_n}$, obtained by replacing $\eta_0$ by $\eta_n$ in the form of $\psi_{\eta_0}$, could be considered. We refer to such an estimator as naive since in general $\eta_n$ need not be tailored to the end goal of estimating $\psi_0$. Furthermore, if $\eta_n$ is estimated flexibly, it is often the case that $\psi_{\eta_n}$ is overly biased and fails to even be $n^{\frac{1}{2}}$--consistent. Debiasing tools are typically used to address this challenge. Here, we employ the one-step debiasing approach based on the efficient influence function \citep{ibragimov1981statistical, pfanzagl1985contributions} as well as the optimal estimating equations framework \citep{van2003unified}.
	
	A standard one-step debiased estimator of $\psi_0$ is given by $\psi_{\eta_n}+\frac{1}{n}\sum_{i=1}^{n}\phi_{\eta_n}(O_i)$. While the simplicity of this estimator is appealing, its asymptotic linearity is only guaranteed to hold under a stringent cap on the flexibility of the procedures used to yield $\eta_n$. To circumvent this constraint, cross-fitting can be incorporated into the construction of the one-step debiased estimator \citep{zheng2011cross, chernozhukov2018double}. In its simplest form, this is achieved by partitioning the sample into two subsamples, using one subsample to obtain $\eta_n$ and the other to build the one-step debiased estimator, repeating this construction with the roles of the subsamples reversed, and finally averaging the two estimators obtained. This procedure can be naturally extended to involve partitioning the sample into $K\geq 2$ subsamples of approximately equal sizes. Specifically, to compute the $K$-fold cross-fitted one-step debiased estimator,  we first randomly partition the index set $\{1,2,\ldots,n\}$ into $K$ subsets, say $\mathcal V_1,\mathcal V_2,\ldots, \mathcal V_K$, of roughly equal sizes $n_1,n_2,\ldots,n_K$. Then, for each $k=1,2,\ldots,K$, an estimate  $\eta_{k,n}=(H_{k,n},\pi_{k,n},G_{k,n},Q_{k,n},S_{k,n})$ of $\eta_0$ is obtained using only observations with indices not in $\mathcal V_k$, and the estimate $\psi^*_{k,n}:=\psi_{\eta_{k,n}}+\frac{1}{n_k}\sum_{i\in\mathcal V_k}\phi_{\eta_{k,n}}(O_i)$ of $\psi_0$ is calculated. Finally, the average $\psi^*_{n}:=\frac{1}{K}\sum_{k=1}^{K}\psi^*_{k,n}$ of fold-specific estimates is taken to be the final estimate of $\psi_0$. For any fixed $K\geq 2$, the cross-fitted one-step debiased estimator is guaranteed to be asymptotically linear without the need to limit the range of algorithms used to estimate $\eta_0$ --- details are provided in Section~\ref{s:inf}.
	
	Alternatively, we consider a second estimator $\psi_n^{**}$ based on solving the efficient influence function estimating equation. To be precise, denoting by $\phi_{\eta,\psi}$ the efficient influence function $\phi_P$ where all nuisances are replaced by corresponding components of $\eta$ but the parameter value $\Psi(P)$ is instead replaced by $\psi$, the fold-specific estimator $\psi_{k,n}^{**}$ is the solution in $\psi$ of the equation \[\sum_{i\in\mathcal{V}_k}\phi_{\eta_{k,n},\psi}(O_i)=0\ ,\]and the estimator $\psi_n^{**}$ is taken to be the average $\frac{1}{K}\sum_{k=1}^{K}\psi_{k,n}^{**}$ of the fold-specific estimators. Because for each fixed realization $o$ of the data unit and each fixed nuisance $\eta$ the mapping $\psi\mapsto\phi_{\eta,\psi}(o)$ is linear in $\psi$, $\psi_n^{**}$ admits a closed-form expression: specifically, we find that $\psi_{k,n}^{**}$ is given explicitly by \[\frac{\sum_{i\in\mathcal{V}_k}\left[\frac{I(A_i=a_0)}{\pi_{k,n}(Z_i)}\gamma_{k,n}(Z_i)\phi_{\text{KM},\eta_{k,n}}(L_{\eta_{k,n},\varphi})(O_i)
		+\mu_{k,n}(Z_i)\left\{\frac{1}{S_{k,n}(W_i\miid A_i,Z_i)}-\phi_{\text{KM},\eta_{k,n}}(\gamma_{k,n,\natural})(O_i)\right\}\right]}{\sum_{i\in\mathcal{V}_k}\left\{\frac{1}{S_{k,n}(W_i\miid A_i,Z_i)}-\phi_{\text{KM},\eta_{k,n}}(\gamma_{k,n,\natural})(O_i)\right\}},\]where we define pointwise $\mu_{k,n}(z):=\int\varphi(t,z)\,\tilde{F}_{k,n}(dt\miid a_0,z)$ with $\tilde{F}_{k,n}:=1-\tilde{S}_{k,n}$. Because $\phi_{\eta,\psi}$ does not have the form $\tilde\phi_{\eta}-\psi$ for any function $\tilde\phi_\eta$ indexed by $\eta$ but not $\psi$, the estimators $\psi_n^*$ and $\psi_n^{**}$ are distinct. As we will see below, these estimators not only differ in their value on given samples but also in at least one key statistical property.
	
	\section{Large-sample inferential results}
	\label{s:inf}
	\subsection{Pointwise statistical inference}
	\label{ss:pinf}
	We study conditions under which the proposed estimators $\psi^*_n$ and $\psi_n^{**}$ are asymptotically linear and nonparametric efficient estimators of the survival integral $\psi_0$. We denote by $\eta_\infty := (H_0, \pi_\infty, G_\infty, Q_\infty, S_{X,\infty})$ the common limit in-probability of the split-specific nuisance estimators $\eta_{1,n},\eta_{2,n}\ldots,\eta_{K,n}$, and by  $\overline \tau(a,z) := \min\{\overline{\tau}_T(a,z),\overline{\tau}_C(a,z)\}$ the maximum possible follow-up time in the subpopulation of individuals with $(A,Z)=(a,z)$. We will refer to the following conditions on the nuisance estimators:
	
	\begin{description}[style=multiline, leftmargin=2.0cm, labelindent=.9cm]
		\item[\namedlabel{ass:pro_conv}{(C1)}]the following consistency conditions hold:
		\begin{enumerate}[label=(\alph*),ref=(C1\alph*)]
			\item $\displaystyle\max_{k}E_0\left|\frac{\bar\gamma_{k,n}(a_0,Z)}{\pi_{k,n}(Z)} - \frac{\bar\gamma_{\infty}(a_0,Z)}{\pi_{\infty}(Z)}\right|^2\stackrel{P}{\longrightarrow}0$; \label{ass:pro_conv_a}
			\item
			$\displaystyle\max_kE_0\left[\sup_{y\in[0,\overline{\tau}(a_0,Z)]}\left|\frac{L_{k,n, \varphi}(y,a_0,Z)}{\tilde S_{k,n}(y\miid a_0, Z)} - \frac{L_{\infty, \varphi}(y,a_0,Z)}{\tilde S_{\infty}(y\miid a_0, Z)}\right|\right]^2\stackrel{P}{\longrightarrow}0$;
			\label{ass:pro_conv_b}
			\item 
			$\displaystyle\max_kE_0\left[\sup_{y\in[0,\overline{\tau}(A,Z)]}\left|\frac{\tilde S_{k,n}(y\miid A, Z)}{ R_{k,n}(y\miid A, Z)} - \frac{\tilde S_{\infty}(y\miid A, Z)}{ R_{\infty}(y\miid A, Z)}\right|\right]^2\stackrel{P}{\longrightarrow}0$;
			\label{ass:pro_conv_c}
			\item
			$\displaystyle\max_kE_0\left[\sup_{y\in[\underline{\tau}_T(A,Z),\overline{\tau}_W(A,Z)]}\left|\frac{\bar{\gamma}_{k,n,\natural}(y, A,Z)}{\tilde S_{k,n}(y\miid A,Z)} - \frac{\bar{\gamma}_{\infty, \natural}(y,A,Z)}{\tilde S_\infty(y\miid A,Z)}\right|\right]^2\stackrel{P}{\longrightarrow}0$;
			\label{ass:pro_conv_d}
			\item $\displaystyle\max_kE_0\left[\sup_{y\in[\underline{\tau}_W(A,Z),\overline{\tau}_W(A,Z)]}\left| \frac{1}{\tilde S_{k,n}(y\miid A,Z)} - \frac{1}{\tilde S_\infty(y\miid A,Z)}\right|\right]^2\stackrel{P}{\longrightarrow}0$;
			\label{ass:pro_conv_e}
			\item 
			$\displaystyle\max_kE_0\left[\sup_{u\in[\underline{\tau}_W(A,Z),\overline{\tau}(A,Z)]}\left|G_{k,n}(u\miid A,Z)-G_{\infty}(u\miid A,Z)\right|\right]^2\stackrel{P}{\longrightarrow}0$;
			\label{ass:pro_conv_f}
		\end{enumerate}
		
		\item[\namedlabel{ass:pro_pos}{(C2)}] there exists some constant $\kappa\in(0,\infty)$ for which the following inequalities hold with $P_0$--probability tending to one:
		\begin{enumerate}[label=(\alph*),ref=(C2\alph*)]
			\item 
			$\displaystyle\frac{\bar\gamma_{k,n}(a_0,Z)}{\pi_{k,n}(Z)}, \frac{\bar\gamma_{\infty}(a_0,Z)}{\pi_{\infty}(Z)} \leq \kappa$;
			\label{ass:pro_pos_a}
			\item 
			$\displaystyle\sup_{y\in[0,\overline{\tau}(a_0,Z)]}\left|\frac{L_{k,n, \varphi}(y,a_0,Z)}{\tilde S_{k,n}(y\miid a_0, Z)}\right|, \sup_{y\in[0,\overline{\tau}(a_0,Z)]}\left|\frac{L_{\infty, \varphi}(y,a_0,Z)}{\tilde S_{\infty}(y\miid a_0, Z)}\right| \leq \kappa$;
			\label{ass:pro_pos_b}
			\item 
			$\displaystyle\sup_{y\in[0,\overline{\tau}(A,Z)]}\left|\frac{\tilde S_{k,n}(y\miid A, Z)}{R_{k,n}(y\miid A, Z)}\right|, \sup_{y\in[0,\overline{\tau}(A,Z)]}\left|\frac{\tilde S_{\infty}(y\miid A, Z)}{R_{\infty}(y\miid A, Z)}\right| \leq \kappa$;
			\label{ass:pro_pos_c}
			\item $\displaystyle\sup_{y\in[\underline{\tau}_T(A,Z),\overline{\tau}_W(A,Z)]}\left|\frac{\bar{\gamma}_{k,n,\natural}(y, A,Z)}{\tilde S_{k,n}(y\miid A,Z)}\right|, \sup_{y\in[\underline{\tau}_T(A,Z),\overline{\tau}_W(A,Z)]}\left|\frac{\bar{\gamma}_{\infty, \natural}(y,A,Z)}{\tilde S_\infty(y\miid A,Z)}\right| \leq \kappa$;
			\label{ass:pro_pos_d}
			\item 
			$\displaystyle\frac{1}{\tilde S_{k,n}(\overline{\tau}_W(A,Z)\miid A,Z)}, \frac{1}{\tilde S_\infty(\overline{\tau}_W(A,Z)\miid A,Z)} \leq \kappa$;
			\label{ass:pro_pos_e}
			\item 
			$\displaystyle|\varphi(Y,Z)|, \int |\varphi(dy,Z)| \leq \kappa$;
			\label{ass:pro_pos_f}
		\end{enumerate}
		
		\item[\namedlabel{ass:consistent}{(C3)}]the limits of the nuisance estimators agree with the true nuisances as follows:
		\begin{enumerate}[label=(\alph*),ref=(C3\alph*)]
			\item $S_{X,\infty}(Y\miid A,Z) = S_{X,0}(Y\miid A,Z)$ and $G_\infty(Y\miid A,Z) = G_0(Y\miid A,Z)$ $P_0$--almost surely;\label{ass:consistent_a}
			\item $\pi_\infty(Z) = \pi_0(Z)$ $P_0$--almost surely;\label{ass:consistent_b}
			\item $Q_\infty(Y\miid A,Z) = Q_0(Y\miid A,Z)$ $P_0$--almost surely;\label{ass:consistent_c}
		\end{enumerate}
		\item[\namedlabel{ass:pro_remainder}{(C4)}] $\max_{k} R(\eta_{k,n}, \eta_0) = o_P(n^{-\frac{1}{2}})$.
	\end{description}
	
	The following theorem describes the large-sample (pointwise) inferential properties of estimators $\psi_n^{*}$ and $\psi_n^{**}$  under appropriate conditions.
	\begin{thm}\label{thm:asymp_point}
		Suppose that conditions~\ref{ass:support}--\ref{ass:positivity} and \ref{ass:pro_conv}--\ref{ass:pro_pos} hold.
		\begin{enumerate}[(i)]
			\item If conditions \ref{ass:consistent_a}--\ref{ass:consistent_b} hold, then $\psi_n^*$  is a consistent estimator of $\psi_0$.\vspace{-.1in}
			\item If condition \ref{ass:consistent_a} holds, then $\psi_n^{**}$ is a consistent estimator of $\psi_0$.\vspace{-.1in}
			\item If conditions~\ref{ass:consistent}--\ref{ass:pro_remainder} hold, then $\psi_n^*$ and $\psi_n^{**}$ are asymptotically linear estimators of $\psi_0$ with common influence function $\phi_0$, that is, \[\psi_n^* \,=\, \psi_n^{**}+o_P(n^{-\frac{1}{2}}) \,=\,  \psi_0+\frac{1}{n}\sum_{i=1}^{n}\phi_0(O_i)+ o_P(n^{-\frac{1}{2}})\ .\] In particular, this implies that  $n^{\frac{1}{2}}\,(\psi_n^* - \psi_0\,)$ and $n^{\frac{1}{2}}\,(\psi_n^{**} - \psi_0)$ each converge in distribution to a normal random variable with mean zero and variance $\sigma_0^2 := var_0\{\phi_{ 0}(O)\}<\infty$.
		\end{enumerate}
	\end{thm}
	A simple estimator of $\sigma_0^2$ can be constructed as $\sigma_n^2:=\frac{1}{n}\sum_{i=1}^{n}\{\phi_{\eta_n}(O_i)-\bar\phi_n\}^2$ with $\bar\phi_n:=\frac{1}{n}\sum_{i=1}^{n}\phi_n(O_i)$. Alternatively, a cross-fitted counterpart of $\sigma_n^2$ with possibly improved finite-sample performance can be obtained as \[\sigma_{n,*}^2:=\frac{1}{K}\sum_{k=1}^K\frac{1}{n_k}\sum_{i\in\mathcal{V}_k}\{\phi_{\eta_{k,n}}(O_i)-\bar\phi_{k,n}\}^2\] with $\bar\phi_{k,n}:=\frac{1}{n_k}\sum_{i\in\mathcal{V}_k}\phi_{\eta_{k,n}}(O_i)$ for $k=1,2,\ldots,K$. Wald confidence intervals with asymptotic coverage $1-\alpha$ can then be constructed as $(\psi_n - q_{\alpha} \sigma_n n^{-\frac{1}{2}}, \psi_n + q_{\alpha} \sigma_n n^{-\frac{1}{2}})$, where $q_{\alpha}$ denotes the $(1-\frac{\alpha}{2})$--quantile of the standard normal distribution. Here, $\sigma_{n,*}$ can also be used instead of $\sigma_n$.
	
	Beyond providing a template for making inference about $\psi_0$, the result above highlights that $\psi_n^*$ and $\psi_n^{**}$ both enjoy some degree of robustness to inconsistent nuisance estimation. Interestingly though, despite the fact that these two estimators are asymptotically equivalent when all nuisance estimators are consistent for their intended target, they have differing behavior when this is not the case. For example, the one-step estimator $\psi_n^*$ retains its consistency for $\psi_0$ even when the censoring distribution is inconsistently estimated provided the time-to-event distribution, truncation distribution, and treatment propensity score are estimated consistently. In contrast, the estimating equations-based estimator $\psi_n^{**}$ is consistent for $\psi_0$ provided the time-to-event and truncation distributions are estimated consistently, irrespective of how poorly the propensity score and censoring distributions may be estimated. As such, $\psi_n^{**}$ exhibits strictly greater robustness than $\psi_n^*$ in terms of consistency.
	
	We note here that if we had parametrized the problem in terms of $(H_{X,0},\pi_{X,0},G_{X,0},Q_0,S_{X,0})$ --- more in line with the work of \cite{wang2022doubly} --- the resulting estimating equations-based estimator would enjoy double robustness, that is, it would be consistent provided $H_{X,0}$ and either $S_{X,0}$ or $(\pi_{X,0},G_{X,0},Q_0)$ are estimated consistently. At first glance, such property may appear superior to the robustness exhibited by $\psi_n^{**}$, which requires consistent estimation of $(H_0,G_0,S_{X,0})$ but not of $(\pi_0,Q_0)$. However, this may not necessarily be so. While there exist flexible strategies for estimating conditional survival functions, such as $S_{X,0}$ and $Q_0$, based on right-censored and/or left-truncated data --- see, e.g., the recent work of  \cite{wolock2024framework} --- to the best of our knowledge, the same cannot be said of $G_{X,0}$, $H_{X,0}$ and $\pi_{X,0}$. Natural estimators of $G_{X,0}$ often involve inverse-weighting an estimator of $G_0$ using an estimator of $S_{X,0}$. Worse yet, natural estimators of $H_{X,0}$ and $\pi_{X,0}$ often involve inverse-weighting estimators of $H_0$ and $\pi_0$ using estimators of both $G_{X,0}$ and $S_{X,0}$. Since, in this problem, robustness to inconsistent estimation of $H_{X,0}$ does not appear possible, consistent estimation of $(G_{X,0},S_{X,0})$ would therefore seem necessary. As such, it is not clear that double robustness could provide any benefit beyond the robustness displayed by $\psi_n^{**}$, all the while rendering estimation of required nuisance functions more challenging. In contrast, all components of the parametrization $(H_0,\pi_0,G_0,Q_0,S_{X,0})$ we have adopted can be readily estimated using machine learning tools.
	
	The conditions imposed in Theorem \ref{thm:asymp_point} can be scrutinized in the context of each application at hand. Condition \ref{ass:pro_conv} requires the weak consistency of certain transformations of the nuisance estimators to their respective (possibly off-target) limits, often in some uniform sense that depends partly on the kernel function $\varphi$ defining the estimand of interest. Condition \ref{ass:pro_pos} requires that these transformations of nuisance estimators as well as their limits be bounded above, at least in large samples, so that all terms involved in the linearization of the survival integral estimator are controlled. Required support recovery assumptions ensure that these conditions hold for the true nuisance values, whereas condition \ref{ass:pro_pos} requires that the same also be true of the nuisance limits. Condition \ref{ass:consistent} is useful to describe various patterns of consistent or inconsistent estimation of certain nuisance components under which consistency of the survival integral estimator may be preserved, as discussed above. Finally, condition \ref{ass:pro_remainder} is a generic condition on the rate of convergence of nuisance estimators --- whether or not it holds in practice depends on the degree of smoothness or structure that the nuisance functions satisfy and whether the nuisance estimators used are able to leverage that structure effectively to achieve fast enough convergence.

	\subsection{Uniform statistical inference}
	\label{ss:unif} 
	We now study conditions under which we can make inference for a class of survival integrals simultaneously. Suppose that $\{\varphi_s:s\in\mathcal{S}\}$ is a collection of kernel functions from $[0,\infty)\times\mathcal{Z}$ to $\mathbb{R}$ indexed by a set $\mathcal{S}$, and that we are interested in learning about a collection of survival integral values $\{\psi_0(s):s\in\mathcal{S}\}$, where $\psi_0(s)$ is the value of $\Psi(P_0)$ corresponding to kernel function $\varphi=\varphi_s$. In most applications, $\mathcal{S}$ is finite-dimensional but that is not a requirement for the developments below. For example, the set of kernels giving rise to the joint distribution function $\mathbb F_0$ of $(T(a_0),Z)$, namely $(t,z)\mapsto \mathbb{P}_0\left\{T(a_0)\leq t,Z\leq z\right\}$, over $[0,\tau)\times \mathcal{Z}$ is specified by $\varphi_s:(u,v)\mapsto I(u\leq t,v\leq z)$ for $s:=(t,z)$ ranging in $\mathcal{S}:=[0,\tau)\times \mathcal{Z}$. The conditions outlined so far pertain to inference for a fixed index $s$. To ensure valid inference uniformly over a range of values, we require the following additional conditions, where for any given kernel function $\varphi=\varphi_s$ we explicitly define $L_{s,P}$ pointwise as $L_{s,P}(y\,|\,a,z):=\int_y^\infty \tilde{S}_P(t\,|\,a,z)\,\varphi_s(dt,z)$, and denote by $\phi_{s, P}$ 
	the nonparametric efficient influence function of $P\mapsto\Psi(P)$ with kernel $\varphi=\varphi_s$ under sampling from $P$, and by $R_{s}(P, P_0)$ the corresponding linearization remainder. We will make use of the  conditions below:
	\begin{description}[style=multiline, leftmargin=2.0cm, labelindent=.9cm] \itemsep-.075in
		\item[\namedlabel{ass:sup_surv}{(D1)}]the following consistency condition holds:
		\begin{flalign*}
			\text{}\ &\ \max_kE_0\left[\sup_{y\in[0,\overline{\tau}(a_0,Z)]}\sup_{s\in\mathcal{S}}\left|\frac{L_{s,k,n, \varphi}(y,a_0,Z)}{\tilde S_{k,n}(y\miid a_0, Z)} - \frac{L_{s,\infty, \varphi}(y,a_0,Z)}{\tilde S_{\infty}(y\miid a_0, Z)}\right|\right]^2\stackrel{P}{\longrightarrow}0\,;&&
		\end{flalign*}
		\item[\namedlabel{ass:uniform_bounded}{(D2)}]there exists some constant $\kappa\in(0,\infty)$ for which the following inequalities hold with $P_0$--probability tending to one:
		\begin{flalign*}
			\text{}\ &\ \sup_{s\in \mathcal S}\sup_{y\in[0,\overline{\tau}(a_0,Z)]}\left|\frac{L_{s,k,n}(y\,|\,a_0,Z)}{S_{k,n}(y\miid a_0,Z)}\right|,\,\sup_{s\in \mathcal S}\sup_{y\in[0,\overline{\tau}(a_0,Z)]}\left|\frac{L_{s,0}(y,a_0,Z)}{S_0(y\miid a_0,Z)}\right|\leq\kappa\,;&&
		\end{flalign*}
		\item[\namedlabel{ass:uniform_remainder}{(D3)}] $\max_k \sup_{s \in \mathcal S}R_{ s}(\eta_{k,n}, \eta_0)=o_P(n^{-\frac{1}{2}})$;
		\item[\namedlabel{ass:phi_donsker}{(D4)}] the set of functions $\left\{\phi_{s, 0} : s \in \mathcal S\right\}$ forms a $P_0$--Donsker class. 
	\end{description}
	Denoting the estimator of $\psi_0(s)$ by $\psi_n(s)$, either corresponding to $\psi_n^*$ or $\psi_n^{**}$ based upon kernel $\varphi=\varphi_s$, we define the standardized process $\mathbb{B}_n:=\{\mathbb{B}_n(s):s\in\mathcal{S}\}$ pointwise as $\mathbb B_n(s) := n^{\frac{1}{2}} \{\psi_n(s) - \psi_0(s)\}$.  The result below provides conditions under which $\mathbb B_n$ converges weakly to the same Gaussian process as the empirical process $\mathbb B^+_n:=\{\mathbb B^+_n(s):s\in\mathcal{S}\}$ defined pointwise as $\mathbb B_n^+(s):=n^{-\frac{1}{2}}\sum_{i=1}^{n}\phi_{s,0}(O_i)$. Below, $\ell^\infty(\mathcal{S})$ refers to the space of uniformly bounded functions from $\mathcal{S}$ to $\mathbb{R}$.
	\begin{thm}
		\label{thm:unif}
		Suppose that condition~\ref{ass:support} holds for each $\varphi=\varphi_s$ with $s\in\mathcal{S}$, and that conditions \ref{ass:pro_conv_a}--\ref{ass:pro_conv_c}, \ref{ass:pro_pos_a}, \ref{ass:consistent} and~\ref{ass:sup_surv}--\ref{ass:phi_donsker} also hold.  Then, the process $\mathbb B_n$ converges weakly to a tight mean-zero Gaussian process with covariance function $\sigma_0(s,t):=E_0\left[\phi_{s,0}(O)\phi_{t,0}(O)\right]$ in the space $\ell^\infty(\mathcal{S})$ relative to the supremum norm over $\mathcal{S}$.
	\end{thm}
	
	This result can be used to numerically construct confidence sets for $\{\psi_0(s):s\in\mathcal{S}\}$ similarly as described in \cite{westling2021inference}. Asymptotically valid fixed-width bands can be readily obtained using a Wald construction and an estimate of relevant quantiles of $\sup_{s\in\mathcal{S}}|\mathbb B_n(s)|$. Alternatively, variable-width bands could be obtained by instead considering re-scaled process $\bar{\mathbb{B}}_n:=\{\bar{\mathbb{B}}_n(s):s\in\mathcal{S}\}$ with $\bar{\mathbb{B}}_n(s):=\sigma(s,s)^{-\frac{1}{2}}\bar{\mathbb{B}}_n(s)$, as in \cite{westling2021inference}. We note that conditions~\ref{ass:sup_surv}, \ref{ass:uniform_bounded} and \ref{ass:uniform_remainder} are uniform counterparts to conditions \ref{ass:pro_conv_b}, \ref{ass:pro_pos_b} and \ref{ass:pro_remainder}. Condition~\ref{ass:phi_donsker} instead puts a constraint on the complexity of the collection consisting of the nonparametric efficient influence function  under sampling from $P_0$ for each survival integral parameter considered. When $\mathcal{S}$ is finite-dimensional, this is achieved, for example, if $s\mapsto\phi_{s,0}$ satisfies a certain Lipschitz condition (see Example 19.7 of \citealp{van2000asymptotic}). More generally, this condition holds if $o\mapsto \phi_{s,0}(o)$ has uniform sectional variation norm uniformly bounded over $s\in\mathcal{S}$. 
	
	\subsection{Extension to smooth functionals}
	\label{ss:had}
	
	The nonparametric inferential procedures we have described pertain to survival integral estimands, which correspond to linear functionals of the G-computation identification $\tilde{\mathbb{F}}_0$ of the joint distribution distribution $\mathbb{F}_0$ of $(T(a_0),Z)$. For any given $(t_0,z_0)$, this identification of $\mathbb{F}_0(t_0,z_0)$ is given by $\upsilon_0$ with $\varphi:(t,z)\mapsto I(t\leq t_0,z\leq z_0)$. However, in some applications, the relevant estimand may be a nonlinear functional of this same distribution. Fortunately, through the delta method, our results on linear functionals can be directly used to tackle a large class of nonlinear functionals.  We demonstrate how the established results readily permit the study of parameters that can be expressed as sufficiently smooth functionals of $\mathbb F_0$.
	
	We denote by $\mathscr{F}$ the collection of distribution functions on $\mathbb{R}\times\mathcal{Z}$ restricted to the identification subset $\mathcal{S}_0:=\{(t,z): 0\leq t\leq \overline{\tau}_C(a_0,z),z\in\mathcal{Z}\}\subseteq \mathbb{R}\times\mathcal{Z}$. Suppose that $\Theta: \mathscr{F}\mapsto\mathbb{R}$ is a given functional, and that we are interested in nonparametric inference on $\theta_0:=\Theta(\tilde{\mathbb{F}}_0)$. The results obtained so far describe how to construct, whenever possible, a uniformly asymptotically linear and regular estimator $\tilde{\mathbb{F}}_n$ of $\tilde{\mathbb{F}}_0$. Provided $\Theta:\mathscr{F}\rightarrow\mathbb{R}$ is sufficiently smooth, it is reasonable to expect that $\theta_n:=\Theta(\tilde{\mathbb{F}}_n)$ is an asymptotically linear and regular estimator of $\theta_0$. Such a result is formalized in the following theorem.
	
	\begin{thm}
		\label{thm:hada}
		Suppose that the conditions of Theorem \ref{thm:unif} hold for $s=(t_0,z_0)$, $\varphi_s:(t,z)\mapsto I(t\leq t_0,z\leq z_0)$ and $\mathcal{S}=\mathcal{S}_0$. If $\Theta$ is Hadamard differentiable at $\mathbb{F}_0$ relative to the supremum norm, it holds that
		\begin{equation}
			\theta_n-\theta_0\,=\,\frac{1}{n}\sum_{i=1}^n \partial\Theta\big{(}\tilde{\mathbb F}_0; (t_0,z_0)\mapsto\phi_{t_0,z_0,0}(O_i)\big{)} + o_P(n^{-\frac{1}{2}})\ ,
		\end{equation}
		where $\partial\Theta\big{(}\tilde{\mathbb{F}}_0;h\big{)}$ denotes the G\^{a}teaux derivative of $\Theta$ at $\tilde{\mathbb{F}}_0$ in a given direction $h:\mathcal{S}\rightarrow\mathbb{R}$. In particular, this implies that $n^\frac{1}{2}(\theta_n-\theta_0)$ tends to a mean-zero normal random variable with variance given by $\sigma^2_\Theta:=var_0\left[\partial\Theta\big{(}\tilde{\mathbb{F}}_0;(t_0,z_0)\mapsto\phi_{t_0,z_0,0}(O)\big{)}\right]$.
	\end{thm}
	
	In Section \ref{s:ex}, we explicitly discuss the implications of this result in the context of Examples 4 and 5, two motivating examples provided in Section \ref{s:id} that feature nonlinear functionals arising in applications. Before proceeding, we use the result above to describe how to study problems in which the estimand can be expressed as the solution of an estimating equation.

	\subsection{Application to estimating equations}
	\label{ss:ee}
	Suppose that the estimand of interest $m_0$ can be expressed as the solution in $m$ of the population estimating equation \[\iint U_m(t,z)\,\tilde{\mathbb{F}}_0(dt,dz)=0\ ,\]where for each $m\in\mathbb{R}$ the function $U_m$ maps from $\mathbb{R}\times\mathcal{Z}$ to $\mathbb{R}$, and $U_m$ has finite uniform sectional variation norm  over $\mathcal{S}_0$ in a neighborhood of $m_0$. Quantiles of the marginal distribution of $T(a_0)$, for example, can be expressed in this manner, as explicitly discussed later. Similarly as before, to ensure identification, it must be the case that, for $P_{X,0}$--almost every $z\in\mathcal{Z}$, the mapping $t\mapsto U_m(t,z)$ is constant for $t>\overline{\tau}_C(a_0,z)$. Under certain regularity conditions, the solution $m_n$ of the empirical version of this estimating equation in $m$, $\iint U_m(t,z)\,\tilde{\mathbb{F}}_n(dt,dz)=0$, is an asymptotically linear and regular estimator of $m_0$ with influence function given by \[\phi_{m,0}:o\mapsto-\left\{\left.\frac{\partial}{\partial m}\iint U_m(t,z)\,\tilde{\mathbb{F}}_0(dt,dz)\right|_{m=m_0}\right\}^{-1}\iint U_{m_0}(t,z)\phi_{dt,dz,0}(o)\ ,\]where $\phi_{dt,dz,0}(o)$ denotes the differential of $(t,z)\mapsto \phi_{t,z,0}(o)$ with $\phi_{t,z,0}$ the influence function of $\tilde{\mathbb{F}}_n(t,z)$.

	\section{Revisiting motivating examples}
	\label{s:ex}
	\subsection*{Example 1: marginal survival probability}
	\label{ss:ex1}
	To begin, we consider perhaps the simplest survival functional, the marginal survival probability $\psi_0:=P_{X,0}\left(T\geq \tau\right)$ for some $\tau>0$. Considered as a function of $\tau$, this estimand consists of the marginal survival function, which describes the entire time-to-event distribution and is commonly reported, for example, in studies of the natural history of disease in a population. A marginal survival probability is typically estimated nonparametrically using the Kaplan-Meier estimator, which handles left truncation in addition to right censoring through a straightforward risk-set adjustment \citep{kaplan1958nonparametric,tsai1987biometrika}. However, this estimator builds upon marginal independence between $T$ and $(W,C)$, and is inconsistent when this condition fails to hold. Here, we  allow dependence between $T$ and $(W,C)$ so long as independence holds within strata defined by a baseline covariate vector $Z$.
	
	In recent decades, several authors have developed methods for estimating a marginal survival probability allowing covariate-dependent right censoring (e.g., \citealp{robins1993information, murray1996nonparametric, zeng2004annals,moore2009increasing}) or covariate-dependent left truncation (e.g., \citealp{chaieb2006estimating, mackenzie2012survival, vakulenko2022nonparametric}). However, to the best of our knowledge, there is currently no method in the literature that accounts for covariate-dependent right censoring and left truncation, facilitates the conduct of valid inference even when flexible learning strategies are used to estimate involved nuisance functions, and also enjoys robustness to inconsistent estimation of certain nuisance functions. The recent work of \cite{wang2022doubly} proposes an approach that comes closest to achieving these desiderata, though their procedure does not appear to generally allow flexible estimation of the conditional censoring distribution and robustness to its possibly inconsistent estimation.
	
	Under conditions~\ref{ass:truncation}--\ref{ass:censoring}, the marginal survival probability $S_{X,0}(\tau)$ at fixed time $\tau$ can be expressed as the survival integral
	$\iint I(t\ge \tau)\,F_{X,0}(dt\miid z)\,H_{X,0}(dz)$ corresponding to kernel function $\varphi_\tau(t,z):=I(t\ge \tau)$. Using Theorem~\ref{thm:if}, and taking $A$ to be degenerate in order to recover the simpler setting in which there is no exposure on which we wish to intervene, the nonparametric efficient influence function of the corresponding parameter under sampling from $P_0$ is given by
	\begin{align*}
		&o\mapsto - S_{X,0}(\tau\miid z)\bar\gamma_0(z)\left[\frac{\delta I(y\le \tau)}{R_0(y\miid z)} - \int_w^{y\wedge \tau}\frac{\Lambda_{X,0}(du\miid z)}{R_0(u\miid  z)}\right]\\
		&\hspace{1in}+ \frac{S_{X,0}(\tau\miid z) - \psi_0}{\gamma_0}\left\{\frac{1}{S_{X,0}(w\miid z)} - \phi_{\text{KM},0}(\gamma_{0,\natural})(z, w, y,\delta)\right\}.
	\end{align*} We can use our uniform results to construct an estimator and make inference on the marginal survival function $\tau\mapsto S_{X,0}(\tau)$ in an interval within which identification is possible. To enforce monotonicity of the resulting estimator, the resulting function-valued estimator can be projected into the space of monotone functions using isotonic regression~\citep{westling2020correcting}. 
	
	\subsection*{Example 2: Brier score}
	\label{ss:ex2}
	The Brier score is defined as the survival integral $\psi_0:=\iint \{I(t\geq \tau)-b(z)\}^2\,F_{X,0}(dt\miid z)\,H_{X,0}(dz)$ corresponding to kernel function $\varphi_\tau(t,z):=\{I(t\geq \tau)-b(z)\}^2$, and has been used to quantify the predictive performance of a given algorithm $z\mapsto b(z)$ for predicting whether the event will occur by some fixed time $\tau$ \citep{brier1950verification}. Here again, we consider $A$ to be degenerate since there is no exposure on which we wish to intervene in this example. The Brier score can be used to compare the performance of several prediction models. In the absence of left truncation, \cite{gerds2006consistent} derived an inferential procedure allowing conditionally-independent right censoring. In the more general case of left-truncated right-censored data, Theorem~\ref{thm:if} implies that the nonparametric influence function for the corresponding parameter under sampling from $P_0$ is given by 
	\begin{align*}
		&o\mapsto  -{ S_{X,0}(\tau\miid z)\bar\gamma_0(z)}\{1 - 2b(z)\}\left\{\frac{\delta I(y\le \tau)}{R_0(y\miid z)} - \int_w^{y\wedge \tau}\frac{\Lambda_{X,0}(du\miid z)}{R_0(u\miid z)}\right\}\\
		&\hspace{1in}+ \frac{\left\{S_{X,0}(\tau\miid z) - b (z)\right\}^2 - \psi_0}{\gamma_0}\left\{\frac{1}{S_{X,0}(w\miid z)} - \phi_{\text{KM},0}(\gamma_{0,\natural})(z, w, y,\delta)\right\}.
	\end{align*}Since it is often used heuristically to compare the performance of competing prediction algorithms, the Brier score is often not itself the target of inference in a given problem. Nevertheless, our results make it straightforward to perform formal tests to rigorously compare candidate prediction algorithms.
	
	\subsection*{Example 3: counterfactual survival probability}
	\label{ss:ex3}
	The counterfactual (or treatment-specific) survival probability $\psi_0:=\mathbb{P}_0\left\{T(a_0)\geq \tau\right\}$, defined as the probability of survival beyond a given time $\tau$ under an intervention that sets treatment (or exposure) to a specified level, is a useful summary to study treatment effects in the context of time-to-event endpoints.  Differences or ratios in counterfactual survival probabilities across treatment levels are commonly used in clinical trials to quantify treatment effects.
	
	In addition to adjustment needed for possibly covariate-dependent right censoring and left truncation, in order to identify the counterfactual survival probability, adjustment for possible confounding between treatment level and survival is also required.
	Under typical causal conditions, including positivity and conditional randomization, the counterfactual survival distribution is identified by the survival integral
	$\iint I(t\ge \tau)\,F_{X,0}(dt\miid a_0, z)\,H_{X,0}(dz)$ corresponding to kernel function $\varphi_\tau(t,z):=I(t\geq \tau)$. This is the same expression as in Example 1 but without degeneracy of $A$. \cite{westling2021inference} derived a nonparametric efficient estimation procedure for a counterfactual survival probability based on right-censored data. In view of Theorem~\ref{thm:if}, the nonparametric efficient influence function of the corresponding parameter under sampling from $P_0$ is given by 
	\begin{align*}
		&o\mapsto- S_{X,0}(\tau\miid a_0,z)\bar\gamma_0(z)\frac{ I(a=a_0)}{\pi_0(z)}\left[\frac{\delta I(y\le \tau)}{R_0(y\miid a_0,z)} - \int_w^{y\wedge \tau}\frac{\Lambda_{X,0}(du\miid a_0,z)}{R_0(u\miid a_0, z)}\right]\\
		&\hspace{1in}+\frac{S_{X,0}(\tau\miid a_0, z) - \psi_0}{\gamma_0}\left\{\frac{1}{S_{X,0}(w\miid a_0, z)} - \phi_{\text{KM},0}(\gamma_{0,\natural})(z, a, w, y,\delta)\right\}.
	\end{align*}When there is no truncation, this influence function simplifies and agrees exactly with that provided in Theorem 2 of \cite{westling2021inference}. When $A$ is degenerate, we recover the expression obtained in Example 1. As in Example 1, our results readily described how to perform uniform inference for the counterfactual survival function over an interval on which identification is possible.
	
	\subsection*{Example 4: median counterfactual event time}
	\label{ss:ex4}
	Due to the presence of right censoring, summaries of the (marginal or counterfactual) time-to-event distribution that depend on its right tail --- for example, the mean event time --- are typically not identified. Because it generally circumvents this right-tail issue, and also because it affords an interpretation that is relevant in many scientific problems, the median event time is often used instead of the mean survival time in survival analysis. Nonparametric inference on the median marginal event time with right-censored data has a long history in survival analysis, dating at least as far back as \cite{reid1981estimating}. Corresponding results for the median counterfactual event time are far more recent, with \cite{diaz2017efficient} and \cite{shepherd2022confounding} studying the problem in the absence of censoring and truncation.
	
	Here, we are specifically interested in nonparametric inference on the median $m_0$ of the counterfactual time-to-event distribution, namely the median corresponding to the marginal distribution of $T(a_0)$, based on left-truncated right-censored data. To the best of our knowledge, this problem has not been studied before. Of course, as discussed in Examples 1 and 3, results obtained for the median counterfactual event time imply results for the median marginal event time.
	While the median counterfactual event time cannot be expressed as a survival integral, it can be framed as the solution of an estimating equation based on a survival integral. Specifically, since $\iint U_m(t,z)\,\tilde{\mathbb{F}}_0(dt,dz)=0$ if and only if $m=m_0$ with $U_m(t,z):=I(t\leq m)-0.5$ under mild conditions, results from the last section can be directly used to characterize the median counterfactual event time parameter. We first note that
	\begin{align*}
		\left.\frac{\partial}{\partial m}\iint U_m(t,z)\,\tilde{\mathbb{F}}_0(dt,dz)\right|_{m=m_0}\ &=\ \left.\frac{\partial}{\partial m}\iint U_m(t,z)\,F_{X,0}(dt\miid a_0,z)\,H_{X,0}(dz)\right|_{m=m_0}\\
		&=\ \int  f_{X,0}(m_0\miid a_0,z)\,H_{X,0}(dz)\ =\ \int  f_{X,0}(m_0\miid a_0,z)\,\tilde{H}_{0}(dz)\ , 
	\end{align*}
	which is simply the G-computation identification of the density function of $T(a_0)$ evaluated at $m_0$. Results from Section \ref{ss:ee} then yield that, under sampling from $P_0$, the nonparametric influence function of the counterfactual median survival time parameter is given by \[o\mapsto \frac{\iint \{0.5-I(t\leq m_0)\}\,\phi_{dt,dz}(o)}{\int f_{X,0}(m_0\miid a_0,z)\tilde H_0(dz)}\ .\]
	
	\subsection*{Example 5: model-agnostic measure of dependence of $T(a_0)$ on $a_0$}
	\label{ss:ex5}
	Our final example serves to illustrate that the class of parameters covered by our theoretical results is large enough to include methodologically challenging estimands of scientific interest. An assumption-lean approach to Cox regression recently proposed in \cite{vansteelandt2022assumption} is based on a nonparametric projection estimand that, in the case of a binary exposure $A$, summarizes  $(t,z)\mapsto \log\Lambda_{X,0}(t\miid 1, z)-\log \Lambda_{X,0}(t\miid 0, z)$, the difference in log conditional cumulative hazard functions. In this spirit, we consider here a novel version of this estimand, defined as the contrast $\mathbb{L}(1)-\mathbb{L}(0)$ with \[\mathbb{L}(a_0):=\int \omega(t)\log\left[-\log\mathbb{P}_0\{T(a_0) > t\}\right]\,dt\] for $a_0\in\{0,1\}$ and some fixed weight function $\omega$ such that $\int \omega(t)\,dt=1$. We have used here that $-\log \mathbb{P}_0\{T(a_0)>t\}$ equals the counterfactual cumulative hazard function at $t$ under continuity of the distribution function. The function $\omega$ is taken to emphasize scientifically relevant values of $t$ and to also restrict attention to values of $t$ at which the distribution functions of $T(0)$ and $T(1)$ are identified. If the structural Cox model $\mathbb{\Lambda}_{1}(t)=\mathbb{\Lambda}_0(t)\exp(\rho)$ holds for some $\rho\in\mathbb{R}$ and all $t\geq 0$, with $\mathbb{\Lambda}_{a_0}(t)$ denoting the cumulative hazard function of the counterfactual survival time $T(a_0)$ evaluated at $t$, the causal contrast $\mathbb{L}(1)-\mathbb{L}(0)$ corresponds to the structural parameter $\rho$.
	
	Notably, this estimand cannot be expressed as a survival integral; nevertheless, it is generally a Hadamard differentiable functional of the counterfactual distribution function. Specifically, writing $\Theta(\tilde{\mathbb{F}}):=\int \omega(t) \log[-\log \{1-\tilde{\mathbb{F}}(t)\}]\,dt$, the corresponding parameter is the difference between the evaluation of $\Theta$ on identifications of the counterfactual distribution function of $T(1)$ and on that of $T(0)$. In view of Theorem~\ref{thm:hada}, it is not difficult to show that, under sampling from $P_0$, this parameter has nonparametric efficient influence function given by
	\begin{align*}
		&o\mapsto \int \omega(t) \left[\frac{\phi_{1,t}(o)}{\{1-\tilde{\mathbb{F}}_1(t)\}\log\{1-\tilde{\mathbb{F}}_1(t)\}}-\frac{\phi_{0,t}(o)}{\{1-\tilde{\mathbb{F}}_0(t)\}\log\{1-\tilde{\mathbb{F}}_0(t)\}}\right]dt\ ,
	\end{align*}
	where $\phi_{a_0,t}$ is the nonparametric efficient influence function of the parameter identifying the cumulative distribution $\mathbb{F}_{a_0}$ of $T(a_0)$ at $t$, presented explicitly in Example~3. While this result applies when $\omega$ is fixed, it is possible to derive extended results for a weight function indexed by $P_0$ itself although those calculations would typically need to be done on a case-by-case basis unless the dependence of the weight function on $P_0$ is smooth enough.
	
	\section{Numerical illustrations}\label{s:sim}
	
	We now present results from a simulation study based on Example 1 described above. Specifically, we make inference about a marginal survival function in the absence of a treatment variable but including covariate-dependent censoring and truncation. We considered settings with high, low and no truncation levels, and high and low censoring levels.
	
	Covariates $Z_1$, $Z_2$, and $Z_3$ are independent random variables distributed uniformly on the set $\{-1,+1\}$. Given covariate vector $Z:=(Z_1,Z_2,Z_3) = z$ with $z:=(z_1,z_2,z_3)$, the study entry time variable $W$ is distributed as $10U$, where $U$ is a Beta random variable with parameters $\alpha=a(z_1)$ and $\beta=b(z_1)$. In the low truncation setting (25\% truncation), we set $a(z_1) = 1$ and $b(z_1) = 1+ 2I(z_1 < 0)$; in the high truncation setting (50\% truncation), we set $a(z_1) = 1 + I(z_1 > 0)$ and $b(z_1) = 1$. Given $W = w$ and $Z = z$, the censoring time is taken  $C=w+D$ with $D$ an independent random variable generated from a Gamma distribution with shape $k$ and scale $\lambda_C = \exp\{-\frac{1}{10}(z_1 + z_2 + z_3)\}$. The parameter $k$ was chosen to yield a low (25\%) or high (50\%) censoring rate depending on the simulation scenario. Given covariate vector $Z = z$, we independently simulated the event time $T$ from a Gamma distribution with shape 6 and scale $\lambda_T = \exp\{\frac{1}{10}(z_1 + z_2 + z_3)\}$.

	The marginal survival function was estimated at times corresponding to the first four quintiles of the population time-to-event distribution. Performance was assessed using the following metrics:\begin{enumerate}[(a)]
		\item empirical bias, scaled by $n^{\frac{1}{2}}$;\vspace{-.15in}
		\item empirical variance, scaled by $n$;\vspace{-.15in}
		\item pointwise confidence interval coverage.
	\end{enumerate} We note here that, in the absence of a treatment variable $A$, estimators $\psi_n^*$ and $\psi_n^{**}$ coincide with each other. As such, all results presented here pertain to this common estimation procedure. Each nuisance function featured in the construction of the proposed estimator were estimated using flexible approaches.  The target conditional time-to-event and censoring survival functions were estimated using global survival stacking \citep{wolock2024framework}, a recent machine learning-based method for estimating conditional survival functions using censored and/or truncated data, with a learner library consisting of the empirical mean, logistic regression, generalized additive models, multivariate adaptive regression splines, random forests, and gradient-boosted trees. For comparison, a correctly specified parametric regression model was also included. The observable conditional truncation distribution was estimated using the stratified empirical distribution function. The five-fold cross-fitted one-step estimator was compared to the Kaplan-Meier estimator \citep{kaplan1958nonparametric}, which assumes independent censoring and truncation.
	
	Across all six scenarios considered, our simulation study reveals similar patterns. The debiased global stacking-based estimator, as implemented in Section~\ref{ss:pinf}, has negligible bias, appropriate coverage, and similar variability as the debiased estimator based on correctly specified parametric nuisance estimators. This is not surprising, since nuisance estimators only play a second-order role in the behavior of debiased estimators --- in fact, as long as nuisance functions are estimated sufficiently well, the debiased estimator with estimated nuisances is asymptotically equivalent to the (oracle) debiased estimator based on the true nuisance functions. As expected, the Kaplan-Meier is biased in all settings, but resulting confidence intervals can be either conservative or anticonservative.

	\begin{figure}
		\begin{center}
			\includegraphics[width=\linewidth]{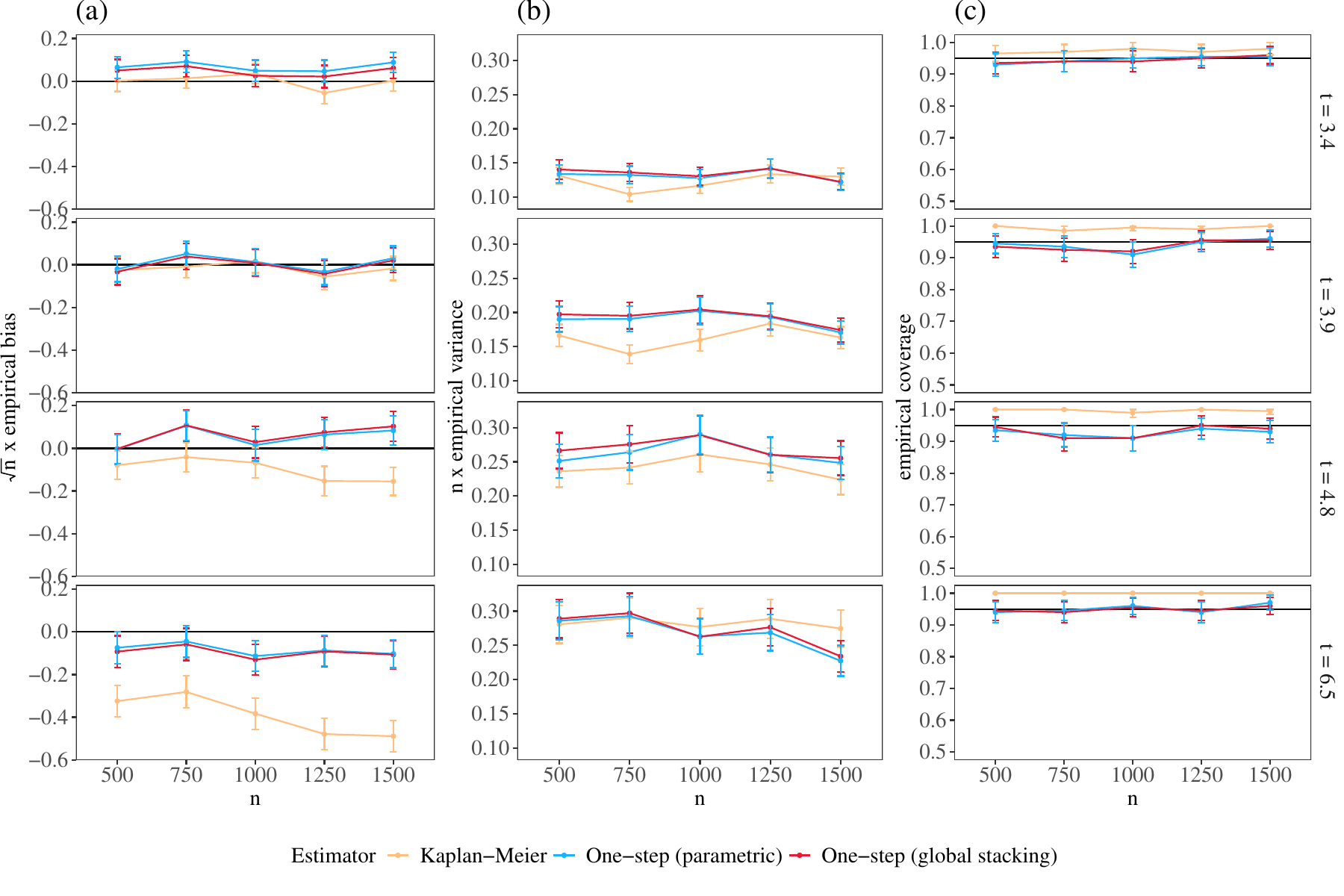}
		\end{center}
		\vspace{-0.5cm}
		\caption{Summary of inferential performance metrics for simulation scenario with 25\% censoring and no truncation. (a) empirical bias scaled by $n^{\frac{1}{2}}$; (b) empirical variance scaled by $n$; (c) empirical coverage of nominal 95\% confidence intervals. Rows correspond to the four times at which the survival function was estimated. The colors denote different estimators, including the Kaplan-Meier estimator and the proposed one-step estimator with nuisance parameters estimated using either correctly specified parametric models or global survival stacking. Vertical bars represent 95\% confidence intervals accounting for Monte Carlo error.}
		\label{fig:C25T0}
	\end{figure}
	
	\begin{figure}
		\begin{center}
			\includegraphics[width=\linewidth]{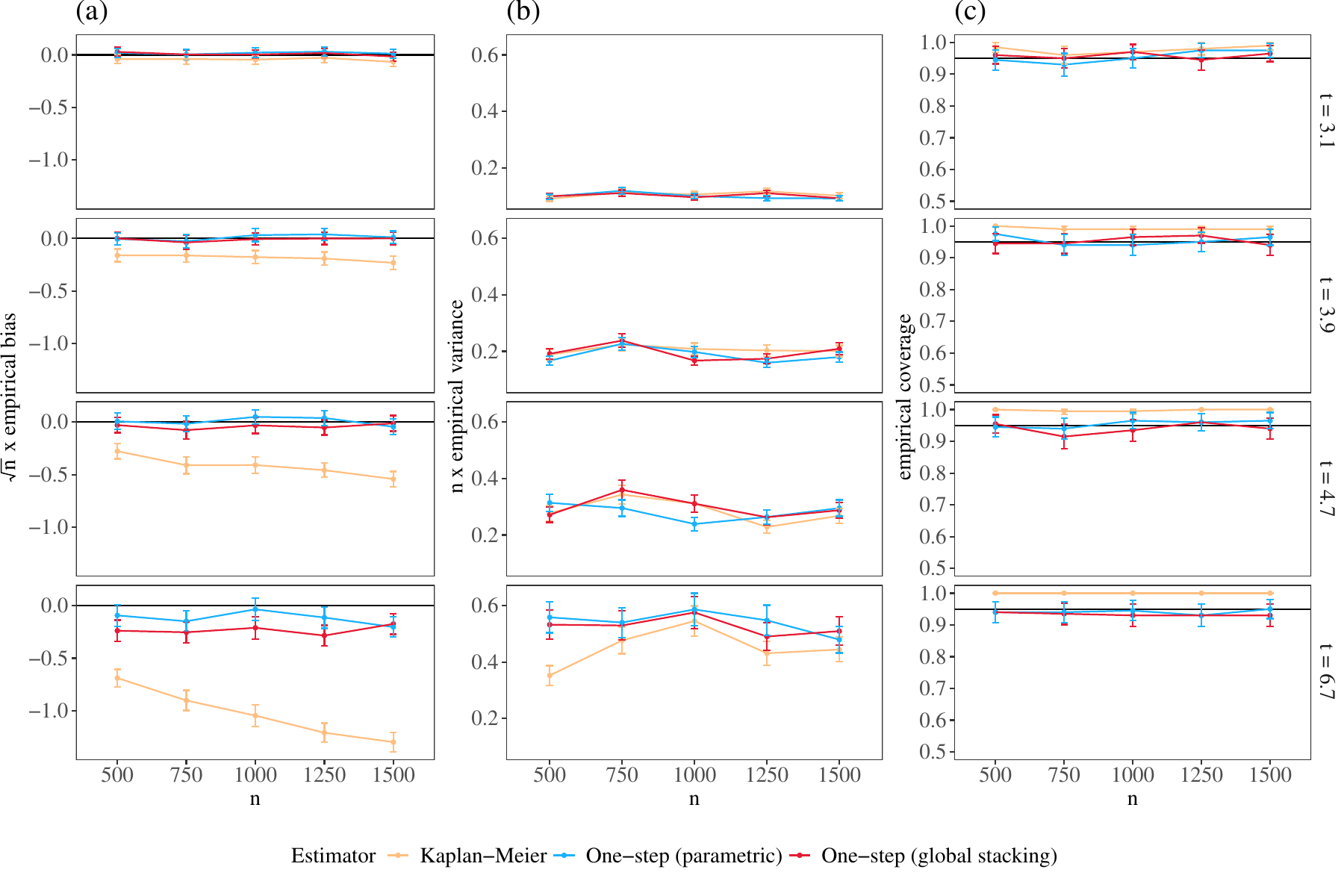}
		\end{center}
		\vspace{-0.5cm}
		\caption{Summary of inferential performance metrics for simulation scenario with 50\% censoring and no truncation. (a) empirical bias scaled by $n^{\frac{1}{2}}$; (b) empirical variance scaled by $n$; (c) empirical coverage of nominal 95\% confidence intervals. Rows correspond to the four times at which the survival function was estimated. The colors denote different estimators, including the Kaplan-Meier estimator and the proposed one-step estimator with nuisance parameters estimated using either correctly specified parametric models or global survival stacking. Vertical bars represent 95\% confidence intervals accounting for Monte Carlo error.}
		\label{fig:C50T0}
	\end{figure}
	
	\begin{figure}
		\begin{center}
			\includegraphics[width=\linewidth]{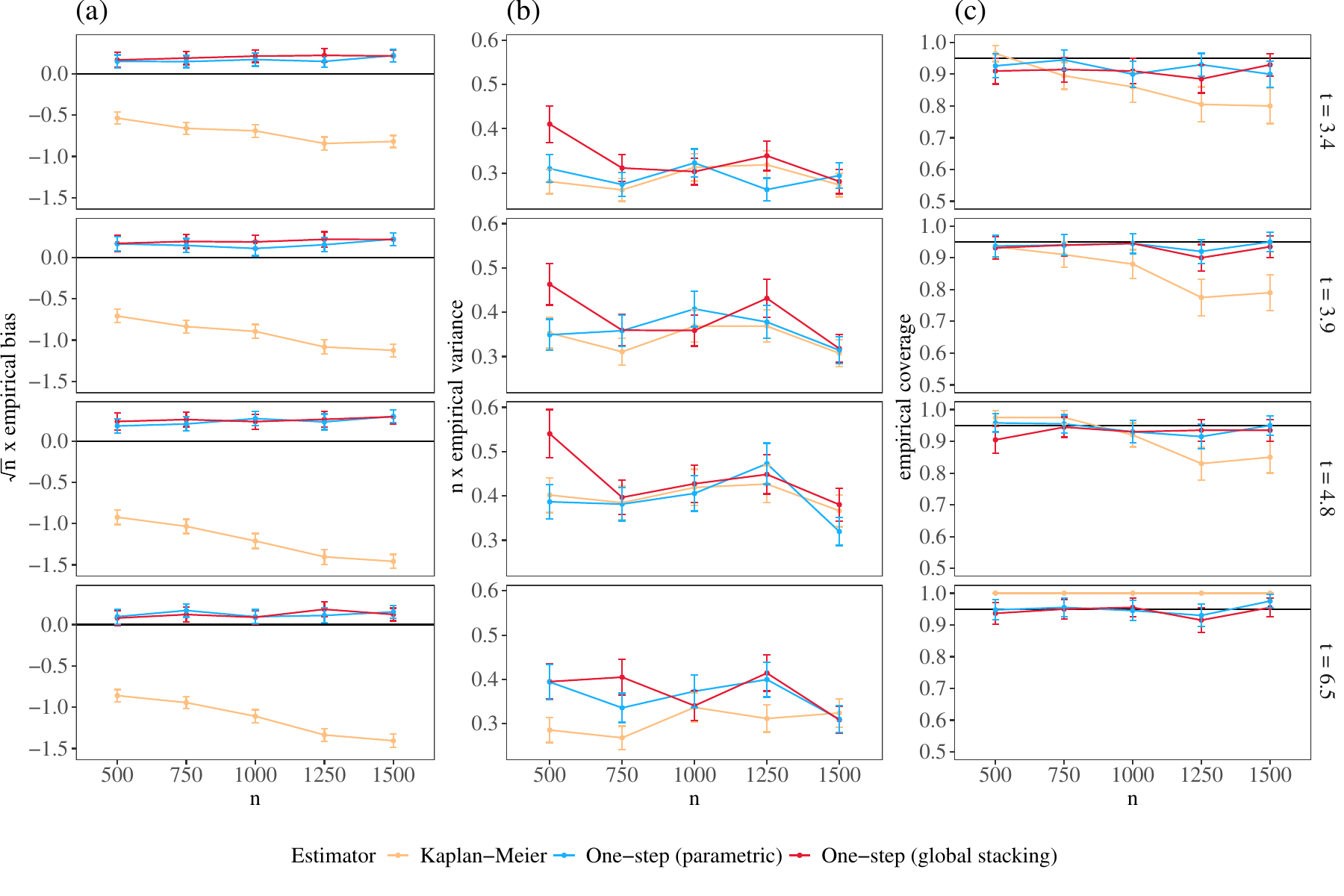}
		\end{center}
		\vspace{-0.5cm}
		\caption{Summary of inferential performance metrics for simulation scenario with 25\% censoring and 25\% truncation. (a) empirical bias scaled by $n^{\frac{1}{2}}$; (b) empirical variance scaled by $n$; (c) empirical coverage of nominal 95\% confidence intervals. Rows correspond to the four times at which the survival function was estimated. The colors denote different estimators, including the Kaplan-Meier estimator and the proposed one-step estimator with nuisance parameters estimated using either correctly specified parametric models or global survival stacking. Vertical bars represent 95\% confidence intervals accounting for Monte Carlo error.}
		\label{fig:C25T25}
	\end{figure}
	
	\begin{figure}
		\begin{center}
			\includegraphics[width=\linewidth]{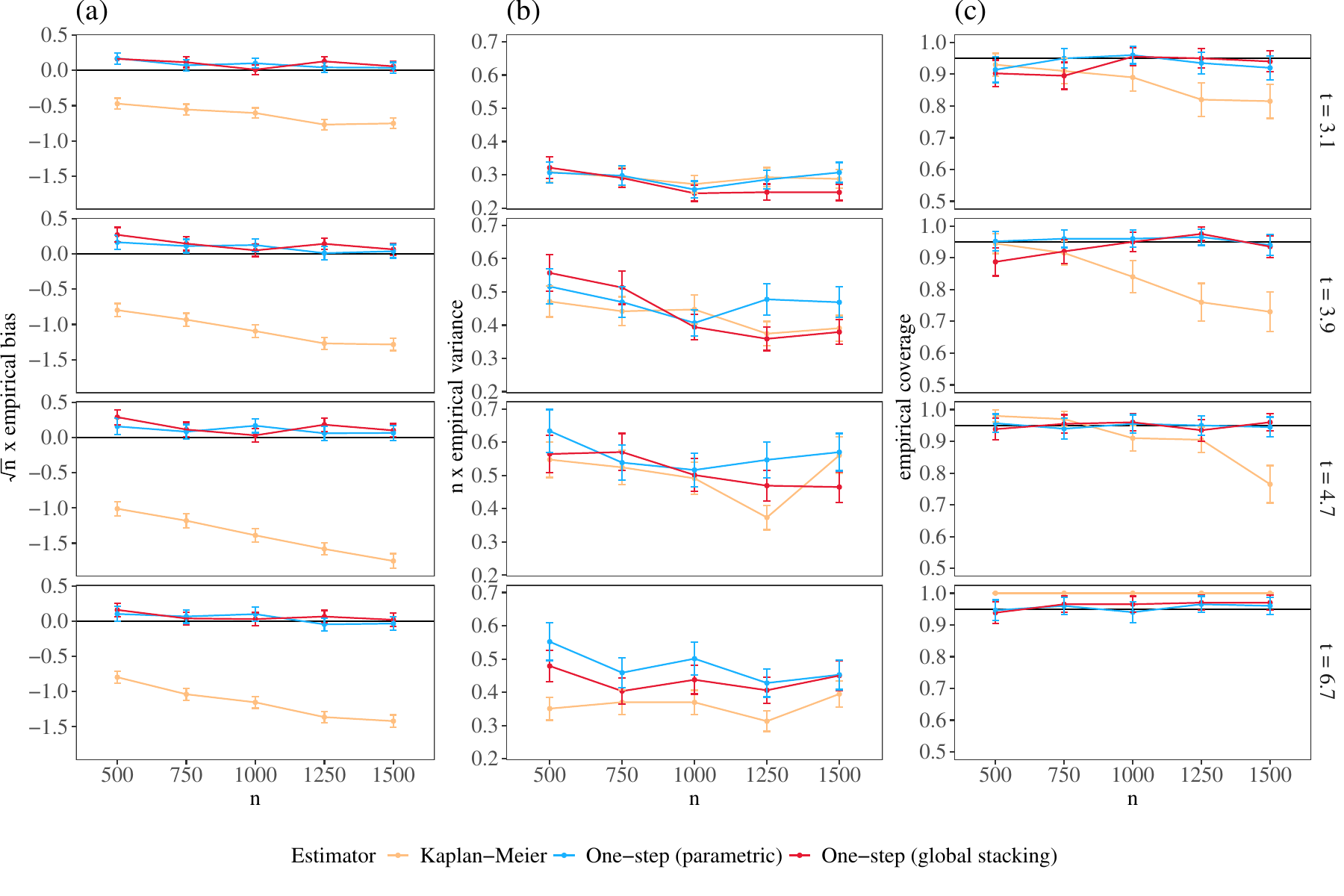}
		\end{center}
		\vspace{-0.5cm}
		\caption{Summary of inferential performance metrics for simulation scenario with 50\% censoring and 25\% truncation. (a) empirical bias scaled by $n^{\frac{1}{2}}$; (b) empirical variance scaled by $n$; (c) empirical coverage of nominal 95\% confidence intervals. Rows correspond to the four times at which the survival function was estimated. The colors denote different estimators, including the Kaplan-Meier estimator and the proposed one-step estimator with nuisance parameters estimated using either correctly specified parametric models or global survival stacking. Vertical bars represent 95\% confidence intervals accounting for Monte Carlo error.}
		\label{fig:C50T25}
	\end{figure}
	
	\begin{figure}
		\begin{center}
			\includegraphics[width=\linewidth]{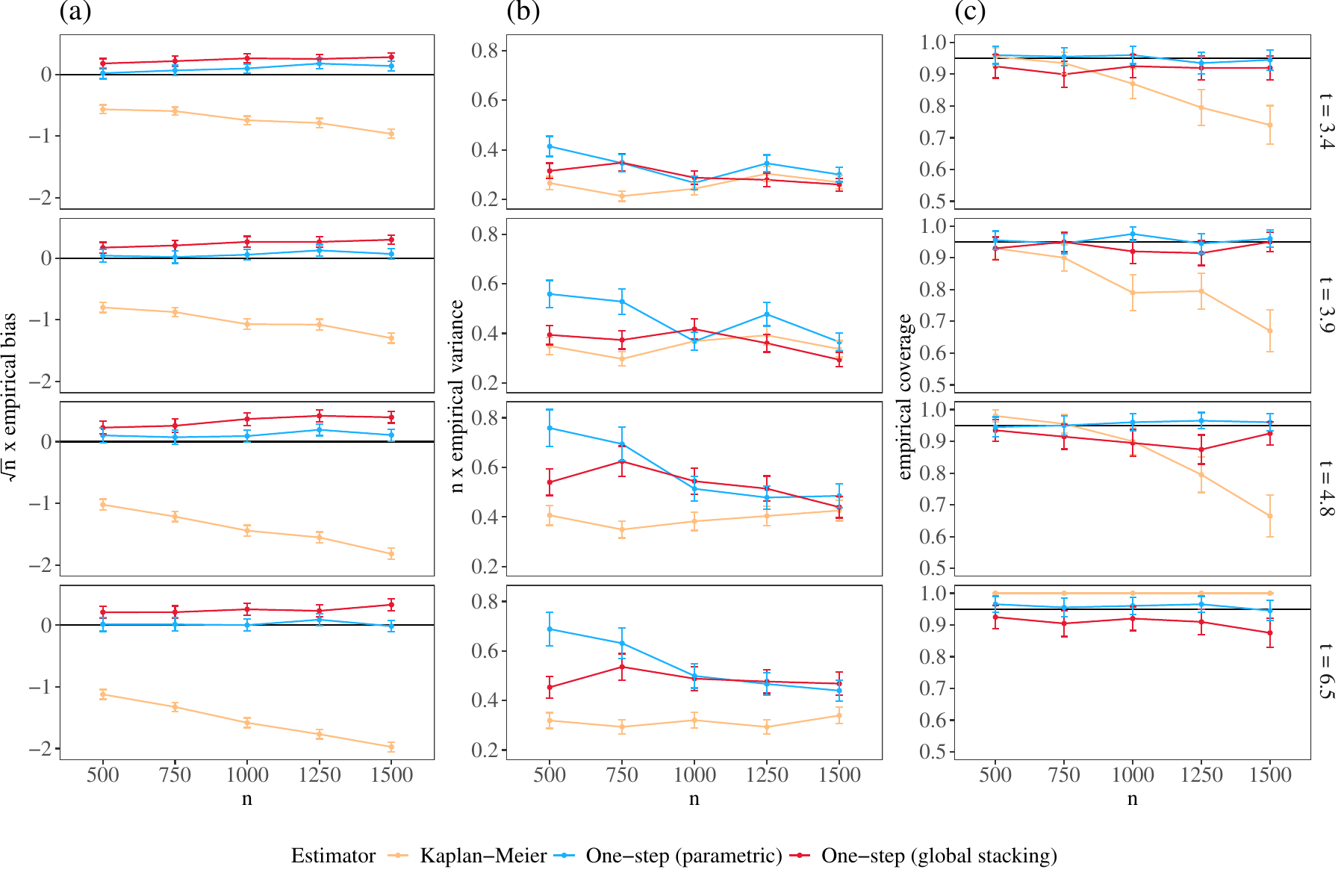}
		\end{center}
		\vspace{-0.5cm}
		\caption{Summary of inferential performance metrics for simulation scenario with 25\% censoring and 50\% truncation. (a) empirical bias scaled by $n^{\frac{1}{2}}$; (b) empirical variance scaled by $n$; (c) empirical coverage of nominal 95\% confidence intervals. Rows correspond to the four times at which the survival function was estimated. The colors denote different estimators, including the Kaplan-Meier estimator and the proposed one-step estimator with nuisance parameters estimated using either correctly specified parametric models or global survival stacking. Vertical bars represent 95\% confidence intervals accounting for Monte Carlo error.}
		\label{fig:C25T50}
	\end{figure}
	
	\begin{figure}
		\begin{center}
			\includegraphics[width=\linewidth]{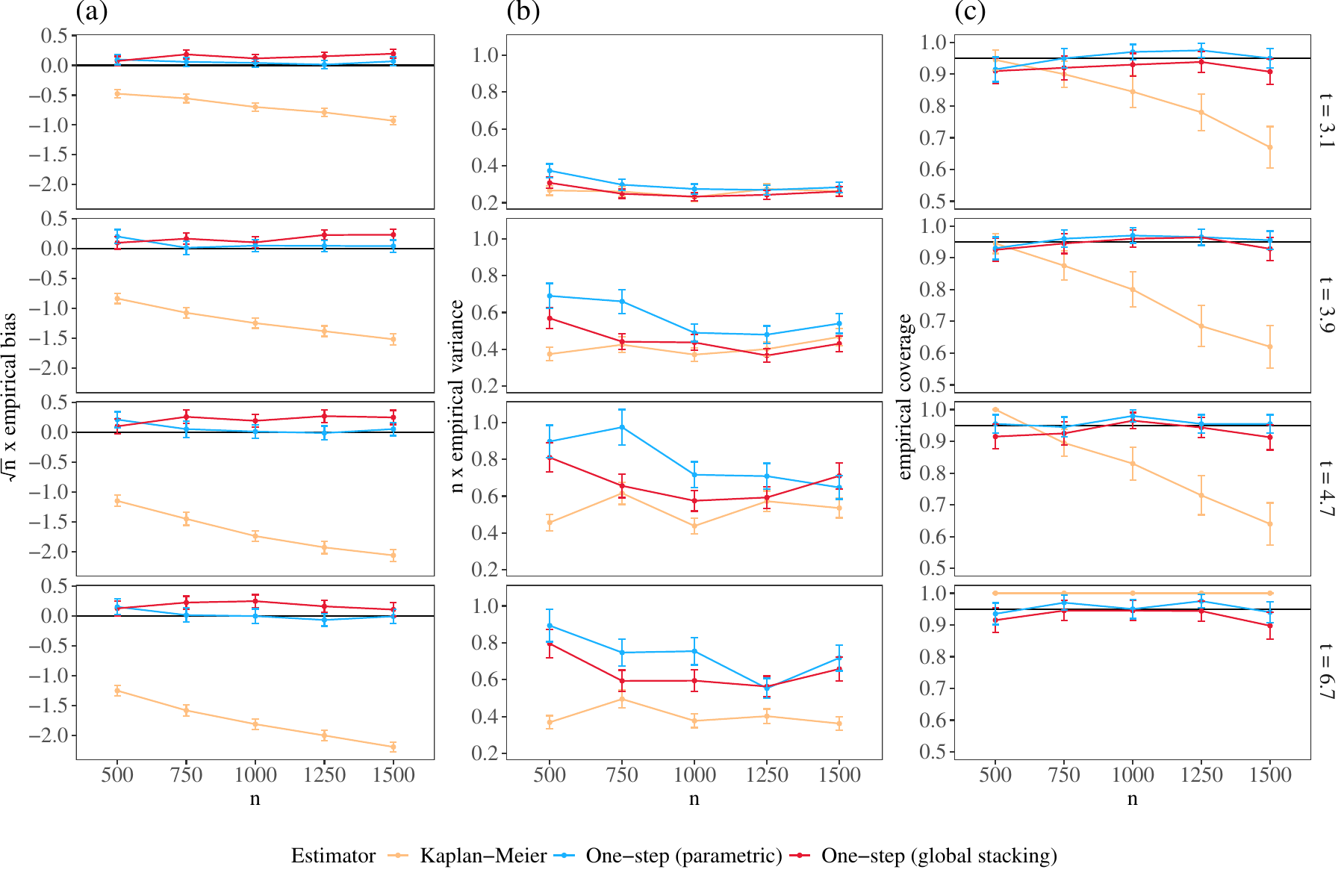}
		\end{center}
		\vspace{-0.5cm}
		\caption{Summary of inferential performance metrics for simulation scenario with 50\% censoring and 50\% truncation. (a) empirical bias scaled by $n^{\frac{1}{2}}$; (b) empirical variance scaled by $n$; (c) empirical coverage of nominal 95\% confidence intervals. Rows correspond to the four times at which the survival function was estimated. The colors denote different estimators, including the Kaplan-Meier estimator and the proposed one-step estimator with nuisance parameters estimated using either correctly specified parametric models or global survival stacking. Vertical bars represent 95\% confidence intervals accounting for Monte Carlo error.}
		\label{fig:C50T50}
	\end{figure}
	
	\clearpage
	
	\section{Concluding remarks}\label{s:concl}
	
	We have developed and studied debiased machine learning methods for making statistical inference on summaries of a counterfactual time-to-event distribution using left-truncated right-censored data. These methods require, as an intermediate step, estimation of various nuisance functions. However, in view of the influence function-based debiasing procedure employed, the use of machine learning techniques is allowed to estimate these nuisances, thereby reducing the risk of inconsistent estimation. In particular, through the use of adaptive ensemble learning, this fact renders more realistic the rate conditions imposed on the nuisance estimators.
	
	Identification of a summary of the time-to-event distribution requires that the set of available covariates be sufficiently rich to explain any dependence between the treatment allocation mechanism and counterfactual outcomes, between the truncation and event times, and between the censoring and event times. While in this work we have focused on baseline covariates exclusively, our methods could be extended to allow covariates recorded at (post-baseline) study entry to possibly inform the relationship between the censoring and event times. Identification also requires that the counterfactual time-to-event distribution itself be identified over a sufficient large portion of its support to allow computation of the summary of interest. Right censoring often precludes identification of the right tail of a time-to-event distribution, rendering unidentifiable summaries that depend on the right tail, such as moments of the time-to-event distribution. Interestingly, in some cases, left truncation can help restore the identifiability of this right tail; this occurs, for example, when left truncation arises due to cross-sectional sampling, and the censoring mechanism only acts on the portion of the event time under follow-up (i.e., from study entry and on) and is independent of the study entry time. Thus, in any given application, it is important to determine the extent to which the time-to-event support may be recovered in a given application, and to consider its implications on which summary can be identified.
	
	The two estimation procedures we derived exhibited some level of robustness to the estimation of involved nuisance functions. Specifically, they allowed a certain degree of inconsistent nuisance estimation under which the summary of interest is still estimated consistently. Of the estimators proposed, we noted that estimating equations-based estimator is qualitatively more robust than the one-step estimator --- this is an interesting example of a setting in which two constructive approaches for nonparametric inference, while equivalent when all nuisances are estimated sufficiently well, differ in behavior when that is not the case. Both procedures considered required consistent estimation of the target conditional time-to-event distribution and the observable conditional truncation distribution; in other words, neither exhibits robustness to inconsistent estimation of these nuisances. In future work, it is important to consider how the use of different parametrizations may lead to different --- and possibly more permissive --- robustness profiles. Additionally, while the robustness discussed here pertains to preservation of consistency, it may be fruitful to also consider how to achieve preservation of asymptotic linearity so that robust confidence intervals and $p$-values may also be constructed, along the lines of \cite{benkeser2017doubly}, for example. 
	
	While the class of summaries we considered in this article is broad, it does not include all functionals for which parametric-rate inference is possible. For example, some summaries that depend inextricably on the counterfactual time-to-event density function fall outside the class considered and appear more difficult to tackle in generality. Similarly, it is challenging to characterize inference for survival integrals for which the kernel $\varphi$ depends on the underlying distribution $P_{X,0}$. However, such survival integrals do arise in contemporary applications, and the developments provided here serve as important building blocks for the study of such integral estimands.
	
	\vspace{1in}
	
	\noindent\textbf{Acknowledgements.} This work was performed as part of the completed doctoral dissertation research of E. Morenz. Financial support was provided by the National Heart, Lung, and Blood Institute through grant R01-HL137808 and by the National Science Foundation Graduate Research Fellowship Program through grant DGE-2140004. The content is solely the responsibility of the authors and does not necessarily represent the official views of the funding agencies.

	\singlespacing
	\bibliographystyle{apa}
	\bibliography{library}
	
	\doublespacing
	
	\newpage
	\section*{Supplementary materials}
	
	\section*{Part A: proofs of theorems}
	\label{parta}
	To begin, we present certain key identities on the linearization of conditional survival integrals that will be used below.
	
	\begin{lemma}
		\label{lem:duham}
		For each $(a,z)\in\{0,1\}\times \mathcal{Z}$ and any $0\leq\alpha<\beta<\infty$, the following identities hold:
		\begin{align*}
			(a)\ &\int \varphi(y, z)(\tilde{F}_P-\tilde{F}_0)(dy \miid a, z) =-\int L_{P,\varphi}(y, a, z) \frac{\tilde{S}_0(y \miid a, z)}{\tilde{S}_P(y \miid a, z)}(\tilde{\Lambda}_P-\tilde{\Lambda}_0)(dy \miid a, z)\,;\\
			(b)\ &\int \frac{(\tilde S_P-\tilde{S}_0)(w\miid a,z)}{\tilde S_P(w\miid a,z)^2}G_P(dw\miid a,z)=-\int \gamma_{P,\natural}(y, a,z) \frac{\tilde S_0(y\miid a,z)}{\tilde S_P(y\miid a,z)}(\tilde \Lambda_P - \tilde \Lambda_0)(dy\miid a,z)\\
			(c)\ &\left|\int_\alpha^\beta \omega(y,z)\left[\frac{L_{P,\varphi}(y,a,z)}{\tilde S_P(y\miid a,z)}\tilde\Lambda_P(dy\miid a,z)-\frac{L_{0,\varphi}(y,a,z)}{\tilde S_0(y\miid a,z)}\tilde\Lambda_0(dy\miid a,z)\right]\right|\\
			&\leq\ 3v_z(\omega,\alpha,\beta)\sup_{y\in[\alpha,\beta]}\left|\frac{L_{P,\varphi}(y,a,z)}{\tilde{S}_P(y\miid a,z)}-\frac{L_{0,\varphi}(y,a,z)}{\tilde{S}_0(y\miid a,z)}\right|\\
			(d)\ &\left|\int_\alpha^\beta \omega(y,z)\left[\frac{\gamma_{P,\natural}(y,a,z)}{\tilde S_P(y\miid a,z)}\tilde\Lambda_P(dy\miid a,z)-\frac{\gamma_{0,\natural}(y,a,z)}{\tilde S_0(y\miid a,z)}\tilde\Lambda_0(dy\miid a,z)\right]\right|\\
			&\leq\ 3v_z(\omega,\alpha,\beta)\sup_{y\in[\alpha,\beta]}\left|\frac{\gamma_{P,\natural}(y,a,z)}{\tilde{S}_P(y\miid a,z)}-\frac{\gamma_{0,\natural}(y,a,z)}{\tilde{S}_0(y\miid a,z)}\right|\\
			&\hspace{0.2in}+\frac{3v_z(\omega,\alpha,\beta)}{S_0(\beta\miid a,z)^3}\sup_{y\in[\alpha,\beta]}\left|G_P(y\miid a,z)-G_0(y\miid a,z)\right|\\
			&\hspace{0.2in}+\frac{2v_z(\omega,\alpha,\beta)}{S_{P}(\beta\miid a,z)S_{0}(\beta\miid a,z)}\sup_{y\in[\alpha,\beta]}\left|\frac{1}{S_{P}(y\miid a,z)}-\frac{1}{S_0(y\miid a,z)}\right|,
		\end{align*}where $v_z(\omega,\alpha,\beta)$ denotes the maximum of the total variation and supremum norm of the function $y\mapsto \omega(y,z)$ over the interval $[\alpha,\beta]$.
	\end{lemma}
	
	\begin{proof}
		The Duhamel equation (Theorem 6 of~\citealp{gill1990survey}) indicates that
		\[
		(S-S_0)(y \miid a, z)= -S(y \miid a, z) \int_0^y \frac{S_0(u \miid a, z)}{S(u \miid a, z)}\left(\Lambda-\Lambda_0\right)(d u \miid a, z)
		\]
		for any two continuous survival functions $S$ and $S_0$ and their corresponding cumulative hazard function $\Lambda$ and $\Lambda_0$. The differential form of this equation is \[(F-F_0)(dy\miid a,z)=S_0(y\miid a,z)(\Lambda-\Lambda_0)(dy\miid a,z)-\int_0^y \frac{S_0(u \miid a, z)}{S(u \miid a, z)}\left(\Lambda-\Lambda_0\right)(d u \miid a, z) F(dy\miid a,z).
		\]
		The latter equation allows us to write
		\begin{align*}
			&\int \varphi(y, z)(\tilde{F}_P-\tilde F_0)(dy \miid a, z)\\
			&=\ \int \varphi(y, z)\tilde S_0(y\miid a,z)(\tilde \Lambda_P-\tilde\Lambda_0)(dy\miid a,z)\\
			&\hspace{0.3in}-\int \varphi(y, z)\int_0^y \frac{\tilde S_0(u \miid a, z)}{\tilde S_P(u \miid a, z)}(\tilde\Lambda_P-\tilde \Lambda_0)(d u \miid a, z) \tilde F_P(dy\miid a,z)\\
			&=\ \int\varphi(y,z)\tilde S_0(y\miid a,z)(\tilde \Lambda_P-\tilde\Lambda_0)(dy\miid a,z)\\
			&\hspace{0.3in}-\iint I(u\leq y) \varphi(y, z) \frac{\tilde S_0(u \miid a, z)}{\tilde S_P(u \miid a, z)} (\tilde\Lambda_P-\tilde\Lambda_0)(d u \miid a, z)\tilde F_P(dy\miid a,z) \\
			&=\ \int \tilde S_0(y\miid a,z)\left\{\varphi(y,z)-\frac{1}{\tilde S_P(y \miid a, z)} \int^\infty_y\varphi(u,z)\tilde F_P(du\miid a,z)\right\}(\tilde \Lambda_P-\tilde\Lambda_0)(dy\miid a,z)\\
			&=\ -\int L_{\varphi, P}(y \miid a, z) \frac{\tilde{S}_0(y \miid a, z)}{\tilde{S}_P(y \miid a, z)}(\tilde{\Lambda}_P-\tilde{\Lambda}_0)(dy \miid a, z)\ ,
		\end{align*}where we have used that \begin{align*}
			\int_y^\infty \varphi(u,z)\tilde{F}_P(du\miid a,z)\ &=\ -\int _y^\infty \varphi(u,z)\tilde{S}_P(du\miid a,z)\\
			&=\ -\left.\varphi(u,z)\tilde S_P(u\miid a,z)\right|_{y}^\infty+\int_y^\infty \tilde{S}_P(u\miid a,z)\varphi(du,z)\\
			&=\ \varphi(y,z)\tilde{S}_P(y\miid a,z)+L_{\varphi,P}(y,a,z)\ .
		\end{align*}
		and this establishes part (a). We again make use of the Duhamel equation and write that
		\begin{align*}
			&\int \frac{(\tilde S_P-\tilde{S}_0)(w\miid a,z)}{\tilde S_P(w\miid a,z)^2}G_P(dw\miid a,z)\\
			&=\ -\int \frac{1}{\tilde{S}_P(w\miid a,z)^2} \tilde{S}_P(w\miid a,z)\int _0^w \frac{\tilde{S}_0(u\miid a,z)}{\tilde{S}_P(u\miid a,z)}(\tilde{\Lambda}_P-\tilde{\Lambda}_0)(du\miid a,z)G_P(dw\miid a,z)\\
			&=\ -\int \left\{\int_u^\infty \frac{G_P(dw\miid a,z)}{\tilde{S}_P(w\miid a,z)}\right\}\frac{\tilde{S}_0(u\miid a,z)}{\tilde{S}_P(u\miid a,z)}(\tilde{\Lambda}_P-\tilde{\Lambda}_0)(du\miid a,z)\\
			&=\ -\int \gamma_{P,\natural}(u,a,z)\frac{\tilde{S}_0(u\miid a,z)}{\tilde{S}_P(u\miid a,z)}(\tilde \Lambda_P - \tilde \Lambda_0)(dy\miid a,z)\ ,
		\end{align*}which establishes part (b). Next, using the fact that \[\frac{L_\varphi(u)}{S(u)}\Lambda(du)=\frac{L_\varphi}{S}(du)+\varphi(du)\]for any survival function $S$, corresponding cumulative hazard function $\Lambda$ and $\varphi$--integrated form $L$, we have that \begin{align*}
			&\int_\alpha^\beta \omega(y,z)\left[\frac{L_{P,\varphi}(y,a,z)}{\tilde S_P(y\miid a,z)}\tilde\Lambda_P(dy\miid a,z)-\frac{L_{0,\varphi}(y,a,z)}{\tilde S_0(y\miid a,z)}\tilde\Lambda_0(dy\miid a,z)\right]\\
			&=\ \int_\alpha^\beta \omega(y,z)\left[\frac{L_{P,\varphi}}{\tilde S_P}(dy\miid a,z)-\frac{L_{0,\varphi}}{\tilde S_0}(dy\miid a,z)\right]\\
			&=\ \left.\omega(y,z)\left[\frac{L_{P,\varphi}(y,a,z)}{\tilde{S}_P(y\miid a,z)}-\frac{L_{0,\varphi}(y,a,z)}{\tilde{S}_0(y\miid a,z)}\right]\right|_{y=\alpha}^\beta+\int_\alpha ^\beta \left[\frac{L_{P,\varphi}(y,a,z)}{\tilde{S}_P(y\miid a,z)}-\frac{L_{0,\varphi}(y,a,z)}{\tilde{S}_0(y\miid a,z)}\right]\omega(dy,z)\ .
		\end{align*}This then implies that \begin{align*}
			&\left|\int_\alpha^\beta \omega(y,z)\left[\frac{L_{P,\varphi}(y,a,z)}{\tilde S_P(y\miid a,z)}\tilde\Lambda_P(dy\miid a,z)-\frac{L_{0,\varphi}(y,a,z)}{\tilde S_0(y\miid a,z)}\tilde\Lambda_0(dy\miid a,z)\right]\right|\\
			&\hspace{1in}\leq\ \left[2\sup_{y\in[\alpha,\beta]}|\omega(y,z)|+\int_\alpha^\beta|\omega(dy,z)|\right] \sup_{y\in[\alpha,\beta]}\left|\frac{L_{P,\varphi}(y,a,z)}{\tilde{S}_P(y\miid a,z)}-\frac{L_{0,\varphi}(y,a,z)}{\tilde{S}_0(y\miid a,z)}\right|.
		\end{align*}which implies the claimed inequality in (c). Finally, using integration by parts, we write \begin{align*}
			&\int_\alpha^\beta \omega(y,z)\left[\frac{\gamma_{P,\natural}(y,a,z)}{\tilde S_P(y\miid a,z)}\tilde\Lambda_P(dy\miid a,z)-\frac{\gamma_{0,\natural}(y,a,z)}{\tilde S_0(y\miid a,z)}\tilde\Lambda_0(dy\miid a,z)\right]\\
			&=\ \left.\omega(y,z)\left[\frac{\gamma_{P,\natural}(y,a,z)}{\tilde{S}_P(y\miid a,z)}-\frac{\gamma_{0,\natural}(y,a,z)}{\tilde{S}_0(y\miid a,z)}\right]\right|_{y=\alpha}^\beta\\
			&\hspace{.3in}+\int_\alpha ^\beta \left[\frac{\gamma_{P,\natural}(y,a,z)}{\tilde{S}_P(y\miid a,z)}-\frac{\gamma_{0,\natural}(y,a,z)}{\tilde{S}_0(y\miid a,z)}\right]\omega(dy,z)+\int_\alpha ^\beta\omega(y,z) \left[\frac{\gamma_{P,\natural}(dy,a,z)}{\tilde{S}_P(y\miid a,z)}-\frac{\gamma_{0,\natural}(dy,a,z)}{\tilde{S}_0(y\miid a,z)}\right]
		\end{align*}and furthermore expand \begin{align*}
			&\int_\alpha ^\beta\omega(y,z) \left[\frac{\gamma_{P,\natural}(dy,a,z)}{\tilde{S}_P(y\miid a,z)}-\frac{\gamma_{0,\natural}(dy,a,z)}{\tilde{S}_0(y\miid a,z)}\right]\ =\ \int_\alpha ^\beta\omega(y,z) \left[\frac{G_{P}(dy\miid a,z)}{\tilde{S}_P(y\miid a,z)^2}-\frac{G_{0}(dy\miid a,z)}{\tilde{S}_0(y\miid a,z)^2}\right]\\
			&=\ \int_\alpha ^\beta\omega(y,z) \left[\frac{1}{\tilde{S}_P(y\miid a,z)^2}-\frac{1}{\tilde{S}_0(y\miid a,z)^2}\right]G_{P}(dy\miid a,z)+\int_\alpha^\beta \frac{\omega(y,z)}{S_0(y\miid a,z)^2}(G_P-G_0)(dy\miid a,z)\\
			&=\ \int_\alpha ^\beta\omega(y,z) \left[\frac{1}{\tilde{S}_P(y\miid a,z)}-\frac{1}{\tilde{S}_0(y\miid a,z)}\right]\left[\frac{1}{\tilde{S}_P(y\miid a,z)}+\frac{1}{\tilde{S}_0(y\miid a,z)}\right]G_{P}(dy\miid a,z)\\
			&\hspace{0.25in}+\int_\alpha^\beta (G_P-G_0)(y\miid a,z)\left[\frac{1}{S_0(y\miid a,z)^2}\omega(dy,z)+\frac{2\omega(y,z)}{S_0(y\miid a,z)^3}F_0(dy\miid a,z)\right].
		\end{align*}This allows us to write \begin{align*}
			&\left|\int_\alpha^\beta \omega(y,z)\left[\frac{\gamma_{P,\natural}(y,a,z)}{\tilde S_P(y\miid a,z)}\tilde\Lambda_P(dy\miid a,z)-\frac{\gamma_{0,\natural}(y,a,z)}{\tilde S_0(y\miid a,z)}\tilde\Lambda_0(dy\miid a,z)\right]\right|\\
			&\leq\ \left[2\sup_{y\in[\alpha,\beta]}|\omega(y,z)|+\int_\alpha^\beta|\omega(dy,z)|\right]\sup_{y\in[\alpha,\beta]}\left|\frac{\gamma_{P,\natural}(y,a,z)}{\tilde S_P(y\miid a,z)}-\frac{\gamma_{0,\natural}(y,a,z)}{\tilde S_0(y\miid a,z)}\right|\\
			&\hspace{.25in}+\sup_{y\in[\alpha,\beta]}|\omega(y,z)|\left[\frac{1}{\tilde{S}_P(\beta\miid a,z)}+\frac{1}{\tilde{S}_0(\beta\miid a,z)}\right]\sup_{y\in[\alpha,\beta]}\left|\frac{1}{\tilde{S}_P(y\miid a,z)}-\frac{1}{\tilde{S}_0(y\miid a,z)}\right|\\
			&\hspace{.25in}+\left[\frac{\int_\alpha^\beta |\omega(dy,z)|}{S_0(\beta\miid a,z)^2}+\frac{2\sup_{y\in[\alpha,\beta]}|\omega(y,z)|}{S_0(\beta\miid a,z)^3}\right]\sup_{y\in[\alpha,\beta]}|G_P(y\miid a,z)-G_0(y\miid a,z)|\ ,
		\end{align*}thus implying the claimed inequality.
	\end{proof}
	
	\subsection*{Proof of Theorem~\ref{thm:if}}
	
	Let $P\in\mathscr{M}$ be given, and take $\{P_{\epsilon} : |\epsilon| \leq \delta\}$ to be a suitably smooth and bounded (i.e., Hellinger-differentiable) path with $P_{\epsilon = 0} = P$ and score  for $\epsilon$ at $\epsilon = 0$ given by $h\in L_2^0(P)$, and let $h=h_1+h_2$ denote the $L_2^0(P)$--unique decomposition of $h$ for which  $o\mapsto h_1(y,\delta,w,a,z)$ and $o\mapsto h_2(a,z)$  are such that, $P$--almost surely, $E_P[h_1(O)\miid A,Z]=0$, $E_P[h_2(O)]=0$, $var_P[h_1(O)\miid A,Z]<\infty$, and $var_P[h_2(O)]<\infty$. We wish to compute the pathwise derivative \begin{align*}
		\left.\frac{\partial}{\partial\epsilon}\Psi(P_\epsilon)\right|_{\epsilon=0}\ &=\ \left.\frac{\partial}{\partial\epsilon}\iint \varphi(t,z)\tilde{F}_{\epsilon}(dt\miid a_0,z)\tilde{H}_{\epsilon}(dz)\right|_{\epsilon=0}\\
		&=\ \left.\frac{\partial}{\partial\epsilon}\iiint \varphi(t,z)\tilde{F}_\epsilon(dt\miid a_0,z)\bar{\gamma}_\epsilon(a,z)J_\epsilon(da,dz)\right|_{\epsilon=0}
	\end{align*}
	where here and below we use the shorthand notation $A_\epsilon$ to refer to $A_{P_\epsilon}$ for any relevant quantity $A_P$ indexed by $P$. Furthermore, under mild regularity conditions allowing interchange of integral and derivative operations, this pathwise derivative can be decomposed as $\text{(1)}+\text{(2)}+\text{(3)}$ with \begin{align*}
		\text{(1)}\ &=\ \iint \left.\frac{\partial}{\partial\epsilon}\int \varphi(t,z)\tilde{F}_\epsilon(dt\miid a_0,z)\right|_{\epsilon=0}\bar{\gamma}_P(a,z)J_P(da,dz)\\
		\text{(2)}\ &=\ \iiint \varphi(t,z)\tilde{F}_P(dt\miid a_0,z)\left.\frac{\partial}{\partial\epsilon}\bar{\gamma}_\epsilon(a,z)\right|_{\epsilon=0}J_P(da,dz)\\
		\text{(3)}\ &=\  \iiint \varphi(t,z)\tilde{F}_P(dt\miid a_0,z)\bar{\gamma}_P(a,z)\left.\frac{\partial}{\partial\epsilon}J_\epsilon(da,dz)\right|_{\epsilon=0}\,.
	\end{align*}Below, we  study each of these summands separately.

	By integration by parts, we first note that $\int \varphi(t,z)\tilde{F}_\epsilon (dt\miid a_0,z)=\varphi(0,z)+\int \tilde{S}_\epsilon(t\miid a_0,z)\varphi(dt,z)$, and so, we can equivalently write \[\text{(1)}=\iiint \left.\frac{\partial}{\partial\epsilon}\tilde{S}_\epsilon (t\miid a_0,z)\right|_{\epsilon=0}\varphi(dt,z)\bar{\gamma}_P(a,z)J_P(da,dz)\ .\] To compute the pathwise derivatives of $\epsilon\mapsto \tilde{S}_\epsilon(t\miid a_0,z)$, we first consider the pathwise derivative of $\epsilon\mapsto \tilde{\Lambda}_\epsilon(t\miid a_0,z)$, where $\tilde{\Lambda}_\epsilon$ is the cumulative hazard function corresponding to $\tilde{S}_\epsilon$, defined as $\tilde{\Lambda}_\epsilon(t\miid a_0,z):=\int_0^t R_{\epsilon}(u\miid a_0,z)^{-1}F_{1,\epsilon}(du\miid a_0,z)$. We can show that \begin{align*}
		&\left.\frac{\partial}{\partial\epsilon}\tilde{\Lambda}_\epsilon(t\miid a_0,z)\right|_{\epsilon=0}\ =\ \int_0^t\frac{\left.\frac{\partial}{\partial\epsilon}F_{1,\epsilon}(du\miid a_0,z)\right|_{\epsilon=0}}{R_P(u\miid a_0,z)}-\int_0^t \frac{\left.\frac{\partial}{\partial\epsilon}R_\epsilon(u\miid a_0,z)\right|_{\epsilon=0}}{R_P(u\miid a_0,z)^2}F_{1,P}(du\miid a_0,z)\\
		&\hspace{0.2in}=\ E_P\left[\left\{\frac{\Delta I(Y\leq t)}{R_P(Y\miid a_0,z)}-\int_0^t \frac{I(W\leq u\leq Y)}{R_P(u\miid a_0,z)}\tilde{\Lambda}_{P}(du\miid a_0,z)\right\}h_1(O)\,\middle|\,A=a_0,Z=z\right]\\
		&\hspace{0.2in}=\ E_P\left[-\phi_{\text{KM},P}((u,a,z)\mapsto I(u\leq t))(O)h_1(O)\,\middle|\,A=a_0,Z=z\right]
	\end{align*} by first showing that \begin{align*}
		&\left.\frac{\partial}{\partial\epsilon}F_{1,\epsilon}(u\miid a_0,z)\right|_{\epsilon=0}\ =\ \iiint \frac{\delta I(y\leq u)}{R_P(y\miid a_0,z)}h_1(y,\delta,w,a_0,z)P(dy,d\delta,dw\miid a_0,z)\end{align*}
	and
	\begin{align*}
		\left.\frac{\partial}{\partial\epsilon}R_{\epsilon}(u\miid a_0,z)\right|_{\epsilon=0}\ &=\ \iiint \int_0^t \frac{I(w\leq u\leq y)}{R_P(u\miid a_0,z)}\tilde{\Lambda}_P(du\miid a_0,z)h_1(y,\delta,w,a_0,z)P(dy,d\delta,dw\miid a_0,z).
	\end{align*}This then implies, using Theorem 8 of \cite{gill1990survey}, we that
	\begin{align*}
		\left.\frac{\partial}{\partial\epsilon}\tilde{S}_\epsilon(t\miid a_0,z)\right|_{\epsilon=0}\ &=\ -\tilde{S}_P(t\miid a_0,z)\left.\frac{\partial}{\partial\epsilon}\tilde{\Lambda}_\epsilon(t\miid a_0,z)\right|_{\epsilon=0}\\
		&=\ E_P\left[-\tilde{S}_P(t\miid a_0,z)\phi_{\text{KM},P}((u,a,z)\mapsto I(u\leq t))(O)h_1(O)\,\middle|\,A=a_0,Z=z\right].
	\end{align*}In particular, this allows us to compute \begin{align*}
		\int \left.\frac{\partial}{\partial\epsilon}\tilde{S}_\epsilon(t\miid a_0,z)\right|_{\epsilon=0} \varphi(dt,z)\ &=\ E_P\left[-\phi_{\text{KM},P}(L_{P,\varphi})(O)h_1(O)\,\middle|\,A=a_0,Z=z\right]\\
		&=\ E_P\left[-\frac{I(A=a_0)}{P(A=a_0\miid Z=z)}\phi_{\text{KM},P}(L_{P,\varphi})(O)h_1(O)\,\middle|\, Z=z\right]
	\end{align*}and therefore, we have that \begin{align*}
		\text{(1)}\ &=\ \iint E_P\left[-\frac{I(A=a_0)}{P(A=a_0\miid Z=z)}\phi_{\text{KM},P}(L_{P,\varphi})(O)h_1(O)\,\middle|\, Z=z\right] \bar\gamma_P(a,z)J_P(da,dz)\\
		&=\ E_P\left[E_P\left[-\frac{I(A=a_0)}{P(A=a_0\miid Z=z)}\phi_{\text{KM},P}(L_{P,\varphi})(O)h_1(O)\,\middle|\, Z\right] \bar\gamma_P(A,Z)\right]\\
		&=\ E_P\left[-\frac{I(A=a_0)}{P(A=a_0\miid Z=z)}\bar\gamma_P(Z)\phi_{\text{KM},P}(L_{P,\varphi})(O)h_1(O)\right],
	\end{align*}where we have defined $\bar\gamma_P:z\mapsto E_P[\bar\gamma_P(A,Z)\miid Z=z]$. From the result $E_P[\phi_{\text{KM},P}(L_{P,\varphi})(O)\miid A=a_0,Z]=0$ $P$--almost surely, we can write \[\text{(1)}=E_P\left[-\frac{I(A=a_0)}{P(A=a_0\miid Z=z)}\bar{\gamma}_P(Z)\phi_{\text{KM},P}(L_{P,\varphi})(O)h(O)\right]\] in view of the fact that $h_2(O)$ is only a function of $(A,Z)$.

	We now turn to computing the pathwise derivative of $\epsilon\mapsto \bar\gamma_\epsilon(a,z)$ at $\epsilon=0$, which is critical for computing $\text{(2)}$. We have that \begin{align*}
		&\left.\frac{\partial}{\partial\epsilon}\bar{\gamma}_\epsilon(a,z)\right|_{\epsilon=0}\ =\ \frac{1}{\gamma_P}\left[\left.\frac{\partial}{\partial\epsilon}\gamma_\epsilon(a,w)\right|_{\epsilon=0}-\bar{\gamma}_P(a,z)\left.\frac{\partial}{\partial\epsilon}\gamma_\epsilon\right|_{\epsilon=0}\right]\\
		&=\ \left.\frac{\partial}{\partial\epsilon}\frac{\gamma_\epsilon(a,w)}{\gamma_P}\right|_{\epsilon=0}-\frac{\bar{\gamma}_P(a,z)}{\gamma_P}\left\{\iint \left.\frac{\partial}{\partial\epsilon}\gamma_\epsilon(\bar{a},\bar{z})\right|_{\epsilon=0}J_P(d\bar{a},d\bar{z})+\iint \gamma_P(\bar{a},\bar{z})\left.\frac{\partial}{\partial\epsilon}J_\epsilon(d\bar{a},d\bar{z})\right|_{\epsilon=0}
		\right\}
	\end{align*}using that $\gamma_\epsilon=\iint \gamma_\epsilon(\bar{a},\bar{z})J_\epsilon(d\bar{a},d\bar{z})$, so that we can decompose $\text{(2)}=\text{(2a)}+\text{(2b)}+\text{(2c)}$ with \begin{align*}
		\text{(2a)}\ &=\ \frac{1}{\gamma_P}\iint \mu_{P,\varphi}(z) \left.\frac{\partial}{\partial\epsilon}\gamma_\epsilon(a,z)\right|_{\epsilon=0}J_P(da,dz)\\
		\text{(2b)}\ &=\ -\Psi(P)\frac{1}{\gamma_P}\iint \left.\frac{\partial}{\partial\epsilon}\gamma_\epsilon(a,z)\right|_{\epsilon=0}J_P(da,dz)\\
		\text{(2c)}\ &=\ -\Psi(P)\frac{1}{\gamma_P}\iint \gamma_P(a,z)h_2(a,z)J_P(da,dz)\ .
	\end{align*} In particular we note that \[\text{(2a)}+\text{(2b)}=\iint \xi_P(z)\left.\frac{\partial}{\partial\epsilon}\gamma_\epsilon(a,z)\right|_{\epsilon=0}J_P(da,dz)\]with $\xi_P(z):=\gamma_P^{-1}\{\mu_{P,\varphi}(z)-\Psi(P)\}$. This expression involves the pathwise derivative of $\epsilon\mapsto \gamma_\epsilon(a,w)$ at $\epsilon=0$, which we can be computed as \begin{align*}
		&\left.\frac{\partial}{\partial\epsilon}\gamma_\epsilon(a,z)\right|_{\epsilon=0}\ =\ \left.\frac{\partial}{\partial\epsilon}\int \frac{G_\epsilon(dw\miid a,z)}{\tilde{S}_\epsilon(w\miid a,z)}\right|_{\epsilon=0}\\
		&=\ \int \frac{1}{\tilde{S}_P(w\miid a,z)}\left.\frac{\partial}{\partial\epsilon}G_\epsilon(dw\miid a,z)\right|_{\epsilon=0}-\int \left.\frac{\partial}{\partial\epsilon}\tilde{S}_\epsilon(w\miid a,z)\right|_{\epsilon=0}\frac{G(dw\miid a,z)}{\tilde{S}_P(w\miid a,z)^2}\\
		&=\ \iiint \frac{1}{\tilde{S}_P(w\miid a,z)}\left.\frac{\partial}{\partial\epsilon}P_\epsilon(dy,d\delta,dw\miid a,z)\right|_{\epsilon=0}+\int \left.\frac{\partial}{\partial\epsilon}\tilde{\Lambda}_\epsilon(w\miid a_0,z)\right|_{\epsilon=0} \frac{G(dw\miid a,z)}{\tilde{S}_P(w\miid a,z)}\\
		&=\ \iiint \frac{1}{\tilde{S}_P(w\miid a,z)}h_1(o)P(dy,d\delta,dw\miid a,z)\\
		&\hspace{.5in}+\int E_P\left[-\phi_{\text{KM},P}((u,a,z)\mapsto I(u\leq w))(O)h_1(O)\,\middle|\,A=a,Z=z\right]\frac{G(dw\miid a,z)}{\tilde{S}_P(w\miid a,z)}\\
		&=\ E_P\left[\left\{\frac{1}{\tilde{S}_P(W\miid a,z)}-\gamma_P(a,z)-\phi_{\text{KM},P}(\gamma_{P,\natural})(O)\right\}h(O)\,\middle|\,A=a,Z=z\right].
	\end{align*}where we re-centered the first summand by $\gamma_P(a,z)=\int \tilde{S}_P(w\miid a,z)^{-1})G_P(dw\miid a,z)$ in the penultimate step, which then allowed $h_1$ to be replaced by $h$ since $h_2(O)$ only depends on $(A,Z)$. We can therefore write that \[
	\text{(2a)}+\text{(2b)}=E_P\left[\xi_P(Z)\left\{\frac{1}{\tilde{S}_P(W\miid A,Z)}-\gamma_P(A,Z)-\phi_{\text{KM},P}(\gamma_{P,\natural})(O)\right\}h(O)\right].
	\]Using the fact that $h_1(O)$ has mean zero conditional on $(A,Z)$ and that $E_P[\bar\gamma_P(A,Z)]=1$, we can write that \[\text{(2c)}=-\Psi(P)\iint \bar\gamma_P(a,z)h_2(a,z)J_P(da,dz)=-E_P\left[\Psi(P)\left\{\bar\gamma_P(A,Z)-1\right\}h(O)\right],\]from which we conclude that \begin{align*}
		\text{(2)}\ &=\ E_P\left[\left[\xi_P(Z)\left\{\frac{1}{\tilde{S}_P(W\miid A,Z)}-\gamma_P(A,Z)-\phi_{\text{KM},P}(\gamma_{P,\natural})(O)\right\}-\Psi(P)\left\{\bar\gamma_P(A,Z)-1\right\}\right]h(O)\right]\\
		&=\ E_P\left[\left[\xi_P(Z)\left\{\frac{1}{\tilde{S}_P(W\miid A,Z)}-\gamma_P(A,Z)-\phi_{\text{KM},P}(\gamma_{P,\natural})(O)\right\}-\Psi(P)\left\{\bar\gamma_P(A,Z)-1\right\}\right]h(O)\right]
	\end{align*}
	
	The final term to compute, $\text{(3)}$, has a simple form. Indeed, using a similar argument as above, it can be written as
	\begin{align*}
		\text{(3)}\ &=\ \iint \mu_{P,\varphi}(z)\bar\gamma_P(a,z)h_2(a,z)J_P(da,dz)\\
		&=\ E_P\left[\mu_{P,\varphi}(Z)\bar\gamma_P(A,Z)h_2(A,Z)\right]\ =\ E_P\left[\left\{\mu_{P,\varphi}(Z)\bar\gamma_P(A,Z)-\Psi(P)\right\}h(O)\right].
	\end{align*}
	Adding the expressions derived for each of $\text{(1)}$, $\text{(2)}$ and $\text{(3)}$, as claimed, we find that the pathwise derivative of $\epsilon\mapsto \Psi(P_\epsilon)$ at $\epsilon=0$ is given by 
	\[E_P\left[\left[-\frac{I(A=a_0)\bar{\gamma}_P(Z)}{P(A=a_0\miid Z=z)}\phi_{\text{KM},P}(L_{P,\varphi})(O)+\xi_P(Z)\left\{\frac{1}{\tilde{S}_P(W\miid A,Z)}-\phi_{\text{KM},P}(\gamma_{P,\natural})(O)\right\}\right]h(O)\right],\]
	which is simply $E_P[\phi_P(O)h(O)]$ with $\phi_P$ as defined in the main text. This establishes the pathwise differentiability of $P\mapsto \Psi(P)$ relative to a nonparametric model as well as the fact that $\phi_P$ is the nonparametric efficient influence function of $\Psi$ at $P$.
	
	We now study the linearization of $\Psi$ around $P_0$ based on $\phi_P$. Specifically, we derive the form of the remainder term $R(P,P_0):=\Psi(P)-\Psi(P_0)-(P-P_0)\phi_P$ from this linearization. We begin by decomposing the difference $\Psi(P) -\Psi(P_0)$ as the sum $\text{(D1)}+\text{(D2)}+\text{(D3)}+\text{(D4)}$ with
	\begin{align*}
		\text{(D1)}\ &=\ \iint \varphi(y,z) (\tilde F_P-F_0)(dy\miid a_0,z) \tilde H_0(dz)\\
		\text{(D2)}\ &=\ \iint \xi_P(z) \left\{\gamma_P(a,z) - \gamma_0(a,z)\right\} J_0(da,dz)\\
		\text{(D3)}\ &=\ \frac{\gamma_0-\gamma_P}{\gamma_0}\iint \xi_P(z) \gamma_0(a,z) J_0(da,dz)\\
		\text{(D4)}\ &=\ \iint \xi_P(z) \gamma_P(a,z) (J_P-J_0)(da,dz)\ .
	\end{align*} 
	First, using part (a) of Lemma~\ref{lem:duham}, we can write
	\begin{align*}
		\text{(D1)}=-\iint L_{P,\varphi}(y, a_0, z) \frac{\tilde{S}_0(y \miid a_0, z)}{\tilde{S}_P(y \miid a_0, z)}(\tilde{\Lambda}_P-\tilde{\Lambda}_0)(dy \miid a_0, z) \tilde{H}_0(d z)\ .
	\end{align*}
	Next, we decompose $\text{(D2)}$ as the sum $\text{(D2a)}+\text{(D2b)}+\text{(D2c)}+\text{(D2d)}$ with
	\begin{align*}
		\text{(D2a)}&:=- \iiint \xi_P(z)\left\{\tilde S_P(w\miid a,z) - \tilde S_0(w\miid a,z)\right\} \frac{G_P(dw\miid a,z)}{\tilde S_P(w\miid a,z)^2}J_0(da,dz)\\
		\text{(D2b)}&:=\iiint \frac{\xi_P(z)}{\tilde S_P(w\miid a,z)} \left(G_P - G_0\right)(dw\miid a,z) J_0(da,dz)\\
		\text{(D2c)}&:=\iiint \xi_P(z)\left\{\frac{1}{\tilde S_0(w\miid a,z)} - \frac{1}{\tilde S_P(w\miid a,z)}\right\}\left(G_P-G_0\right)(dw\miid a,z) J_0(da,dz)\\
		\text{(D2d)}&:=\iiint \left\{\frac{\xi_P(z)}{\tilde S_0(w\miid a,z)} - \frac{\xi_P(z)}{\tilde S_P(w\miid a,z)}\right\}\left\{\tilde S_0(w\miid a,z) -\tilde S_P(w\miid a,z)\right\}\frac{G_P(dw\miid a,z)}{\tilde S_P(w\miid a,z)}J_0(da,dz)
	\end{align*}
	Using part (b) of  Lemma~\ref{lem:duham}, we can rewrite 
	\begin{align*}
		\text{(D2a)} =\iint \xi_P(z) \left\{\int  \gamma_{P,\natural}(y, a,z) \frac{\tilde S_0(y\miid a,z)}{\tilde S(y\miid a,z)}(\tilde \Lambda - \tilde \Lambda_0)(dy\miid a,z)\right\}J_0(da,dz)\ .
	\end{align*}
	We then note that we can decompose $\text{(D3)}$ as the sum $\text{(D3a)}+\text{(D3b)}$ with
	\begin{align*}
		\text{(D3a)}&:=\frac{\gamma_P - \gamma_0}{\gamma_0}\iint \xi_P(z)\left[ \gamma_P(a,z) - \gamma_0(a,z)\right]  J_0(da,dz)\\
		\text{(D3b)}&:=\frac{\gamma_P - \gamma_0}{\gamma_0}\iint \xi_P(z) \gamma_P(a,z)  \left(J_P-J_0\right)(da,dz)
	\end{align*}using the fact that $\iint \xi_P(z)\gamma_P(a,z)J_P(da,dz)=0$.
	
	We now compute the linear term $(P-P_0)\phi_P=-P_0\phi_P$, which can be decomposed as the sum $\text{(L1)}+\text{(L2)}+\text{(L3)}+\text{(L4)}$, where we define
	\begin{align*} 
		\text{(L1)}&:=-\iint \frac{\pi_0(z)}{\pi(z)}\frac{\bar\gamma_P(a_0,z)}{\bar\gamma_0(a_0,z)}\frac{\tilde R_0(y,a_0,z)}{\tilde R_P(y,a_0,z)} L_{P,\varphi}(y,a_0,z)(\tilde \Lambda_P-\tilde\Lambda_0)(dy\miid a_0,z)\tilde H_0(dz)\\
		\text{(L2)}&:=\iiint \xi_P(z)\frac{\tilde R_0(y,a,z)}{\tilde R_P(y,a,z)} \gamma_{P,\natural}(y,a,z)(\tilde \Lambda_P-\tilde\Lambda_0)(dy\miid a,z)J_0(da,dz)\\
		\text{(L3)}&:=\iiint \frac{\xi_P(z)}{\tilde S_P(w\miid a,z)} (G_P-G_0)(dw\miid a,z) J_0(da,dz)\\
		\text{(L4)}&:=\iint \xi_P(z) \gamma_P(a,z) (J_P-J_0)(da,dz)\ .
	\end{align*}
	We now scrutinize the terms appearing in $R(P,P_0)=\text{(D1)}+\text{(D2)}+\text{(D3)}+\text{(D4)}-\text{(L1)}-\text{(L2)}-\text{(L3)}-\text{(L4)}$. First, we observe that $\text{(D4)} -\text{(L4)} = 0$ and $\text{(D2b)}-\text{(L3)} = 0$. Next, we use $\nu_{P}(y,a,z):=\frac{\tilde S_P(y\miid a,z)}{\tilde R_P(y,a,z)}$ for simplicity and we note that\begin{align*}
		&\text{(D1)}-\text{(L1)}\\
		&=\ \iint L_{P,\varphi}(y,a_0,z)\left\{\frac{\pi_0(z)}{\pi(z)}\frac{\bar\gamma_P(a_0,z)}{\bar\gamma_0(a_0,z)}\frac{\tilde R_0(y,a_0,z)}{\tilde R_P(y,a_0,z)}-\frac{\tilde{S}_0(y\miid a_0,z)}{\tilde{S}_P(y\miid a_0,z)}\right\}(\tilde\Lambda_P-\tilde\Lambda_0)(dy\miid a_0,z)\tilde H_0(dz)\\
		&=\ \iint L_{P,\varphi}(y,a_0,z)\left\{\frac{\pi_0(z)}{\pi(z)}\frac{\bar\gamma_P(a_0,z)}{\bar\gamma_0(a_0,z)}\frac{\nu_P(y,a_0,z)}{\nu_0(y,a_0,z)}-1\right\}\frac{\tilde{S}_0(y\miid a_0,z)}{\tilde{S}_P(y\miid a_0,z)}(\tilde\Lambda_P-\tilde\Lambda_0)(dy\miid a_0,z)\tilde H_0(dz)\\
		&=\ \iint L_{P,\varphi}(y,a_0,z)\left\{\frac{\pi_0(z)}{\pi(z)}\frac{\bar\gamma_P(a_0,z)}{\bar\gamma_0(a_0,z)}\frac{\nu_P(y,a_0,z)}{\nu_0(y,a_0,z)}-1\right\}\left(\frac{\tilde S_P}{\tilde S_0}-1\right)(dy\miid a_0,z)\tilde H_0(dz)\\
		&=\,R_1(P,P_0)\, ,
	\end{align*}where we used the fact that \[
	\frac{\tilde S_0(y\miid a,z)}{\tilde S_P(y\miid a,z)}(\tilde \Lambda_P -\tilde \Lambda_0)(dy\miid a,z) = \left(\frac{\tilde S_0}{\tilde S_P}-1\right)(dy\miid a,z)\ ,
	\]which is a consequence of the Duhamel equation in Theorem 6 of \cite{gill1990survey}. Using the same argument, we note that\begin{align*}
		&\text{(D2a)}-\text{(L2)}\\
		&=\ \iiint \xi_P(z)\gamma_{P,\natural}(y,a,z)\left\{\frac{\tilde{S}_0(y\miid a,z)}{\tilde{S}_P(y\miid a,z)}-\frac{\tilde{R}_0(y\miid a,z)}{\tilde{R}_P(y\miid a,z)}\right\}(\tilde\Lambda_P-\tilde\Lambda_0)(dy\miid a,z)J_0(da,dz)\\
		&=\ \iiint \xi_P(z)\gamma_{P,\natural}(y,a,z)\left\{1-\frac{\nu_P(y\miid a,z)}{\nu_0(y\miid a,z)}\right\}\frac{\tilde{S}_0(y\miid a,z)}{\tilde{S}_P(y\miid a,z)}(\tilde\Lambda_P-\tilde\Lambda_0)(dy\miid a,z)J_0(da,dz)\\
		&=\ \iiint \xi_P(z)\gamma_{P,\natural}(y,a,z)\left\{1-\frac{\nu_P(y\miid a,z)}{\nu_0(y\miid a,z)}\right\}\left(\frac{\tilde{S}_P}{\tilde{S}_0}-1\right)(dy\miid a,z)J_0(da,dz)\, =\, R_2(P,P_0)\, .
	\end{align*}We also note that $\text{(D3)}=\text{(D3a)}+\text{(D3b)}=R_4(P,P_0)$ and that $\text{(D2c)}+\text{(D2d)}=R_3(P,P_0)$. As such, we find that the remainder from the linear approximation of $\Psi(P)-\Psi(P_0)$ by $(P-P_0)\phi_P$ is given by $\{\text{(D1)}+\text{(D1)}+\text{(D1)}+\text{(D1)}\}-\{\text{(L1)}+\text{(L2)}+\text{(L3)}+\text{(L4)}\}$, and this quantity coincides precisely with the form of $R(P,P_0)$ given in the theorem.
	
	\subsection*{Proof of Theorem~\ref{thm:asymp_point}}
	
	We begin by studying $\psi_n^*$.  By Theorem~\ref{thm:if}, we note that $\psi_{\eta_{n,k}} - \psi_0= -P_0\phi_{\eta_{n,k}}+R(\eta_{n,k},\eta_0)$ for each $k=1,2,\ldots,K$, and so, denoting $\bar{\phi}_\infty:=\phi_\infty-P_0\phi_\infty$, we can write \begin{align*}
		\psi_n^*-\psi_0\ &=\ \frac{1}{K}\sum_{k=1}^{K}\left(\psi_{\eta_{n,k}}+\mathbb{P}_{n,k}\phi_{\eta_{n,k}}-\psi_0\right)\\
		&=\  \frac{1}{K}\sum_{k=1}^{K}\left\{(\mathbb{P}_{n,k}-P_0)\phi_\infty+(\mathbb{P}_{n,k}-P_0)(\phi_{\eta_{n,k}}-\phi_\infty)+R(\eta_{n,k},\eta_0)\right\}\\
		&=\ P_n\bar{\phi}_\infty+r_{an}+r_{bn}+r_{cn}\ ,
	\end{align*}where we have defined $r_{an}:=\frac{1}{K}\sum_{k=1}^{K}(\mathbb{P}_{n,k}-\mathbb{P}_n)\phi_\infty$, $r_{bn}:=\frac{1}{K}\sum_{k=1}^{K}(\mathbb{P}_{n,k}-P_0)(\phi_{\eta_{n,k}}-\phi_\infty)$ and $r_{cn}:=\frac{1}{K}\sum_{k=1}^{K}R(\eta_{k,n},\eta_0)$.
	
	We first note that $P_n\bar{\phi}_\infty=o_P(n^{-\frac{1}{2}})$ under Condition \ref{ass:pro_pos} in view of the fact that $n^{\frac{1}{2}}P_n\bar{\phi}_\infty$ tends to a normal random variable with mean zero and variance $P_0\bar{\phi}_\infty^2<\infty$.
	
	Next, we show that $r_{an}=o_P(n^{-\frac{1}{2}})$. To see this, we note that since we can always find $n_1,n_2,\ldots,n_K$ such that $|n_kK-n|\leq K$ for each $k=1,2,\ldots,n$, we have that \begin{align*}
		|r_{an}|\ =\ \left|\sum_{k=1}^{K}\sum_{i\in V_k}\left(\frac{1}{Kn_k}-\frac{1}{n}\right)\phi_\infty(O_i)\right|\ &\leq\ \max_k \left|\frac{n}{Kn_k}-1\right|\frac{1}{n}\sum_{i=1}^{n}|\phi_\infty(O_i)|\\
		&\leq\ \left(\frac{K}{n-K}\right)\frac{1}{n}\sum_{i=1}^{n}|\phi_\infty(O_i)|\ =\ o_P(n^{-1})\\
		&=\ o_P(n^{-\frac{1}{2}})
	\end{align*}in view of the fact that $var_0\{\phi_\infty(O)\}<\infty$. Next, we show that $r_{bn}=o_P(n^{-\frac{1}{2}})$ under Conditions~\ref{ass:pro_conv}--\ref{ass:pro_pos}. Denoting by $\mathcal{D}_k:=\cup_{j\notin k}\mathcal{V}_j$ the portion of the dataset used to construct $\eta_{n,k}$ and writing $A_{n,k}:=n^{\frac{1}{2}}(\mathbb{P}_{n,k}-P_0)(\phi_{\eta_{n,k}}-\phi_\infty)$, by Chebyshev's inequality, for any $\varepsilon>0$, we have that \begin{align*}
		P_0\left(\left|A_{n,k}\right|>\varepsilon\,|\,\mathcal{D}_k\right)\ \leq\ \frac{var_0\left(A_{n,k}\,|\,\mathcal{D}_{n,k}\right)}{\varepsilon^2}\ \leq\ \frac{n}{n_k}\cdot\frac{P_0\left(\phi_{n,k}-\phi_\infty\right)^2}{\varepsilon^2}\ \longrightarrow\ 0
	\end{align*}provided $P_0\left(\phi_{n,k}-\phi_\infty\right)^2$ tends to zero in probability and using that $n/n_k\rightarrow K<\infty$. By the Bounded Convergence Theorem, it follows that $A_{n,k}$ tends to zero in probability since we can write  \[P_0\left(|A_{n,k}|>\varepsilon\right)=E_0\left[P_0\left(\left|A_{n,k}\right|>\varepsilon\,|\,\mathcal{D}_k\right)\right]\longrightarrow 0\]  for each $\varepsilon>0$. This implies that $r_{bn}=n^{-\frac{1}{2}}\frac{1}{K}\sum_{k=1}^{K}A_{n,k}=o_P(n^{-\frac{1}{2}})$, as claimed, provided we can show that $P_0\left(\phi_{n,k}-\phi_\infty\right)^2=o_P(1)$. To do so, we first observe that we can express $\phi_{n,k}-\phi_\infty$ as the sum $U_{1,n,k}+U_{2,n,k}+\ldots+U_{11,n,k}$, where we define
	\begin{align*}
		U_{1,n,k}&:o\mapsto I(a=a_0) \left[\frac{\bar\gamma_{n,k}(a,z)}{\pi_{n,k}(z)} - \frac{\bar\gamma_{\infty}(a,z)}{\pi_{\infty}(z)}\right]  \phi_{\text{KM},\infty}(L_{\infty, \varphi})(z,a,w,y,\delta)\\
		U_{2,n,k}&:o\mapsto \frac{I(a=a_0,\delta=1)\bar\gamma_{n,k}(a,z)\tilde S_{\infty}(y\miid a,z)}{\pi_{n,k}(z) R_{\infty} (y\miid a,z)}\left[\frac{L_{n,k, \varphi}(y,a,z)}{\tilde S_{n,k}(y\miid a,z)} -\frac{L_{\infty, \varphi}(y,a,z)}{\tilde S_{\infty}(y\miid a,z)}\right]\\
		U_{3,n,k}&:o\mapsto \frac{I(a=a_0, \delta=1)\bar\gamma_{n,k}(a,z)L_{n,k, \varphi}(y,a,z)}{\pi_{n,k}(z)S_{n,k} (y\miid a,z)}\left[\frac{S_{n,k} (y\miid a,z)}{\tilde R_{n,k} (y\miid a,z)}-\frac{S_{\infty} (y\miid a,z)}{\tilde R_{\infty} (y\miid a,z)}\right]\\
		U_{4,n,k}&:o\mapsto - \frac{I(a=a_0)\bar\gamma_{n,k}(a,z)}{\pi_{n,k}(z)} \int_w^y \left[\frac{\tilde S_{n,k}(u\miid a,z)}{ R_{n,k}(u\miid a,z)} -\frac{\tilde S_{\infty}(u\miid a,z)}{R_\infty(u\miid a,z)} \right] \frac{L_{\infty, \varphi}(u,a,z)}{\tilde S_\infty(u\miid a,z)}\tilde\Lambda_\infty(du\miid a,z)\\
		U_{5,n,k}&:o\mapsto - \frac{I(a=a_0)\bar\gamma_{n,k}(a,z)}{\pi_{n,k}(z)} \\
		&\hspace{.5in} \int_w^y \frac{\tilde S_{n,k}(u\miid a,z)}{R_{n,k}(u\miid a,z)}\left[\frac{L_{n,k, \varphi}(u,a,z)}{\tilde S_{n,k}(u\miid a,z)}\tilde \Lambda_{n,k}(du\miid a,z) - \frac{L_{\infty, \varphi}(u,a,z)}{\tilde S_\infty(u\miid a,z)}\tilde \Lambda_\infty(du\miid a,z)\right]\\
		U_{6,n,k}&:o\mapsto \frac{\mu_{n,k}(z) - \mu_{\infty}(z)}{\gamma_\infty}\left[\frac{1}{S_\infty(w\miid a,z)}-\phi_{\text{KM}, \infty}(\gamma_{\infty, \natural})(z,a,w,y,\delta)\right]\\
		U_{7,n,k}&:o\mapsto \mu_{n,k}(z)\left[\frac{1}{\tilde S_{n,k}(w\miid a,z)\gamma_{n,k}} - \frac{1}{\tilde S_\infty(w\miid a,z)\gamma_\infty}\right]\\
		U_{8,n,k}&:o\mapsto\frac{\delta \mu_{n,k}(z)\tilde S_{\infty}(y\miid a,z)}{R_{\infty}(y\miid a,z)} \left[\frac{\gamma_{n,k, \natural}(y,a,z)}{\gamma_{n,k}\tilde S_{n,k}(y\miid a,z)} - \frac{\gamma_{\infty, \natural}(y,a,z)}{\gamma_{\infty}\tilde S_{\infty}(y\miid a,z)}\right]\\
		U_{9,n,k}&:o\mapsto - \frac{\delta \mu_{n,k}(z)\gamma_{n,k, \natural}(y,a,z)}{\gamma_{n,k}\tilde S_{n,k}(y\miid a,z)}\left[\frac{\tilde S_{n,k}(y\miid a,z)}{R_{n,k}(y\miid a,z)}-\frac{\tilde S_{\infty}(y\miid a,z)}{R_{\infty}(y\miid a,z)}\right]\\
		U_{10,n,k}&:o\mapsto \mu_{n,k}(z) \int_w^y \left[\frac{\tilde S_{n,k}(u\miid a,z)}{R_{n,k}(u\miid a,z)}-\frac{\tilde S_{\infty}(u\miid a,z)}{R_{\infty}(u\miid a,z)}\right]\frac{\gamma_{\infty, \natural}(u\miid a,z)}{\gamma_\infty \tilde S_\infty(u\miid a,z)}\tilde \Lambda_\infty(du\miid a,z)\\
		U_{11,n,k}&:o\mapsto\mu_{n,k}(z) \int_w^y \frac{\tilde S_{n,k}(u\miid a,z)}{R_{n,k}(u\miid a,z)}\left[\frac{\gamma_{n,k, \natural}(u,a,z)}{\gamma_{n,k}\tilde S_{n,k}(u\miid a,z)}\tilde \Lambda_{n,k}(du\miid a,z) - \frac{\gamma_{\infty, \natural}(u,a,z)}{\gamma_{\infty}\tilde S_\infty(u\miid a,z)}\tilde \Lambda_\infty(du\miid a,z)\right]
	\end{align*}with $\mu_{n,k}(z):=\int \varphi(u,z)(F_{n,k}-F_0)(du\miid a_0,z)$. By the triangle inequality, we then have that \[P_0\left(\phi_{n,k}-\phi_\infty\right)^2\leq\{(P_0U^2_{1,n,k})^\frac{1}{2}+(P_0U^2_{2,n,k})^\frac{1}{2}+\ldots+(P_0U^2_{11,n,k})^\frac{1}{2}\}^2,\]and so, we can focus on bounding each $P_0U^2_{j,n,k}$ separately. We define $\overline \tau(z) := \min\{\overline\tau_C(z), \overline\tau_T(z)\}$ and the (random) bounding terms $M_{1,n,k},M_{2,n,k},\ldots,M_{6,n,k}$ given by
	\begin{align*}
		M_{1,n,k}^2&:=E_0\left|\frac{\bar\gamma_{n,k}(a_0,Z)}{\pi_{n,k}(Z)} - \frac{\bar\gamma_{\infty}(a_0,Z)}{\pi_{\infty}(Z)}\right|^2\\
		M_{2,n,k}^2&:=E_0\left[\sup_{y\in[0,\overline{\tau}(a_0,Z)]}\left|\frac{L_{n,k, \varphi}(y,a_0,Z)}{\tilde S_{n,k}(y\miid a_0, Z)} - \frac{L_{\infty, \varphi}(y,a_0,Z)}{\tilde S_{\infty}(y\miid a_0, Z)}\right|\right]^2\\
		M_{3,n,k}^2&:=E_0\left[\sup_{y\in[0,\overline{\tau}(A,Z)]}\left|\frac{\tilde S_{n,k}(y\miid A, Z)}{ R_{n,k}(y\miid A, Z)} - \frac{\tilde S_{\infty}(y\miid A, Z)}{ R_{\infty}(y\miid A, Z)}\right|\right]^2\\
		M_{4,n,k}^2&:=E_0\left[\sup_{y\in[\underline{\tau}_T(A,Z),\overline{\tau}_W(A,Z)]}\left|\frac{\bar{\gamma}_{n,k,\natural}(y, A,Z)}{\tilde S_{n,k}(y\miid A,Z)} - \frac{\bar{\gamma}_{\infty, \natural}(y,A,Z)}{\tilde S_\infty(y\miid A,Z)}\right|\right]^2\\
		M_{5,n,k}^2&:=E_0\left[\sup_{y\in[\underline{\tau}_W(A,Z),\overline{\tau}_W(A,Z)]}\left| \frac{1}{\tilde S_{n,k}(y\miid A,Z)} - \frac{1}{\tilde S_\infty(y\miid A,Z)}\right|\right]^2\\
		M_{6,n,k}^2&:=E_0\left[\sup_{u\in[\underline{\tau}_W(A,Z),\overline{\tau}(A,Z)]}\left|G_{n,k}(u\miid A,Z)-G_{\infty}(u\miid A,Z)\right|\right]^2,
	\end{align*}where $E_0$ is a $P_0$--expectation  over the random data unit $(W,A,Z)$ drawn independently of $\eta_{n,k}$. Below, we restrict our attention to the portion of the sample space on which $\tilde F_0(t\miid a_0,z)$ is identified --- this is the relevant event to focus on since it has $P_0$--probability tending to one by Condition~\ref{ass:support}. Before proceeding, we note that $
	|\phi_{\text{KM},\infty}(L_{\infty,\varphi})(z,a_0,w,y,\delta)|\leq V_{1,\infty}(y,z)+V_{2,\infty}(y,w,z)$, where we write 
	\begin{align*}
		V_{1,\infty}(y,z)\ :=&\ \ \left|\frac{L_{\infty,\varphi}(y,a_0,z)}{R_\infty(y\miid a_0,z)}\right|\ =\ \left|\frac{L_{\infty,\varphi}(y,a_0,z)}{S_\infty(y\miid a_0,z)}\right|\left|\frac{S_\infty(y\miid a_0,z)}{R_\infty(y\miid a_0,z)}\right|\\
		V_{2,\infty}(y,w,z)\ :=&\ \ \left|\int_w^y \frac{L_{\infty,\varphi}(u\miid a_0,z)}{R_\infty(u\miid a_0,z)}\Lambda_\infty(du\miid a_0,z)\right|\\
		=&\ \ \left|\int_w^y \frac{L_{\infty,\varphi}(u\miid a_0,z)}{R_\infty(u\miid a_0,z)}S_\infty(u\miid a_0,z)\frac{\Lambda_\infty(du\miid a_0,z)}{S_\infty(u\miid a_0,z)}\right|\\
		=&\ \ \left|\int_w^y L_{\infty,\varphi}(u\miid a_0,z)\frac{S_\infty(u\miid a_0,z)}{R_\infty(u\miid a_0,z)}\frac{1}{S_\infty}(du\miid a_0,z)\right|\\
		\leq&\ \ \sup_{u\in[w,y]}\left|\frac{S_\infty(u\miid a_0,z)}{R_\infty(u\miid a_0,z)}\right|\int_w^y |L_{\infty,\varphi}(u\miid a_0,z)|\frac{1}{S_\infty}(du\miid a_0,z)\ .
	\end{align*}Now, for each $z$ over $[0,\overline{\tau}(z)]$, $u\mapsto \varphi(u,z)$ is assumed to have finite variation, and so, we can write $\varphi(\cdot,z)=\varphi_1(\cdot,z)-\varphi_2(\cdot,z)$ for non-decreasing functions $\varphi_1(\cdot,z)$ and $\varphi_2(\cdot,z)$ with finite variation. This implies that we can write $L_{\infty,\varphi}=L_{\infty,\varphi_1}-L_{\infty,\varphi_2}$ with $L_{\infty,\varphi_1},L_{\infty,\varphi_2}\geq 0$, which implies that $|L_{\infty,\varphi}(u\miid a_0,z)|\leq L_{\infty,\varphi_1}(u\miid a_0,z)+L_{\infty,\varphi_2}(u\miid a_0,z)$ and \begin{align*}
		&\int_w^y |L_{\infty,\varphi}(u\miid a_0,z)|\frac{1}{S_\infty}(du\miid a_0,z)\\
		&\hspace{.3in}\leq\ \int_w^y L_{\infty,\varphi_1}(u\miid a_0,z)\frac{1}{S_\infty}(du\miid a_0,z)+\int_w^y L_{\infty,\varphi_2}(u\miid a_0,z)\frac{1}{S_\infty}(du\miid a_0,z)\ .
	\end{align*}By integration by parts, we have that \begin{align*}
		&\int_w^y L_{\infty,\varphi_j}(u\miid a_0,z)\frac{1}{S_\infty}(du\miid a_0,z)\\
		&\hspace{.3in}=\ \frac{L_{\infty,\varphi_j}(y\miid a_0,z)}{S_\infty(y\miid a_0,z)}-\frac{L_{\infty,\varphi_j}(w\miid a_0,z)}{S_\infty(w\miid a_0,z)}-\int_w^y \frac{1}{S_\infty(u\miid a_0,z)}L_{\infty,\varphi_j}(du\miid a_0,z)\\
		&\hspace{.3in}\leq\ 2\sup_{u\in[w,y]}\left|\frac{L_{\infty,\varphi_j}(u\miid a_0,z)}{S_\infty(u\miid a_0,z)}\right|+\int_w^y \varphi_j(du,z)\ \leq\ 2\sup_{u\in[w,y]}\left|\frac{L_{\infty,\varphi}(u\miid a_0,z)}{S_\infty(u\miid a_0,z)}\right|+\|\varphi_j(\cdot,z)\|_{v,[0,\overline{\tau}(z)]}\ ,
	\end{align*}which then implies that \[\int_w^y |L_{\infty,\varphi}(u\miid a_0,z)|\frac{1}{S_\infty}(du\miid a_0,z)\ \leq\ 4\sup_{u\in[w,y]}\left|\frac{L_{\infty,\varphi}(u\miid a_0,z)}{S_\infty(u\miid a_0,z)}\right|+\|\varphi(\cdot,z)\|_{v,[0,\overline{\tau}(z)]}\]and therefore  that $|\phi_{\text{KM},\infty}(L_{\infty,\varphi})(z,a_0,w,y,\delta)|$ is bounded above by  \[ \sup_{u\in[w,y]}\left|\frac{S_\infty(u\miid a_0,z)}{R_\infty(u\miid a_0,z)}\right|\left[5\sup_{u\in[w,y]}\left|\frac{L_{\infty,\varphi}(u\miid a_0,z)}{S_\infty(u\miid a_0,z)}\right|+\|\varphi(\cdot,z)\|_{v,[0,\overline{\tau}(z)]}\right].\]
	In view of the above, we have that \begin{align*}
		P_0U_{1,n,k}^2\,&\leq\,E_0\left[\left|\frac{\bar\gamma_{n,k}(a_0,Z)}{\pi_{n,k}(Z)} - \frac{\bar\gamma_{\infty}(a_0,Z)}{\pi_{\infty}(Z)}\right|^2 \phi_{\text{KM},\infty}(L_{\infty,\varphi})(Z,a_0,W,Y,\Delta)^2\right]\,\leq\,36\,\kappa^4M_{1,n,k}^2\ .
	\end{align*}We can write that \begin{align*}
		P_0U_{2,n,k}^2\,&\leq\,E_0\left[\left|\frac{L_{n,k, \varphi}(Y,a_0,Z)}{\tilde S_{n,k}(Y\miid a_0,Z)} -\frac{L_{\infty, \varphi}(Y,a_0,Z)}{\tilde S_{\infty}(Y\miid a_0,Z)}\right| \frac{\bar\gamma_{n,k}(a_0,Z)}{\pi_{n,k}(Z) }\frac{\tilde S_{\infty}(Y\miid a_0,Z)}{R_{\infty} (Y\miid a_0,Z)}\right]^2\,\leq\,\kappa^4M_{2,n,k}^2\ .
	\end{align*}We can also write that \begin{align*}
		P_0U_{3,n,k}^2\,&\leq\,E_0\left[\left|\frac{\tilde S_{n,k}(Y\miid a_0,Z)}{R_{n,k}(Y\miid a_0,Z)} -\frac{\tilde S_{\infty}(Y\miid a_0,Z)}{R_{\infty}(Y\miid a_0,Z)}\right| \frac{\bar\gamma_{n,k}(a_0,Z)}{\pi_{n,k}(Z) }\frac{L_{n,k, \varphi}(Y,a_0,Z)}{\tilde S_{n,k}(Y\miid a_0,Z)}\right]^2\,\leq\,\kappa^4M_{3,n,k}^2\ .
	\end{align*}We have that \begin{align*}
		&P_0U_{4,n,k}^2\,=\,E_0\left[\frac{\bar\gamma_{n,k}(a_0,Z)}{\pi_{n,k}(Z)}\int_W^Y\left[\frac{\tilde S_{n,k}(u\miid a_0,Z)}{ R_{n,k}(u\miid a_0,Z)} -\frac{\tilde S_{\infty}(u\miid a_0,Z)}{R_\infty(u\miid a_0,Z)}\right]L_{\infty,\varphi}(u,a_0,Z)\frac{1}{S_\infty}(du\miid a_0,Z)\right]^2\\
		&\leq\, E_0\left[\frac{\bar\gamma_{n,k}(a_0,Z)}{\pi_{n,k}(Z)}\sup_{u\in[W,Y]}\left|\frac{\tilde S_{n,k}(u\miid a_0,Z)}{ R_{n,k}(u\miid a_0,Z)} -\frac{\tilde S_{\infty}(u\miid a_0,Z)}{R_\infty(u\miid a_0,Z)}\right|\int_W^Y |L_{\infty,\varphi}(u\miid a_0,Z)|\frac{1}{S_\infty}(du\miid a_0,Z)\right]^2\\
		&\leq\, \frac{25\kappa^4}{\gamma_{n,k}^2} M^2_{3,n,k}\ .
	\end{align*}Next, using part (c) of Lemma 1, and writing for convenience \[\Theta_{n,k}(du):=\frac{L_{n,k, \varphi}(u,a_0,Z)}{\tilde S_{n,k}(u\miid a_0,Z)}\tilde \Lambda_{n,k}(du\miid a_0,Z)\text{\ \ and\ \ }\Theta_\infty(du):=\frac{L_{\infty, \varphi}(u,a_0,Z)}{\tilde S_\infty(u\miid a_0,Z)}\tilde \Lambda_\infty(du\miid a_0,Z)\ ,\]we have that  \[
	P_0U^2_{5,n,k}\,\leq\,E_0\left[\frac{\bar\gamma_{n,k}(a_0,Z)}{\pi_{n,k}(Z)} \int_W^Y \frac{\tilde S_{n,k}(u\miid a_0,Z)}{R_{n,k}(u\miid a_0,Z)}\left(\Theta_{n,k}-\Theta_\infty\right)(du)\right]^2.\]
	To make further progress on this bound, we note that \begin{align*}
		&\left|\int_w^y \frac{\tilde S_{n,k}(u\miid a_0,z)}{R_{n,k}(u\miid a_0,z)}(\Theta_{n,k}-\Theta_\infty)(du)\right|\\
		&\leq\, \left|\int_w^y \left[\frac{\tilde S_{n,k}(u\miid a_0,z)}{R_{n,k}(u\miid a_0,z)}-\frac{\tilde S_{\infty}(u\miid a_0,z)}{R_{\infty}(u\miid a_0,z)}\right](\Theta_{n,k}-\Theta_\infty)(du)\right|+\left|\int_w^y \frac{\tilde S_{\infty}(u\miid a_0,z)}{R_{\infty}(u\miid a_0,z)}(\Theta_{n,k}-\Theta_\infty)(du)\right|\\
		&\leq\, \sup_{u\in[w,y]}\left|\frac{\tilde S_{n,k}(u\miid a_0,z)}{R_{n,k}(u\miid a_0,z)}-\frac{\tilde S_{\infty}(u\miid a_0,z)}{R_{\infty}(u\miid a_0,z)}\right|\int_w^y|(\Theta_{n,k}-\Theta_\infty)(du)|\\
		&\hspace{.5in}+\left|\int_w^y \frac{\tilde S_{\infty}(u\miid a_0,z)}{R_{\infty}(u\miid a_0,z)}(\Theta_{n,k}-\Theta_\infty)(du)\right|.
	\end{align*}Using arguments used above, we observe that \begin{align*}
		&\int_w^y|(\Theta_{n,k}-\Theta_\infty)(du)|\,\leq\,\int_w^y|\Theta_{n,k}(du)|+\int_w^y |\Theta_\infty(du)|\\
		&\hspace{0.5in}=\, \int_w^y \frac{|L_{n,k,\varphi}(u\miid a_0,z)|}{\tilde{S}_{n,k}(u\miid a_0,z)}\tilde\Lambda_{n,k}(du\miid a_0,z)+\int_w^y \frac{|L_{\infty,\varphi}(u\miid a_0,z)|}{\tilde{S}_{\infty}(u\miid a_0,z)}\tilde\Lambda_{\infty}(du\miid a_0,z)\\
		&\hspace{0.5in}\leq\,4\sup_{u\in[w,y]}\left|\frac{L_{\infty,\varphi}(u\miid a_0,z)}{S_\infty(u\miid a_0,z)}\right|+4\sup_{u\in[w,y]}\left|\frac{L_{n,k,\varphi}(u\miid a_0,z)}{S_{n,k}(u\miid a_0,z)}\right|+2\|\varphi(\cdot,z)\|_{v,[0,\overline{\tau}(z)]}\,\leq\,10\kappa
	\end{align*}and also, in view of part (c) of Lemma 1, that \begin{align*}
		\left|\int_W^Y \frac{\tilde S_{\infty}(u\miid a_0,Z)}{R_{\infty}(u\miid a_0,Z)}(\Theta_{n,k}-\Theta_\infty)(du)\right|\,\leq\,3\kappa \sup_{u\in[W,Y]}\left|\frac{L_{n,k,\varphi}(u,a_0,Z)}{\tilde{S}_{n,k}(u\miid a_0,Z)}-\frac{L_{\infty,\varphi}(u,a_0,Z)}{\tilde{S}_{\infty}(u\miid a_0,Z)}\right|
	\end{align*}holds $P_0$--almost surely. As a consequence, we find that \begin{align*}
		P_0U^2_{5,n,k}\,&\leq\,E_0\left[\frac{\bar\gamma_{n,k}(a_0,Z)}{\pi_{n,k}(Z)} \int_W^Y \frac{\tilde S_{n,k}(u\miid a_0,Z)}{R_{n,k}(u\miid a_0,Z)}\left(\Theta_{n,k}-\Theta_\infty\right)(du)\right]^2\\
		&\leq\,\kappa^4(10 M_{3,n,k}+3 M_{2,n,k})^2
	\end{align*}
	Before studying the remaining terms, we note that \begin{align*}
		&\mu_{n,k}(z)-\mu_{\infty}(z)\,=\, \int \varphi(y,z)(F_{n,k}-F_0)(dy\miid a_0,z)-\int \varphi(y,z)(F_{\infty}-F_0)(dy\miid a_0,z)\\
		&\hspace{.3in}=\, -\left[\int\frac{ L_{n,k,\varphi}(y, a_0, z) }{S_{n,k}(y \miid a_0, z)}{S}_0(y \miid a_0, z)(\Lambda_{n,k}-\Lambda_0)(dy \miid a_0, z) \right.\\
		&\hspace{1in}-\left.\int \frac{L_{\infty,\varphi}(y, a_0, z)}{S_{\infty}(y \miid a_0, z)}{S}_0(y \miid a_0, z) (\Lambda_{\infty}-\Lambda_0)(dy \miid a_0, z) \right]\\
		&\hspace{.3in}=\int\left[\frac{ L_{n,k,\varphi}(y, a_0, z) }{S_{n,k}(y \miid a_0, z)} - \frac{L_{\infty,\varphi}(y, a_0, z)}{S_{\infty}(y \miid a_0, z)}\right]{S}_0(y \miid a_0, z)\Lambda_0(dy \miid a_0, z) \\
		&\hspace{1in}-\int {S}_0(y \miid a_0, z) \left[\frac{ L_{n,k,\varphi}(y, a_0, z) }{S_{n,k}(y \miid a_0, z)}\Lambda_{n,k}(dy \miid a_0, z)-\frac{L_{\infty,\varphi}(y, a_0, z)}{S_{\infty}(y \miid a_0, z)}\Lambda_{\infty}(dy \miid a_0, z)\right],
	\end{align*}and so, in view of part (c) of Lemma~\ref{lem:duham}, we find that \begin{align*}
		|\mu_{n,k}(z)-\mu_{\infty}(z)|\ \leq\ 4\sup_{y\in[0,\bar{\tau}(z)]}\left|\frac{ L_{n,k,\varphi}(y, a_0, z) }{S_{n,k}(y \miid a_0, z)}-\frac{L_{\infty,\varphi}(y, a_0, z)}{S_{\infty}(y \miid a_0, z)}\right|.
	\end{align*}
	Similarly as above, before proceeding, we note that $
	|\phi_{\text{KM},\infty}(L_{\infty,\varphi})(z,a,w,y,\delta)|$ is bounded above by $ V_{1,\infty}(y,a,z)+V_{2,\infty}(y,w,a,z)$, where we write 
	\begin{align*}
		V_{1,\infty}(y,a,z)\ :=&\ \ \left|\frac{\gamma_{\infty,\natural}(y,a,z)}{R_\infty(y\miid a,z)}\right|\ =\ \left|\frac{\gamma_{\infty,\natural}(y,a,z)}{S_\infty(y\miid a,z)}\right|\left|\frac{S_\infty(y\miid a,z)}{R_\infty(y\miid a,z)}\right|\\
		V_{2,\infty}(y,w,a,z)\ :=&\ \ \left|\int_w^y \frac{\gamma_{\infty,\natural}(u\miid a,z)}{R_\infty(u\miid a,z)}\Lambda_\infty(du\miid a,z)\right|\\
		=&\ \ \left|\int_w^y \frac{\gamma_{\infty,\natural}(u\miid a,z)}{R_\infty(u\miid a,z)}S_\infty(u\miid a,z)\frac{\Lambda_\infty(du\miid a,z)}{S_\infty(u\miid a,z)}\right|\\
		=&\ \ \left|\int_w^y \gamma_{\infty,\natural}(u\miid a,z)\frac{S_\infty(u\miid a,z)}{R_\infty(u\miid a,z)}\frac{1}{S_\infty}(du\miid a,z)\right|\\
		\leq&\ \ \sup_{u\in[w,y]}\left|\frac{S_\infty(u\miid a,z)}{R_\infty(u\miid a,z)}\right|\int_w^y \gamma_{\infty,\natural}(u\miid a,z)\frac{1}{S_\infty}(du\miid a,z)\ ,
	\end{align*} and since we have that \begin{align*}
		0\ \leq\ \int_w^y\gamma_{\infty,\natural}(u\miid a,z)\frac{1}{S_\infty}(du\miid a,z)\ &=\ \left.\frac{\gamma_{\infty,\natural}(u\miid a,z)}{S_\infty(u\miid a,z)}\right|_{u=w}^{y}-\int_w^y \frac{1}{S_\infty(u\miid a,z)}\gamma_{\infty,\natural}(du\miid a,z)\\
		&=\ \left.\frac{\gamma_{\infty,\natural}(u\miid a,z)}{S_\infty(u\miid a,z)}\right|_{u=w}^{y}+\int_w^y \left[\frac{1}{S_\infty(u\miid a,z)}\right]^2G_\infty(du\miid a,z)\\
		& \leq\ 2\sup_{u\in[w,y]}\left|\frac{\gamma_{\infty,\natural}(u\miid a,z)}{S_\infty(u\miid a,z)}\right|+\left[\frac{1}{S_\infty(\overline{\tau}_W(z)\miid a,z)}\right]^2,
	\end{align*}it follows that $|\phi_{\text{KM},\infty}(L_{\infty,\varphi})(Z,A,W,Y,\Delta)|\leq \kappa(2\kappa+\kappa^2)<\infty$ $P_0$--almost surely. Using this fact, we then can write that \begin{align*}
		P_0U^2_{6,n,k}\,&=\,\gamma_0^{-2}E_0\left[[\mu_{n,k}(Z)-\mu_{\infty}(Z)]\left[\frac{1}{S_\infty(W\miid A,Z)}-\phi_{\text{KM},\infty}(\gamma_{\infty,\natural})(Z,A,W,Y,\Delta)\right]\right]^2\\
		&\leq\, \left[\frac{\kappa+\kappa(2\kappa+\kappa^2)}{\gamma_0}\right]^2E_0\left[\mu_{n,k}(Z)-\mu_\infty(Z)\right]^2\,\leq\, 16\left[\frac{\kappa+\kappa(2\kappa+\kappa^2)}{\gamma_0}\right]^2M^2_{2,n,k}\ .
	\end{align*}
	Next, we have that
	\begin{align*}
		&P_0U_{7,n,k}^2\,=\,E_0\left[\mu_{n,k}(Z)\left|\frac{1}{\tilde S_{n,k}(W\miid A,Z)\gamma_{n,k}}-\frac{1}{\tilde S_{\infty}(W\miid A,Z)\gamma_{\infty}}\right| \right]^2\\
		&\leq\, E_0\left[\mu_{n,k}(Z)\left|\frac{1}{\tilde S_{n,k}(W\miid A,Z)}\left(\frac{1}{\gamma_{n,k}}-\frac{1}{\gamma_\infty}\right)+\frac{1}{\gamma_\infty}\left[\frac{1}{\tilde S_{\infty}(W\miid A,Z)}-\frac{1}{\tilde S_{\infty}(W\miid A,Z)}\right]\right| \right]^2\\
		&\leq\ \Bigg{[}\left|\frac{1}{\gamma_{n,k}}-\frac{1}{\gamma_\infty}\right|\Bigg{[}E_0\left|\frac{\mu_{n,k}(Z)}{\tilde{S}_{n,k}(W\miid A,Z)}\right]^2\Bigg{]}^\frac{1}{2}  
		+\Bigg{[}E_0\left|\frac{\mu_{n,k}(Z)}{\gamma_\infty}\Bigg{[}\frac{1}{\tilde S_{\infty}(W\miid A,Z)}-\frac{1}{\tilde S_{\infty}(W\miid A,Z)}\Bigg{]}\right|^2\Bigg{]}^\frac{1}{2}\Bigg{]}^2\\
		&\leq\ \left[2\kappa^2\left|\frac{1}{\gamma_{n,k}}-\frac{1}{\gamma_\infty}\right|+2\kappa\frac{1}{\gamma_\infty}M_{5,n,k}\right]^2.
	\end{align*} Using a similar expansion, it can be shown that \begin{align*}
		&P_0U_{8,n,k}^2\leq 2\kappa^2 \left[2\kappa^2\left|\frac{1}{\gamma_{n,k}}-\frac{1}{\gamma_\infty}\right|+2\kappa\frac{1}{\gamma_\infty}M_{4,n,k}\right]^2
	\end{align*}We have that \begin{align*}
		P_0U^2_{9,n,k}\,=\,E_0\left[|\mu_{n,k}(Z)|\left|\frac{\bar{\gamma}_{n,k,\natural}(Y,A,Z)}{\tilde{S}_{n,k}(Y\miid A,Z)}\right|\left|\frac{\tilde S_{n,k}(Y\miid A,Z)}{R_{n,k}(Y\miid A,Z)} -\frac{\tilde S_{\infty}(Y\miid A,Z)}{R_{\infty}(Y\miid A,Z)}\right|\right]^2\,\leq\, 4\kappa^4 M_{3,n,k}^2\ .
	\end{align*}Next, we note that \begin{align*}
		& \left|\int_w^y \left[\frac{\tilde S_{n,k}(u\miid a,z)}{R_{n,k}(u\miid a,z)}-\frac{\tilde S_{\infty}(u\miid a,z)}{R_{\infty}(u\miid a,z)}\right]\frac{\gamma_{\infty, \natural}(u\miid a,z)}{\gamma_\infty \tilde S_\infty(u\miid a,z)}\tilde \Lambda_\infty(du\miid a,z)\right|\\
		&\leq\,\gamma_\infty^{-1}\sup_{u\in[w,y]}\left|\frac{\tilde S_{n,k}(u\miid a,z)}{R_{n,k}(u\miid a,z)}-\frac{\tilde S_{\infty}(u\miid a,z)}{R_{\infty}(u\miid a,z)}\right|\int_w^y \gamma_{\infty,\natural}(u\miid a,z)\frac{1}{S_\infty}(du\miid a,z)\\
		&\leq\,\gamma_\infty^{-1}\sup_{u\in[w,y]}\left|\frac{\tilde S_{n,k}(u\miid a,z)}{R_{n,k}(u\miid a,z)}-\frac{\tilde S_{\infty}(u\miid a,z)}{R_{\infty}(u\miid a,z)}\right|\left[2\sup_{u\in[w,y]}\frac{\gamma_{\infty,\natural}(u\miid a,z)}{S_\infty(u\miid a,z)}+\left[\frac{1}{S_\infty(\overline{\tau}_W(z)\miid a,z)}\right]^2\right]
	\end{align*} using the fact that \begin{align*}
		0\,\leq\,\int_w^y\gamma_{\infty,\natural}(u\miid a,z)\frac{1}{S_\infty}(du\miid a,z)\, &=\,\left.\frac{\gamma_{\infty,\natural}(u\miid a,z)}{S_\infty(u\miid a,z)}\right|_{u=w}^y-\int_w^y\left[\frac{1}{S_\infty(u\miid a,z)}\right]^2G_\infty(du\miid a,z)\\
		&\leq\,2\sup_{u\in[w,y]}\frac{\gamma_{\infty,\natural}(u\miid a,z)}{S_\infty(u\miid a,z)}+\left[\frac{1}{S_\infty(\overline{\tau}_W(z)\miid a,z)}\right]^2.
	\end{align*}Thus, using that $|\mu_\infty(z)|\leq \|u\mapsto\varphi(u,z)\|_{\infty}$, we find that \begin{align*}
		P_0U^2_{10,n,k}\,&\leq\,\left[\frac{\|(u,z)\mapsto\varphi(u,z)\|_\infty(2\kappa+\kappa^2)}{\gamma_\infty}\right]^2M_{3,n,k}^2\ .
	\end{align*}To study the final term, similarly as in the study of $P_0U^2_{5,n,k}$, we first note that \begin{align*}
		&\left|\int_w^y \frac{\tilde S_{n,k}(u\miid a,z)}{R_{n,k}(u\miid a,z)}(\Theta_{n,k}-\Theta_\infty)(du)\right|\\
		&\leq\, \left|\int_w^y \left[\frac{\tilde S_{n,k}(u\miid a_0,z)}{R_{n,k}(u\miid a,z)}-\frac{\tilde S_{\infty}(u\miid a,z)}{R_{\infty}(u\miid a,z)}\right](\Theta_{n,k}-\Theta_\infty)(du)\right|+\left|\int_w^y \frac{\tilde S_{\infty}(u\miid a,z)}{R_{\infty}(u\miid a,z)}(\Theta_{n,k}-\Theta_\infty)(du)\right|\\
		&\leq\, \sup_{u\in[w,y]}\left|\frac{\tilde S_{n,k}(u\miid a,z)}{R_{n,k}(u\miid a,z)}-\frac{\tilde S_{\infty}(u\miid a,z)}{R_{\infty}(u\miid a,z)}\right|\int_w^y|(\Theta_{n,k}-\Theta_\infty)(du)|+\left|\int_w^y \frac{\tilde S_{\infty}(u\miid a,z)}{R_{\infty}(u\miid a,z)}(\Theta_{n,k}-\Theta_\infty)(du)\right|,
	\end{align*}where we have now redefined $\Theta_{n,k}$ and $\Theta_\infty$ via the differentials \[\Theta_{n,k}(du):=\frac{\gamma_{n,k, \natural}(u,a,z)}{\tilde S_{n,k}(u\miid a,z)}\tilde \Lambda_{n,k}(du\miid a,z)\text{\ \ and\ \ }\Theta_\infty(du):=\frac{\gamma_{\infty, \natural}(u,a,z)}{\tilde S_\infty(u\miid a,z)}\tilde \Lambda_\infty(du\miid a,Z)\ .\]Using similar arguments as used above, we observe that \begin{align*}
		&\int_w^y|(\Theta_{n,k}-\Theta_\infty)(du)|\,\leq\,\int_w^y|\Theta_{n,k}(du)|+\int_w^y |\Theta_\infty(du)|\\
		&=\, \int_w^y \frac{\gamma_{n,k, \natural}(u,a,z)}{\tilde{S}_{n,k}(u\miid a,z)}\tilde\Lambda_{n,k}(du\miid a,z)+\int_w^y \frac{\gamma_{\infty, \natural}(u,a,z)}{\tilde{S}_{\infty}(u\miid a,z)}\tilde\Lambda_{\infty}(du\miid a,z)\\
		&\leq\,2\sup_{u\in[w,y]}\left|\frac{\gamma_{n,k,\natural}(u\miid a,z)}{S_{n,k}(u\miid a,z)}\right|+2\sup_{u\in[w,y]}\left|\frac{\gamma_{\infty,\natural}(u\miid a,z)}{S_{\infty}(u\miid a,z)}\right|+\left[\frac{1}{S_{\infty}(\overline{\tau}_W(z)\miid a,z)}\right]^2+\left[\frac{1}{S_{n,k}(\overline{\tau}_W(z)\miid a,z)}\right]^2\\
		&\leq\,4\kappa+2\kappa^2
	\end{align*}and also, in view of part (d) of Lemma~\ref{lem:duham}, that $\left|\int_W^Y \frac{\tilde S_{\infty}(u\miid A,Z)}{R_{\infty}(u\miid A,Z)}(\Theta_{n,k}-\Theta_\infty)(du\miid A,Z)\right|$ is $P_0$--almost surely bounded above by\begin{align*}
		&3\kappa\sup_{u\in[W,Y]}\left|\frac{\gamma_{n,k,\natural}(u\miid A,Z)}{S_{n,k}(u\miid A,Z)}-\frac{\gamma_{\infty,\natural}(u\miid A,Z)}{S_{\infty}(u\miid A,Z)}\right|+3\kappa^4\sup_{u\in[W,Y]}\left|G_{n,k}(u\miid A,Z)-G_\infty(u\miid A,Z)\right|\\
		&\hspace{3.4in}+2\kappa^4\sup_{u\in[W,Y]}\left|\frac{1}{S_{n,k}(u\miid A,Z)}-\frac{1}{S_\infty(u\miid A,Z)}\right|.
	\end{align*}As a consequence, we find that \begin{align*}
		P_0U^2_{11,n,k}\,&\leq\,E_0\left[\mu_{n,k}(Z) \int_W^Y \frac{\tilde S_{n,k}(u\miid A,Z)}{R_{n,k}(u\miid A,Z)}\left[\Theta_{n,k}(du\miid A,Z) - \Theta_\infty(du\miid A,Z)\right]\right]^2\\
		&\leq\,4\kappa^2\left[(4\kappa+2\kappa^2)M_{3,n,k}+3\kappa M_{4,n,k}+3\kappa^4M_{6,n,k}+2\kappa^4M_{5,n,k}\right]^2.
	\end{align*}In view of all the inequalities derived, we see that if $M_{1,n,k},M_{2,n,k},\ldots,M_{6,n,k}$ tend to zero in probability, then so does each $P_0U^2_{j,n,k}$ for $j=1,2,\ldots,11$ and thus $P_0(\phi_{n,k}-\phi_\infty)^2$ itself tends to zero in probability, and thus, that $r_{bn}=o_P(n^{-\frac{1}{2}})$.
	
	Next, we study $r_{cn}$. We argue first that $r_{cn}=o_P(1)$ under Conditions \ref{ass:pro_conv}--\ref{ass:pro_pos} and \ref{ass:consistent_a}--\ref{ass:consistent_b}. We observe that $R(\eta_\infty,\eta_0)=0$ under Conditions \ref{ass:consistent_a}--\ref{ass:consistent_b}, and so, by Conditions \ref{ass:pro_conv}--\ref{ass:pro_pos}, it follows that $R(\eta_{k,n},\eta_\infty)=o_P(1)$. This then readily implies that  $r_{cn}=o_P(1)$, establishing part (i) of the theorem. If, instead, condition \ref{ass:pro_remainder} holds, a condition that necessarily requires all parts of Condition \ref{ass:consistent} to hold, it follows that $r_{cn}=o_P(n^{-\frac{1}{2}})$, thereby establishing the asymptotic linearity of $\psi_n^*$ with influence function $\phi_0$. This partially establishes part (iii) of the theorem.
	
	We now establish part (ii). To do so, we show that $P_0 \phi_{\eta_\infty, \psi_0} =0$ under Condition \ref{ass:consistent_a}, that is, provided $S_{X,\infty}=S_{X,0}$, $G_{\infty}=G_0$ and $\pi_\infty=\pi_0$, even if possibly $Q_\infty\neq Q_0$. Suppose that Condition \ref{ass:consistent_a} holds. We compute $P_0\phi_{\eta_\infty,\psi_0}=P_0\phi_{1,\eta_\infty,\psi_0}+P_0\phi_{2,\eta_\infty,\psi_0}$ and study these two summands separately. We begin by noting that \begin{align*}
		P_0\phi_{1,\eta_\infty,\psi_0}\ &=\ E_0\left\{\frac{I(A=a_0)}{\pi_\infty(Z)}\bar{\gamma}_\infty(Z)\phi_{\text{KM},\eta_\infty}(L_{\eta_\infty,\varphi})(O)\right\}\\
		&=\ E_0\left[\frac{I(A=a_0)}{\pi_\infty(Z)}\bar{\gamma}_\infty(Z)E_0\{\phi_{\text{KM},\eta_\infty}(L_{\eta_\infty,\varphi})(O)\,|\,A,Z\}\right]
	\end{align*} and similarly, using that $\mu_\infty=\mu_0$ under Condition \ref{ass:consistent_a}, 
	\[
	\gamma_\infty P_0\phi_{2,\eta_\infty,\psi_0}\,=\,E_0\left\{\frac{\mu_0(Z)-\psi_0}{\tilde{S}_0(W\,|\,A,Z)}\right\}-E_0\left[\left\{\mu_0(Z)-\psi_0\right\}E_0\left\{\phi_{\text{KM},\eta_\infty}(\gamma_{\eta_\infty,\natural})(O)\,|\,A,Z\right\}\right].
	\]
	We can compute the resulting observable conditional subdistribution function as 
	\begin{align*}
		F_{1,\infty}(u\miid a,z)\ &:=\ P_\infty(Y\le u, \Delta=1\miid A=a, Z=z) \\
		&=\ \frac{\int I_{[0,u]}(t)\left\{\int I_{[0,t]}(w) Q_\infty(t\,|\,w,a,z)G_{X,0}(dw\,|\,a,z)\right\}F_{X,0}(dt\,|\,a,z)}{\int G_{X,0}(t\,|\,a,z)F_{X,0}(dt\,|\,a,z)}\, ,
	\end{align*}
	and that the observable conditional at-risk probability function as 
	\begin{align*}
		R_\infty(u\miid a,z) :=&\ P_\infty(W\le u\le Y\miid A=a,Z=z)\\
		=&\, \frac{\int I_{[0,u]}(w)Q_\infty(u\,|\,w,a,z)G_{X,0}(dw\,|\,a,z)}{\int G_{X,0}(t\,|\,a,z)F_{X,0}(dt\,|\,a,z)}S_{X,0}(u\miid a,z)\, ,
	\end{align*}
	which implies that $F_{1,\infty}(du\miid a,z)/R_\infty(u\miid a,z) = F_{X,0}(du\miid a,z)/S_{X,0}(u\miid a,z)$. Similarly, we have that $\gamma_{\eta_\infty, \natural} = \gamma_{0,\natural}$, so that we find that
	\begin{align*}
		&E_0\left\{\phi_{\text{KM},\eta_\infty}(L_{\eta_\infty, \varphi})(O)\miid A=a, Z=z\right\}\\
		&\hspace{1in}=\ \int\frac{L_{\eta_\infty, \varphi}(u,a,z)R_0(u\miid a,z)}{R_\infty(u\miid a,z)}\left\{\frac{F_{1, \infty}(du\miid a,z)}{R_\infty(u\miid a,z)} - \frac{F_{1, 0}(du\miid a,z)}{R_0(u\miid a,z)}\right\}=\,0\\
		&E_0\left\{\phi_{\text{KM},\eta_\infty}(\gamma_{\eta_\infty, \natural})(O)\miid A=a, Z=z\right\}\\
		&\hspace{1in}=\ \int\frac{\gamma_{\eta_\infty, \natural}(u,a,z)R_0(u\miid a,z)}{R_\infty(u\miid a,z)}\left\{\frac{F_{1, \infty}(du\miid a,z)}{R_\infty(u\miid a,z)} - \frac{F_{1, 0}(du\miid a,z)}{R_0(u\miid a,z)}\right\}=\,0\ .
	\end{align*}
	It follows then that 
	\[P_0 \phi_{\eta_\infty, \psi_0}=\frac{1}{\gamma_\infty}E_0\left\{\frac{\mu_0(Z) -\psi_0}{\tilde{S}_0(W\,|\,A,Z)}\right\}=\frac{\gamma_0}{\gamma_\infty}E_0\left[\{\mu_0(Z) -\psi_0\}\frac{\gamma_0(A,Z)}{\gamma_0}\right]=0\ .\]Since $\psi_0$ is the unique solution in $\psi$ of the equation  $P_0\phi_{\eta_\infty,\psi}=0$, we find that $\psi\mapsto \phi_{\eta_\infty,\psi}$ is a proper estimating function for $\psi_0$. The solution $\psi_n^{**}$ of the corresponding cross-fitted estimating equation can then be shown to be consistent for $\psi_0$ using usual estimation equations theory. Alternatively, the explicit form of $\psi_n^{**}$ could be used to study its consistency directly. This argument establishes part (ii) of the theorem.
	
	Finally, we show that $\psi_n^*$ and $\psi_n^{**}$ are asymptotically equivalent under Conditions \ref{ass:pro_conv}--\ref{ass:pro_remainder}, thus completing the proof of part (iii) of the theorem. To do so, we note that the difference between these estimators can be algebraically expressed as
	\[\psi^*-\psi_n^{**}=\frac{1}{K }\sum_{k=1}^K\frac{1}{ n_k}\sum_{i\in\mathcal V_k}(\psi_{k,n}-\psi^{**}_{k,n})\left[1-\frac{1}{\gamma_{k,n}}\frac{1}{n_k}\sum_{i\in\mathcal{V}_k}\left\{\frac{1}{\tilde{S}_{k,n}(W_i\miid A_i,Z_i)}-\phi_{\text{KM},\eta_{k,n}}(\gamma_{k,n,\natural\}})(O_i)\right\}\right].\]
	Conditions~\ref{ass:pro_conv}--\ref{ass:pro_remainder} are sufficient to establish that $\frac{1}{n_k}\sum_{i\in\mathcal{V}_k}\tilde{S}_{k,n}(W_i\miid A_i,Z_i)^{-1}=\gamma_{k,n}+O_P(n_k^{-\frac{1}{2}})$ and $\frac{1}{ n_k}\sum_{i\in\mathcal V_k}\phi_{\text{KM},\eta_{k,n}}(\gamma_{k,n,\natural})(O_i)=O_P(n_k^{-\frac{1}{2}})$. Given that both  $\psi_{k,n}$ and $\psi^{**}_{k,n}$ tend to $\psi_0$ in probability under these conditions, it follows then that $\psi^*_{n} -\psi^{**}_{n} = o_P(n^{-\frac{1}{2}})$, establishing the equivalence between $\psi^*_n$ and $\psi^{**}_n$.
	
	\subsection*{Proof of Theorem~\ref{thm:unif}}
	
	We first show that the stochastic processes $\mathbb B_n := \{n^{\frac{1}{2}}\,[\psi_n(s) - \psi_0(s)]: s\in (0,\tau)\}$ and  $\bar{\mathbb B}_n:=\{n^{-\frac{1}{2}} \sum_i\phi_{s,0}(O_i): s\in (0,\tau)\}$ are asymptotically equivalent.
	Under Condition~\ref{ass:uniform_bounded}, $\bar{\mathbb B}_n$ converges weakly to $\mathbb B_0$ relative to the supremum norm. The limiting distribution $\mathbb B_0 $ is a mean-zero Gaussian process with covariance function $\sigma_0^2(s,t) := \E_0\left\{\phi_{s,0}(O)\phi_{t,0}(O)\right\}$, in view of Theorem 19.3~\citep{van2000asymptotic}. Let $g$ be any function bounded by $M_0$, and Lipschitz with constant $L_0$. We know that $\E_0\left\{g(\bar{\mathbb B}_n)\right\}$ converges to $\E_0\left\{g(\mathbb B_0)\right\}\rvert$ from the convergence of $\bar{\mathbb B}_n$ to $\mathbb B_0$, therefore the inequality
	\begin{align*}
		\lvert \E_0\left\{g(\mathbb B_n)\right\} -\E_0\left\{g(\mathbb B_0)\right\}\rvert 
		\le&\,  \lvert \E_0\left\{g(\mathbb B_n) -g(\bar{\mathbb B}_n)\right\}\rvert + 
		\lvert \E_0\left\{g(\bar{\mathbb B}_n)\right\} -\E_0\left\{g(\mathbb B_0)\right\}\rvert \,
	\end{align*}
	will establish uniform convergence if we can show that $\lvert \E_0\left\{g(\mathbb B_n) -g(\bar{\mathbb B}_n)\right\}\rvert $ tends to zero. We note that condition~\ref{ass:phi_donsker}) and the definition of the function $g$ ensures that 
	\begin{align*}
		\lvert \E_0\left\{g(\mathbb B_n) -g(\bar{\mathbb B}_n)\right\}\rvert
		\le &\, \E_0\left\{\min (L_0 n^{\frac{1}{2}}\lVert r_n\rVert_{\infty}, 2M_0)\right\}\, 
	\end{align*}
	and it straightforward to establish that $\E_0\left\{\min (L_0 n^{\frac{1}{2}}\lVert r_n\rVert_{\infty}, 2M_0)\right\} \rightarrow 0$. Because $\E_0\left\{g(\mathbb B_n)\right\}$ converges in expectation to $\E_0\left\{g(\mathbb B_0)\right\}$ for any bounded and Lipschitz function $g$, we conclude that $\mathbb B_n$ converges weakly to $\mathbb B_0$ relative to the supremum norm.

	Following this result, Condition~\ref{ass:sup_surv} is a modification of Condition~\ref{ass:pro_conv} that ensure the upper bound holds uniformly. We define $\mathbb B_n^*::=\{n^{\frac{1}{2}} \{\psi_n(s) - \psi_0(s)\}/\sigma_0(s): s\in (0,\tau)\}$ and seek to accomplish the task of establishing that $\mathbb B_n^*$ converges weakly to a mean-zero standard process $(\mathbb B_0^*)$ relative to the supremum norm. We define $\sigma_L := \inf_{s\in(0,\tau)}  \sigma_0(s)$, which is strictly positive. By the definition of a Lipschitz function, the following bound holds
	\begin{align*}
		\lvert g(\mathbb B_n^*) - g(\mathbb B_0^*)\rvert
		\le &\, \frac{L_0}{\sigma_L}\lVert \mathbb B_n -\mathbb B_0\rVert_{\infty}\,
	\end{align*}
	which implies that $\mathbb B_n^*$ converges to a standardized Gaussian process defined by $\{\mathbb B_0(s)/\sigma_0(s) : s\in (0,\tau)\}$ when $\mathbb B_n$ converges to $\mathbb B_0$. The above result therefore implies that $\mathbb B_n^*$ converges weakly to a standard mean-zero Gaussian process relative to the supremum norm.

	\subsection*{Proof of Theorem 4}
	For each $(t_0,z_0)\in\mathbb{R}\times\mathcal{Z}$, we denote by $R_n(t_0,z_0):=\tilde{\mathbb{F}}_n(t_0,z_0)-\tilde{\mathbb{F}}_0(t_0,z_0)-\frac{1}{n}\sum_{i=1}^{n}\phi_{t_0,z_0,0}(O_i)$ the remainder from the linearization of $\tilde{\mathbb{F}}_n(t_0,z_0)$ as estimator of $\tilde{\mathbb{F}}_0(t_0,z_0)$. Under the Conditions of Theorem 3, it holds that $\sup_{t_0,z_0}|R_n(t_0,z_0)|=o_P(n^{-\frac{1}{2}})$. We can write that
	\begin{align*}
		\Theta(\mathbb F_n) - \Theta(\mathbb F_0)\ &=\ 
		\partial\Theta\left(\mathbb F_0; \mathbb F_n - \mathbb F_0\right) + o_P(n^{-\frac{1}{2}})\\
		&=\ \partial\Theta\Big{(}\mathbb F_0; (t_0,z_0)\mapsto\frac{1}{n}\sum_{i=1}^n \phi_{t_0,z_0,0}(O_i) + R_n(t,z)\Big{)} + o_P(n^{-\frac{1}{2}})\\
		&=\ \frac{1}{n}\sum_{i=1}^n \partial\Theta\left(\mathbb F_0;(t_0,z_0)\mapsto \phi_{t_0,z_0,0}(O_i) \right) +  \partial\Theta\left(\mathbb{F}_0;(t_0,z_0)\mapsto R_{n}(t_0,z_0)\right) + o_P(n^{-\frac{1}{2}})\\
		&=\ \frac{1}{n}\sum_{i=1}^n \partial\Theta\left(\mathbb F_0;(t_0,z_0)\mapsto \phi_{t_0,z_0,0}(O_i) \right) +  o_P(n^{-\frac{1}{2}})\ ,
	\end{align*}as claimed. The first equality follows from Hadamard differentiability of $\Theta$ at $\tilde{\mathbb{F}}_0$ and Theorem 20.8 of~\cite{van2000asymptotic}. The asymptotic linearity of $\tilde{\mathbb{F}}_n$ is used in the second equality. The third equality follows from the linearity of $h\mapsto \partial\Theta(\tilde{\mathbb{F}}_0;h)$. Finally, the fourth equality follows from the boundedness of $h\mapsto \partial\Theta(\tilde{\mathbb{F}}_0;h)$, which implies that \[|\partial\Theta\left(\mathbb{F}_0;(t_0,z_0)\mapsto R_{n}(t_0,z_0)\right)|=O_P\left(\sup_{t_0,z_0}|R_n(t_0,z_0)|\right)=o_P(n^{-\frac{1}{2}})\ .\]
	
	\section*{Part B: identification}\label{partb}
	
	Before establishing the identification of $\upsilon_0$, we show that components of the ideal data distribution $P_{X,0}$ can be identified in terms of the observed data distribution $P_0$.
	
	Under conditions~\ref{ass:support}--\ref{ass:positivity}, in view of Theorem 11 of~\cite{gill1990survey}, for any $t\in (0,\tau_2(z))$, we can write
	\begin{equation}
		1-F_{X,0}(t\miid a, z)\ =\ S_{X,0}(t\miid a, z)\ =\ \prodi_{u\in [0,
			t)}\left\{1-\Lambda_{X,0}(du\miid a,z)\right\} \label{prodint}
	\end{equation}
	for $\Lambda_{X,0}(t\miid a,z) := \int_0^t\frac{F_{X,0}(du\miid a,z)}{S_{X,0}(u\miid a,z)}$. We first express the observed follow-up time subdistribution function in terms of $P_{X,0}$. Defining $F_0(t\,|\,w,a,z):=P_{X,0}\left(T\leq t\,|\,W=w,A=a,Z=z,T\geq W\right)$ and $G_0(w\,|\,a,z):=P_{X,0}\left(W\leq w\,|\,A=a,Z=z,T\geq W\right)$, we note that \[F_0(dt\,|\,w,a,z)G_0(dw\,|\,a,z)= \frac{ I_{[0,t]}(w)F_{X,0}(dt\,|\,a,z)G_{X,0}(dw\,|\,a,z)}{\int G_{X,0}(u\,|\,a,z)F_{X,0}(du\,|\,a,z)}\]under conditions~\ref{ass:truncation}--\ref{ass:censoring}. So, we have that
	\begin{align*}
		&F_{1,0}(u\miid a,z)\ =\ P_{0}\left(Y\le u, \Delta =1 \miid A=a, Z=z\right)\\
		&\hspace{.15in}=\ P_{X,0}\left(T\le u, T\le C \miid A=a, Z=z, T\ge W\right)\\
		&\hspace{.15in}=\ \iint I_{[0,u]}(t)P_{X,0}\left(C\geq t\,|\,T=t,W=w,A=a,Z=z,T\geq W\right)F_0(dt\,|\,w,a,z)G_0(dw\,|\,a,z)\\
		&\hspace{.15in}=\ \iint I_{[0,u]}(t)Q_0(t\,|\,w,a,z)F_0(dt\,|\,w,a,z)G_0(dw\,|\,a,z)\\
		&\hspace{.15in}=\ \frac{\iint I_{[w,u]}(t)Q_0(t\,|\,w,a,z)F_{X,0}(dt\,|\,a,z)G_{X,0}(dw\,|\,a,z)}{\int G_{X,0}(t\,|\,a,z)F_{X,0}(dt\,|\,a,z)}\\
		&\hspace{.15in}=\ \frac{\int I_{[0,u]}(t)\left\{\int I_{[0,t]}(w) Q_0(t\,|\,w,a,z)G_{X,0}(dw\,|\,a,z)\right\}F_{X,0}(dt\,|\,a,z)}{\int G_{X,0}(t\,|\,a,z)F_{X,0}(dt\,|\,a,z)}\ ,
	\end{align*}which implies that \[F_{1,0}(du\,|\,a,z)=\frac{\int I_{[0,u]}(w)Q_0(u\,|\,w,a,z)G_{X,0}(dw\,|\,a,z)}{\int G_{X,0}(t\,|\,a,z)F_{X,0}(dt\,|\,a,z)}F_{X,0}(du\,|\,a,z)\ .\]
	Next, we can write that 
	\begin{align*}
		&R_0(u\miid a, z)\ =\ P_0(W\le u\le Y\miid A=a, Z=z)\\
		&\hspace{.3in}=\ \int I_{[0,u]}(w)P_{X,0}\left(T\geq u,C\geq u\miid W=w,A=a,Z=z,T\geq W\right)G_0(dw\miid a,z)\\
		&\hspace{.3in}=\ \iint I_{[w,t]}(u)Q_0(u\miid w,a,z)F_0(dt\,|\,w,a,z)G_0(dw\miid a,z)\\
		&\hspace{.3in}=\ \frac{\iint I_{[w,t]}(u)Q_0(u\miid w,a,z) I_{[0,t]}(w)F_{X,0}(dt\miid a,z)G_{X,0}(dw\miid a,z)}{\int G_{X,0}(u\miid a,z)F_{X,0}(du\miid a,z)}\\
		&\hspace{.3in}=\ \frac{\int I_{[0,u]}(w)Q_0(u\,|\,w,a,z)G_{X,0}(dw\,|\,a,z)}{\int G_{X,0}(t\,|\,a,z)F_{X,0}(dt\,|\,a,z)}S_{X,0}(u\miid a,z)\ .
	\end{align*}
	Thus, under conditions~\ref{ass:truncation}--\ref{ass:censoring}, we find that
	\[\tilde\Lambda_0(t\miid a,z)=\int_0^t \frac{F_{0,1}(du\miid a,z)}{R_0(u\miid a,z)}=\int_0^t\frac{F_{X,0}(du\miid a,z)}{S_{X,0}(u\miid a,z)} = \Lambda_{X,0}(t\miid a,z)\ ,\]and so, in view of \eqref{prodint},  $F_{X,0}(t\miid a,z)$ is identified by $\tilde{F}_0(t\miid a,z):=1-\prodi_{u\in(0,t)}\{1-\tilde{\Lambda}(t\miid a,z)\}$. The fact that $F_{X,0}$ is identified directly implies that the target conditional truncation distribution $G_{X,0}$ is itself identified in view of the fact that $G_{X,0}(du\miid a,z) \propto_{a,z}  S_{X,0}(u\miid a,z)^{-1}G_0(du\miid a,z)$, which allows us to write that 
	\begin{align*}
		G_{X,0}(w\miid a,z)=\frac{1}{\gamma_0(a,z)}\int_0^w S_{X,0}(u\miid a,z)^{-1}G_0(du\miid a,z)
	\end{align*} 
	with $\gamma_0(a,z):=\int S_{X,0}(u\miid a,z)^{-1}G_0(du\miid a,z)$ the appropriate normalizing constant. Similarly, identification of $F_{X,0}$ implies identification of the target joint exposure-covariate distribution function $J_{X,0}$ in view of the fact that \[J_{X,0}(da,dz)\propto_{a,z}\frac{J_0(da,dz)}{\int S_{X,0}(u\miid a,z)G_{X,0}(du,a,z)}=\gamma_0(a,z)J_0(da,dz)\]so that $J_{X,0}(da,dz)=\bar{\gamma}_0(a,z)J_0(da,dz)$ with $\bar{\gamma}_0(a,z):=\gamma_0(a,z)/\iint \gamma_0(a,z)J_0(da,dz)$. Of course, the target marginal covariate distribution $H_{X,0}$ is then itself identified from the identification for $J_{X,0}$ through marginalization. Since $\theta_0$ is a functional of $F_{X,0}$ and $H_{X,0}$, the identification of the latter distributions directly implies that of $\theta_0$.

	The observable conditional censoring survival function $Q_{0}$ can also be identified using product-integration as in the well-known context of right censoring without truncation. Defining with some abuse of notation the subdistribution function $F_{0,0}(u\miid w,a,z):=P_0\left(T\leq u,\Delta=0\miid A=a,Z=z\right)$ and at-risk probability function $R_0(u\miid w,a,z):=P_0\left(W\leq u\leq Y\miid W=w,A=a,Z=z\right)$, we have that \begin{align*}
		F_{0,0}(u\miid w,a,z)\ &=\ -\int I(c\leq u)P_0\left(T\geq c\miid W=w,A=a,Z=z,C=c\right)Q_0(dc\miid w,a,z)\\
		&=\ -\int I(c\leq u)P_0\left(T\geq c\miid W=w,A=a,Z=z\right)Q_0(dc\miid w,a,z)\ ,
	\end{align*}which implies that $F_{0,0}(du\miid w,a,z)=-P_0\left(T\geq u\miid W=w,A=a,Z=z\right)Q_0(du\miid w,a,z)$, and furthermore, for $u\geq w$, we have that \begin{align*}
		R_{0}(u\miid w,a,z)\ &=\ P_0\left(Y\geq u\miid W=w,A=a,Z=z\right)\\
		&=\ P_0\left(C\geq u\miid W=w,A=a,Z=z\right)P_0\left(T\geq u\miid W=w,A=a,Z=z\right).
	\end{align*}Thus, we find that $F_{0,0}(du\miid w,a,z)/R_0(u\miid w,a,z)=-Q_0(du\miid w,a,z)/Q_0(u\miid w,a,z)$, which then implies that $Q_0(c\miid w,a,z)$ can be identified by the product-integral \[\Prodi_{c\in[0,u)}\left\{1-\frac{F_{0,0}(du\miid w,a,z)}{R_0(u\miid w,a,z)}\right\}.\]

	\section*{Part C: reparametrization}\label{partc}
	
	We establish the identification, variation independence and reparametrization of the observed data distribution $P_0$ in terms of the chosen nuisance parameters $\eta(P) := (F_{X,0}, Q_{0}, G_0, J_0)$. We first write
	\begin{align*}
		P_0(do) =&\, P_0(dy, d\delta\miid w,a,z)G_0(dw\miid a,z)J_0(da,dz)\,,
	\end{align*}
	which leaves the conditional distribution $P_0(dy, d\delta\miid w,a,z)$ left to identify in terms of the distributions $T\miid A,Z$ and $C\miid W,A,Z,T\ge W$. This is done by noting the following equality
	\begin{align*}
		P_0(Y\le y, \Delta = 1\miid w,a,z) 
		=&\, \int P_0(T\le c, T\le y\miid C=c, W=w, A=a, Z=z)P_0(dc\miid w,a,z)\\
		(A1)=&\, \int P_0(T\le c\wedge y\miid  W=w, A=a, Z=z)P_0(dc\miid w,a,z)\\
		(A2)=&\, \int \frac{P_{X,0}(w\le T\le c\wedge y\miid  A=a, Z=z)}{S_{X,0}(w\miid a,z)}P_0(dc\miid w,a,z)\, ,
	\end{align*}
	which gives the desired result for $\delta =1$. Similarly for $\delta = 0$, we have
	\begin{align*}
		P_0(Y\le y, \Delta = 0\miid w,a,z) 
		=&\, \int P_0(C\le t, C\le y\miid  W=w, A=a, Z=z)F_0(dt\miid w,a,z)\\
		=&\, \int [1-Q_0(c\wedge y\miid w,a,z)]\frac{F_{X,0}(dt\miid a,z)}{S_{X,0}(w\miid a,z)}
	\end{align*}
	Using the results from identification we can write the observed at risk probability as 
	\begin{align*}
		\tilde R_0(u\miid a,z) := S_{X,0}(u\miid a,z)\int_0^uQ_0(u\miid w,a,z)\frac{G_0(dw\miid a,z)}{S_{X,0}(w\miid a,z)}\, ,
	\end{align*}
	which together with the previous results gives a representation of the parameter and influence function as functionals of $\eta(P)$. 
	
\end{document}